\documentclass[a4paper, amsfonts, amssymb, amsmath, reprint, showkeys, twoside,superscriptaddress, onecolumn, nofootinbib]{revtex4-2}
\usepackage{amsmath,amsthm}
\usepackage{amssymb}
\usepackage{array}
\usepackage[export]{adjustbox}
\usepackage{rotating}
\usepackage{mathtools}
\usepackage{mathdots}
\usepackage{bm}
\usepackage{dsfont}
\usepackage{xcolor}
\usepackage{algorithm2e}
\usepackage{graphicx}

\usepackage{thmtools}
\usepackage{thm-restate}

\usepackage{mathrsfs}
\usepackage{nccmath}
\usepackage{physics}
\usepackage{comment}
\usepackage{ulem}
\usepackage[colorlinks=true, urlcolor=blue, linkcolor=blue, citecolor=blue]{hyperref}
\usepackage[capitalize]{cleveref}
\usepackage[acronym]{glossaries}
\setkeys{glslink}{hyper=false}
\usepackage{makecell}

\usepackage{tikz}
\usetikzlibrary{quantikz2}

\usepackage[compat=0.4]{yquant}
\useyquantlanguage{groups}

\newtheorem{theorem}{Theorem}
\newtheorem{lemma}[theorem]{Lemma}

\newtheorem{proposition}[theorem]{Proposition}

\usepackage{todonotes}

\newcommand{\il}[1]{\textcolor{orange}{[IL: #1]}}

\newcommand{\bea}{\begin{eqnarray}}
\newcommand{\eea}{\end{eqnarray}}

\newcommand{\qft}{\textrm{QFT}}

\newacronym{qft}{QFT}{quantum Fourier transform}
\newacronym{qsp}{QSP}{quantum signal processing}
\newacronym{com}{COM}{centre of mass}
\newacronym[
  plural={DOFs},
  longplural={degrees of freedom}
]{dof}{DOF}{degree of freedom}
\newacronym{lcu}{LCU}{linear combination of unitaries}

\newcommand{\ceil}[1]{\left\lceil #1 \right\rceil}
\newcommand{\floor}[1]{\left\lfloor #1 \right\rfloor}
\newcommand{\prep}{$\textup{PREP}$ }
\newcommand{\sel}{$\textup{SEL}$ }

\definecolor{X1}{HTML}{4D53C8}
\definecolor{X2}{HTML}{52BD7C}
\definecolor{X3}{HTML}{44ACE8}
\definecolor{X4}{HTML}{D7333B}
\definecolor{X5}{HTML}{FFBE0C}
\definecolor{X6}{HTML}{9467BD}

\let\oldaddcontentsline\addcontentsline
\newcommand{\stoptocentries}{\renewcommand{\addcontentsline}[3]{}}
\newcommand{\starttocentries}{\let\addcontentsline\oldaddcontentsline}

\begin{document}

\title{Efficient Simulation of Pre-Born-Oppenheimer Dynamics on a Quantum Computer}

\author{Matthew Pocrnic}
\affiliation{Xanadu, Toronto, ON, M5G 2C8, Canada}
\affiliation{Department of Physics, University of Toronto, Toronto, ON, M5S 1A7, Canada}
\author{Ignacio Loaiza}
\affiliation{Xanadu, Toronto, ON, M5G 2C8, Canada}
\author{Juan Miguel Arrazola}
\affiliation{Xanadu, Toronto, ON, M5G 2C8, Canada}
\author{Nathan Wiebe}
\affiliation{Department of Computer Science, University of Toronto, Toronto, ON, M5S 2E4, Canada}
\affiliation{Pacific Northwest National Laboratory, Richland, WA, 99352, USA}
\affiliation{Canadian Institute for Advanced Research, Toronto, ON, M5G 1M1, Canada}
\author{Danial Motlagh}
\affiliation{Xanadu, Toronto, ON, M5G 2C8, Canada}

\begin{abstract}
In this work, we present a quantum algorithm for direct first-principles simulation of electron-nuclear dynamics on a first-quantized real-space grid. Our algorithm achieves best-in-class efficiency for block-encoding the pre-Born-Oppenheimer molecular Hamiltonian by harnessing the linear scaling of swap networks for implementing the quadratic number of particle interactions, while using a novel alternating sign implementation of the Coulomb interaction that exploits highly optimized arithmetic routines. We benchmark our approach for a series of scientifically and industrially relevant chemical reactions. We demonstrate over an order-of-magnitude reduction in costs compared to previous state-of-the-art for the $\rm NH_3+BF_3$ reaction, achieving a Toffoli cost of  $8.7\times10^{9}$ per femtosecond using $1362$ logical qubits (system + ancillas). Our results significantly lower the resources required for fault-tolerant simulations of photochemical reactions, while providing a suite of algorithmic primitives that are expected to serve as foundational building blocks for a broader class of quantum algorithms.
\end{abstract}
\maketitle

\stoptocentries
\section{Introduction} \label{sec:intro}

As hardware developments bring fault-tolerant quantum computing closer to reality, the need for expanding the pool of useful applications of quantum computers intensifies. Quantum chemistry has long been viewed as a natural domain for such applications, offering a rich landscape of industrially relevant problems whose classical scaling is prohibitive. However, algorithmic development has overwhelmingly focused on solving the electronic structure problem within the Born-Oppenheimer approximation, with little attention given to development of quantum algorithms for dynamical problems.\\

Although reaction rates and other dynamical observables can sometimes be inferred indirectly from electronic structure calculations through models such as transition state theory, their accuracy fundamentally relies on the validity of the Born–Oppenheimer approximation within the given physical context. Cases where this approximation breaks down are pervasive in chemistry \cite{nbo_1,nbo_2,nbo_3,nbo_4,nbo_5,nbo_6,nbo_7}, particularly in photochemistry and reactions involving highly reactive radical species, where multiple electronic states are closely coupled. In these scenarios, non-adiabatic effects become dominant, rendering transition state theory and other methods based on a single potential energy surface insufficient; this necessitates going beyond the Born-Oppenheimer approximation for accurately obtaining quantities like reaction rates. Despite the existence of heuristics for including non-adiabatic effects on transition state calculations, their accurate treatment requires information from a full non-adiabatic dynamical treatment \cite{tst_1,tst_2}. \\

Simulating pre-Born-Oppenheimer chemistry classically is a formidable challenge with a long history of developments, resulting in a wide array of approximations that typically trade accuracy for computational efficiency \cite{nbo_3,nbo_5,nbo_6,nbo_7,dyn_1,dyn_2,dyn_3}. Despite the major successes of many methods, no single, efficient classical method exists that is applicable to general systems; the varying degrees of approximations translate into a mishandling of potentially critical quantum phenomena, such as electronic decoherence \cite{subotnik2016understanding}, zero-point leakage \cite{zpl_1,zpl_2}, and non-adiabatic branching ratios \cite{ci_1,ci_2}, while also relying on the quality of the underlying electronic structure calculations. In contrast, exact classical methods that propagate the full molecular wavefunction suffer from the curse of dimensionality, exhibiting exponential cost scaling with the number of particles \cite{nbo_4,nbo_6,lubich2008quantum}. For an arbitrary reaction without prior knowledge of which effects are dominant, a method accounting for all effects is required. Such comprehensive classical methods are computationally intractable for all but the smallest systems. This establishes the simulation of non-adiabatic dynamics as a high-value frontier where quantum computers could greatly impact current computational workflows. \\

\begin{figure}[]
    \centering
    \includegraphics[width=\linewidth]{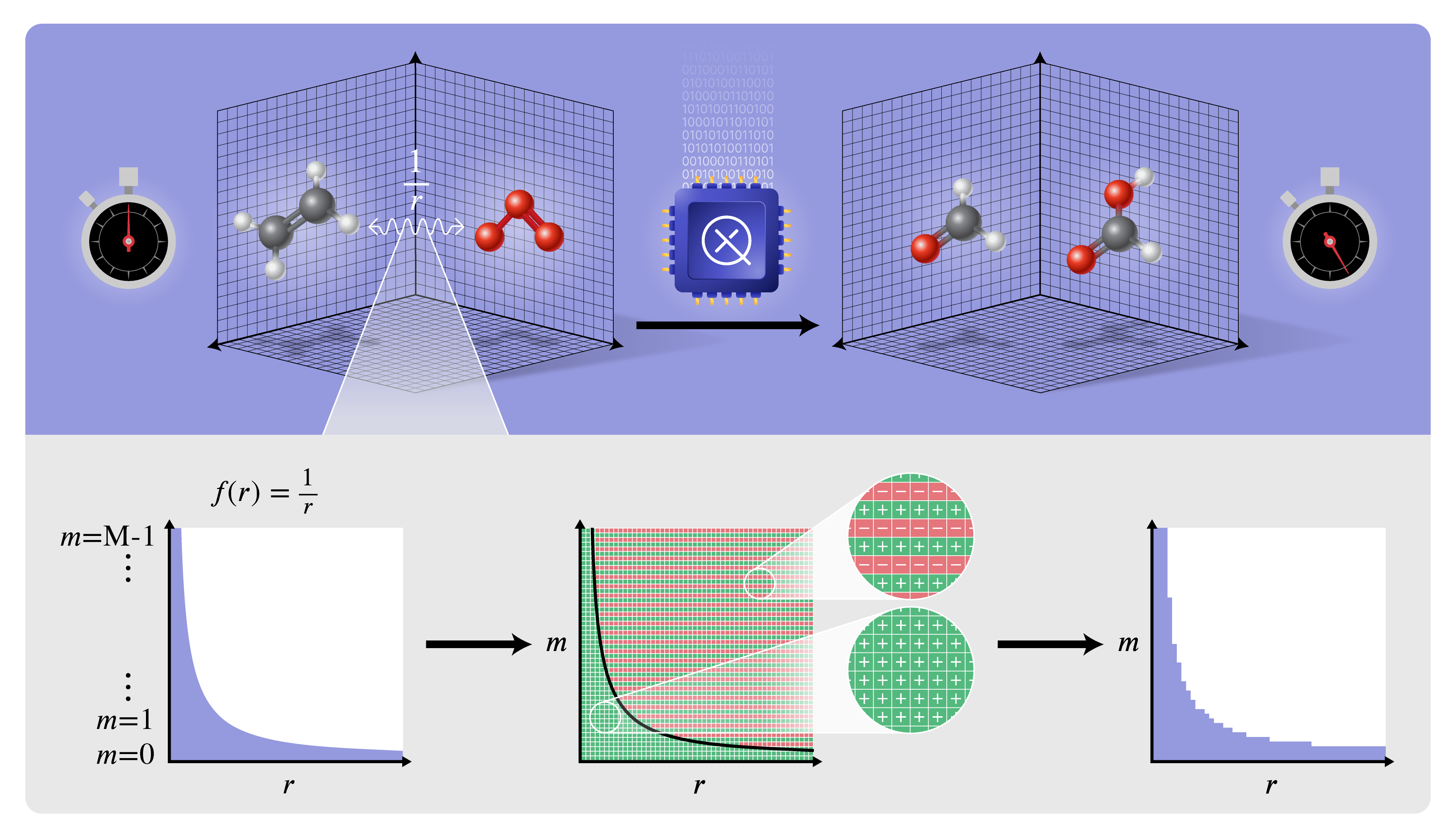}
    \caption{Quantum simulation of chemical reactions on a real-space grid. The Coulomb interaction $1/r$ is implemented as a linear combination of unitaries using the alternating sign technique: an auxiliary register $\ket{m}$ is prepared in an equal superposition over $M$ values and flagged with $\pm1$ phases depending on the value of $r$, such that the resulting sum approximates $1/r$ up to $\mathcal{O}(1/M)$ accuracy. The flagging routine is efficiently implemented using quantum arithmetic.}
    \label{fig:hero}
\end{figure}

Quantum algorithms that go beyond the Born-Oppenheimer approximation have been proposed \cite{motlagh2025quantum,jornada_comprehensive_2025,ollitrault2020nonadiabatic}. One avenue is the simulation of vibronic Hamiltonians, for which efficient low-cost quantum algorithms have been developed \cite{motlagh2025quantum}. However, the construction of vibronic Hamiltonians requires prior knowledge of the relevant chemical space, while also relying on complex fitting procedures based on underlying electronic structure calculations \cite{vibronic_1,vibronic_2}. As such, the quality of this approach hinges on the quality of the electronic structure calculations and of the polynomial fit of the associated potential energy surfaces. Photochemical processes which typically exhibit large amplitude motion of nuclei require construction of highly anharmonic surfaces, with a steep increase in the required number of electronic structure calculations for vibronic model construction.\\

In contrast, first-principles simulations of molecular dynamics bypass the need for electronic structure calculations when building a model Hamiltonian. Notably, Ref.~\cite{jornada_comprehensive_2025} presented, to the best of our knowledge, the first full resource estimates for such simulations, using a plane-wave basis for both electronic and nuclear degrees of freedom in a first-quantized qubitization framework. While attractive for periodic systems like surfaces, this approach has two major drawbacks: the resource requirements remain outside the reach of early fault-tolerant quantum computers, and the periodic basis is ill-suited for non-periodic chemistry, especially for systems with a non-neutral electric charge, since the periodic boundary conditions cause a net infinite charge. \\

In this work, we present a first-principles algorithm for simulating chemical dynamics on a first-quantized real-space grid, treating electrons and nuclei on equal footing. The amplitudes of the many-body wavefunction are encoded in a Cartesian grid, while the electrostatic Coulomb interaction is directly implemented in Cartesian coordinates. This approach is ideal for organic and photo-organic systems, avoiding the complexities of plane-wave bases. One of the main technical contributions of this work is a ``swap network'' block-encoding architecture, which generalizes the ``swap up'' technique from Ref.~\cite{jornada_comprehensive_2025}. This framework allows us to recursively build block-encodings of multi-particle operators with a linear complexity in the number of particles $\eta$. Specifically, we evaluate all $\mathcal O(\eta^2)$ pairwise electrostatic interactions with only $\mathcal O(\eta)$ cost and reduce the required \gls{qft} applications for kinetic terms from $6\eta$ to only $2$. 

To address the $1/r$ Coulomb bottleneck, we implement an efficient \gls{lcu}-based approach utilizing an alternating sign technique for block-encoding diagonal operators, as illustrated in \cref{fig:hero}. By combining this with quantum-arithmetic-based comparators, we bypass the need for costly evaluation of inverse functions. Furthermore, we introduce two techniques to minimize the 1-norm: a spectral shifting technique that achieves a twofold reduction without altering the dynamics, and a saturation technique that exploits statistical bounds for the minimum nuclear-nuclear distance. These combined innovations achieve over an order-of-magnitude reduction in Toffoli counts for the $\rm NH_3+BF_3$ reaction compared to previous benchmarks \cite{jornada_comprehensive_2025}, while maintaining a small ancilla footprint. Overall, our results bring the first-principles simulation of organic chemistry within closer reach of the first generations of fault-tolerant quantum computers. \\

This paper is organized as follows. In \cref{sec:preliminaries} we give an overview of the molecular Coulomb Hamiltonian and the real-space encoding of the associated wavefunction, alongside a high-level summary of our approach for qubitizing this Hamiltonian. We then present our qubitization-based quantum algorithm in \cref{sec:algorithm}. Resource estimates for a series of reactions of industrial relevance are presented in \cref{sec:discussion}, alongside a discussion of our results. Finally, we present our conclusions in \cref{sec:conclusion}.

\glsresetall

\begin{table}[]
    \centering
    \begin{tabular}{c|c|c|c|c|c}
       \textbf{Reactants} & \textbf{Number of particles (electrons)} & \textbf{Toffolis per fs} & \textbf{Toffolis per }$\mathcal{W}$ & \textbf{Logical qubits} & \textbf{1-norm} \\ \hline
       $\rm NH_3+BF_3$ & $50\ (42)$ & $8.72\times 10^{9}$& $8.2 \times 10^3$ & $1362 \ (312)$& $1.88 \times 10^4$ \\
       $\rm 2NO_2$ & $52\ (46)$ & $1.05\times 10^{10}$ & $8.6 \times 10^3$ & $1419 \ (327)$ & $2.15 \times 10^4$ \\
       $\rm C_2H_4+O_2$ & $40\ (32)$ & $7.07\times 10^{10}$ & $8.5 \times 10^3$ & $1453 \ (373)$ & $1.46 \times 10^5$ \\
       $\rm C_2H_4+O_3$ & $49\ (40)$ & $8.05\times 10^{9}$ & $8.1 \times 10^3$ & $1341 \ (312)$ & $1.76 \times 10^4$ \\
       $\rm C_{23}H_{20}N_3O$ & $234\ (187)$ & $2.73\times 10^{11}$& $2.2 \times 10^4$& $6198 \ (582)$ & $2.16 \times 10^5$\\
    \end{tabular}
    \caption{Cost estimates for performing $1\ \rm fs$ of time evolution for various reactions using our algorithm, as detailed in \cref{sec:discussion}. Logical qubits are reported in the format $x\,(y)$, where $x$ is the total number of qubits (system + ancilla) with $y$ being the number of ancillas, and $x-y$ the number qubits required for encoding the problem. Due to the linear scaling of our algorithm in evolution time, the associated Toffoli cost for $t$ femtoseconds of evolution can be estimated by multiplying the numbers presented by $t$.}
    \label{tab:estimates}
\end{table}

\section{Preliminaries} \label{sec:preliminaries}
We start by introducing the pre-Born-Oppenheimer molecular Hamiltonian. Given a system with $\eta_{e}$ electrons and $\eta_{n}$ nuclei, the full non-relativistic Coulomb Hamiltonian in position basis written in atomic units corresponds to
\begin{equation} \label{eq:full_ham_original}
     H = \frac{1}{2}\left( -\sum_{l =1}^{\eta_{e}} \nabla^2_l - \sum_{\ell =1}^{\eta_{n}} \frac{\nabla^2_\ell}{m_\ell}\right) +\sum_{l\neq k}^{\eta_{e}} \frac{1}{2\|\vec{r_l} - \vec{r_k}\|} +\sum_{\ell\neq \kappa}^{\eta_{n}} \frac{\zeta_\ell\zeta_\kappa}{2\|\vec{R_\ell} - \vec{R_\kappa}\|} - \sum_{l=1}^{\eta_{e}} \sum_{\ell=1}^{\eta_{n}} \frac{\zeta_\ell}{\|\vec{r_l} - \vec{R_\ell}\|}.
\end{equation}
Since both electrons and nuclei are treated on the same footing, we write the Hamiltonian more compactly
\begin{equation} \label{eq:full_ham}
     H = \underbrace{\sum_{i=1}^{\eta} \frac{-\nabla^2_i}{2m_i}}_{ T} + \underbrace{\sum_{i\neq j} (-1)^{\sigma_i+\sigma_j}\frac{\zeta_i\zeta_j}{2r_{ij}}}_{ V},
\end{equation}
where $i$ runs over all $\eta=\eta_{e}+\eta_{n}$ particles with $r_{ij} = \|\vec r_i - \vec r_j\|$. Here $m_i$ and $\zeta_i$ correspond to particle $i$'s mass and magnitude of its charge correspondingly, and $\sigma_i$ encodes the sign of the particle's charge
\begin{equation} \label{eq:s_i}
    \sigma_i = \begin{cases}
        1, \textup{ if particle } i \textup{ is an electron } \\
        0, \textup{ if particle } i \textup{ is a nucleus}.
    \end{cases}
\end{equation}
Considering a cubic simulation box of volume $\Omega=L^3$, where each dimension has a width of $L$, we represent our wavefunction using a first-quantized real-space grid. We allocate 3 registers of size $n_g$ for each particle, where the position of the particle along the $x$, $y$, and $z$ axes are encoded in computational basis states of each register using a two's complement encoding. Using $n_g$ qubits to discretize each Cartesian coordinate induces the 3-dimensional grid 
\begin{equation}
    G = [-2^{n_g-1}, 2^{n_g-1}-1]^3
\end{equation}
with spacing
\begin{equation} \label{eq:grid_delta}
    \Delta = \frac{L}{2^{n_g}-1}.
\end{equation}
The position of each particle is encoded by computational basis states $\ket{\vec q}=\ket{q_x}\ket{q_y}\ket{q_z}\in G$, giving the physical position and momentum of the particle via the transformation
\begin{equation}
    r_q = q \Delta, \; \; k_p=\frac{2\pi p}{\Delta},
\end{equation}
where position and momentum are related via the \gls{qft} as $\ket p = \qft \ket q$. The associated molecular wavefunction will then be represented as
\begin{equation}
    \ket{\Psi} = \sum_{\vec q_1, \vec q_2,..., \vec q_\eta}\Psi(\vec r_{q_1},\cdots,\vec r_{q_\eta}) \bigotimes_{i=1}^\eta \ket{\vec q_{i}}.
\end{equation}
Since the Coulomb interaction from \cref{eq:full_ham} could become unbounded if $r_{ij}=0$, we define the discretized potential as
\begin{equation}\label{eq:V}
     V = \sum_{i \neq j} (-1)^{\sigma_i+\sigma_j} \zeta_i\zeta_j\sum_{\vec q,\vec s} \frac{S(\|\vec r_q-\vec r_s\|) }{2} \ketbra{\vec q}{\vec q}_i \otimes \ketbra{\vec{s}}{\vec{s}}_j, 
\end{equation}
where
\begin{equation} \label{eq:sat_inv}
    S(r) = \begin{cases}
        \frac{1}{r},&\textrm{ if } r > \Delta \\[8pt]
        \frac{1}{\Delta},& \textrm{ if } r \leq \Delta.
    \end{cases}
\end{equation}
To make this a faithful representation, $n_g$ must be chosen such that the physical wavefunction for the problem of interest has negligible overlap with configurations where two particles come closer than $\Delta$. A discussion on how to choose the number of qubits $n_g$ such that the induced grid spacing faithfully represents the molecular wavefunction is provided in Appendix~\ref{app:grid_spacing}. Lastly, we define the momentum basis for each particle to be related to the position basis via a \gls{qft} of each of the Cartesian registers with associated discretized kinetic operator defined as
\begin{equation} \label{eq:T}
     T = \frac{2\pi^2}{L^2} \sum_i \sum_{\vec q} \frac{\|\vec q\|^2}{m_i} \textup{QFT}_i^\dagger \cdot \ketbra{\vec q}{\vec q}_i \cdot \textup{QFT}_i.
\end{equation}

\section{Quantum algorithm} \label{sec:algorithm}
In this section, we describe our \gls{qsp}-based simulation algorithm for implementing the time evolution $e^{-i H t}$ under the molecular Coulomb Hamiltonian \cref{eq:full_ham}. The main contribution of this work is the development of a highly optimized block-encoding protocol for $ H$ that scales linearly with the number of particles in the system despite the quadratic number of interactions in the Hamiltonian. Our key results are summarized in the following theorems that state the time and space complexity of our algorithm for block-encoding and performing time evolution under $H$.

\begin{restatable}[Block-encoding with linear scaling]{theorem}{theobe} \label{theo:ham}
Let $ H= T+ V\in \mathcal L(\mathbb{C}^{2^{3\eta\cdot n_g}})$ be the discretized Hamiltonian shown in Eqs.(\ref{eq:V},\ref{eq:T}). There exists an $(\alpha_H,n_H,\epsilon)$ block-encoding of $H$ (up to a constant shift) that can be constructed using 
\begin{equation}
    \tilde{\mathcal{O}}\left(\eta +\log^2\frac{1}{\epsilon}\right)
\end{equation}    
Toffolis, and where we have used the following definitions for the 1-norm and number of block-encoding qubits respectively: 
    \begin{align}
        \alpha_H &= \frac{\lambda_{V}}{4\Delta} + \frac{3\pi^2 2^{2(n_g-1)}}{L^2}\lambda_T \in \mathcal{O}(\eta^2) ,\label{eq:alpha_ham}  \\
        n_H &= 2\ceil{\log_2\eta} + n_M + 3 \in \mathcal{O}\left(\log\eta+\log\frac{1}{\epsilon}\right),
    \end{align}
    and $\epsilon$ corresponds to the overall error of this operation. Here $n_M\in\mathcal{O}(\log\eta/\epsilon)$ is the number of ancillas for implementing the Coulomb interaction with the target accuracy. We simplify notation by defining the sum of charge interactions $\lambda_{V} := \sum_{j\neq k =1}^\eta \zeta_j \zeta_k\in\mathcal{O}(\eta^2)$ and the sum of reciprocal masses $\lambda_T:= \sum_j  m_j^{-1}\in\mathcal{O}(\eta)$. 
\end{restatable}
Once the qubitized walk operator is constructed from the block-encoding of $ H/\alpha_H$ we can use the technique of \gls{qsp} \cite{low2017optimal, low2019hamiltonian, berry2024doubling, motlagh2024generalized} to implement the time-evolution operator using $\mathcal{O}(\alpha_H t + \log(1/\epsilon))$ calls to the walk operator. This gives us the following complexity for implementation of the $e^{-iHt}$ operator using our algorithm.

\begin{restatable}[Time evolution]{theorem}{theodyn} \label{theo:time_evolution}
Given a Hamiltonian $H$ of the form defined in Eq.~\eqref{eq:full_ham}, a simulation time $t$, and error $\epsilon$, there exists a quantum algorithm to implement the time-evolution operator $e^{-iHt}$, up to a global phase, using
\begin{equation}
\tilde{\mathcal{O}}\left(\eta^3t + \eta\log\frac{1}{\epsilon} \right)
\end{equation}
Toffoli gates, while using $\displaystyle \mathcal{O}\left(\log\eta+\log\frac{t}{\epsilon}\right)$ ancilla qubits. 
\end{restatable}
We now describe our quantum algorithm in more details, starting from our strategy for block-encoding the potential operator.

\subsection{Block-encoding the potential}
At a high-level, our algorithm block-encodes $V$ via an \gls{lcu} of the block-encodings of each saturated Coulomb term $\displaystyle  V_{ij} = \frac{S(r_{ij})}{2}$. Notice that
\begin{align}\label{eq:LCU_BE}
     V &= \sum_{i\neq j}^{\eta} (-1)^{\sigma_i+\sigma_j} \zeta_i\zeta_j\, V_{ij},\nonumber\\
    &= \sum_{i\neq j}^{\eta} (-1)^{\sigma_i+\sigma_j} \zeta_i\zeta_j \,\left({\textup{SWAP}_{1\leftrightarrow i, \,2\leftrightarrow j}}\right)  V_{12} \left({\textup{SWAP}_{1\leftrightarrow i, \,2\leftrightarrow j}}\right),
\end{align}
where ${\textup{SWAP}_{1\leftrightarrow i, \,2\leftrightarrow j}}$ is the operator that swaps the register corresponding to particle $i$ with that of particle $1$, and particle $j$ with particle $2$. Defining $\bra{0} U_{12}^{(V)}\ket{0} = 2\Delta \cdot  V_{12}$ as the unitary block-encoding $V_{12}$, we can therefore implement the block-encoding $\bra{0}  V_{\text{BE}} \ket{0} =  \Delta\cdot V/\lambda_V$ via
\begin{align}
     V_{\text{BE}} = \text{PREP}_V^\dagger\cdot \text{SEL}_V \cdot\text{PREP}_V,
\end{align}
where
\begin{align}
    \text{PREP}_V \ket{0}\ket{0} = \frac{1}{\sqrt{\lambda_V}}\sum_{i \neq j}^{\eta} \sqrt{\zeta_i\zeta_j} \, \ket{i}\ket{j}, \quad \quad \lambda_V = \sum_{i \neq j}^{\eta} \zeta_i\zeta_j,
\end{align}
and
\begin{align}
    \text{SEL}_V =\left(\sum_{i,j}^\eta (-1)^{\sigma_i + \sigma_j}\,\ketbra{i}{i} \otimes \ketbra{j}{j}\otimes {\textup{SWAP}_{1\leftrightarrow i, \,2\leftrightarrow j}}\right) U^{(V)}_{12} \left(\sum_{i,j}^\eta \ketbra{i}{i} \otimes \ketbra{j}{j}\otimes {\textup{SWAP}_{1\leftrightarrow i, \,2\leftrightarrow j}}\right).
\end{align}
Here the factor of $2$ that multiplies $\Delta$ comes from the fact that we are double counting the particle interactions in the sum. Note how the swap network above can be implemented with complexity linear in $\eta$ due to its product form 
\begin{align}
    \sum_{i,j}^\eta \ketbra{i}{i} \otimes \ketbra{j}{j}\otimes {\textup{SWAP}_{1\leftrightarrow i, \,2\leftrightarrow j}}= \left(\sum_{i}^\eta \ketbra{i}{i} \otimes \mathbb{I}\otimes {\textup{SWAP}_{1\leftrightarrow i}}\right)\left(\sum_{j}^\eta \,\mathbb{I} \otimes \ketbra{j}{j} \otimes {\textup{SWAP}_{2\leftrightarrow j}}\right).
\end{align}
We now describe how the block-encoding $U^{(V)}_{12}$ is implemented via the alternating sign technique. We start by considering the normalized saturated Coulomb interaction $2\Delta \cdot V_{12} = \Delta \cdot S(r_{12})$, which can take values between $0\leq \Delta \cdot S(r) \leq 1$. The idea is to partition this interval into $M$ components with uniform weights $1/M$, where each one will correspond to a different unitary in an \gls{lcu}, recovering $\Delta \cdot S(r)$ by adding all $M$ components. This amounts to partitioning this interval as $[0,1/M,2/M,\cdots,m/M,\cdots,1)$. For a given component $m$, we flag it with a value of $1$ if $m/M<\Delta \cdot S(r)$, and with an alternating $(-1)^m$ otherwise. The alternating sign destructively cancels these $(-1)^m$ components, while the constructive sum adds up to the target $\Delta \cdot S(r)$. This procedure is illustrated in \cref{fig:hero}, which amounts to encoding the Coulomb interaction as the \gls{lcu}
\begin{align}
    2\Delta \cdot V_{12} &= \lim_{M\to \infty}\, \frac{1}{M} \sum_{m=0}^{M-1} \underbrace{\sum_{\vec q_1, \vec q_2} u_m\left(\|\vec q_1 - \vec q_2\|^2\right) \,\,\ketbra{\vec q_1}{\vec q_1} \otimes \ketbra{\vec q_2}{\vec q_2}}_{ U_m},
\end{align}
where
\begin{equation}\label{eq:u_m}
    u_m \left( x \right) := 
    \begin{cases}
        1, & \textup{  if  } \,\, m^2 \cdot x <  M^2, \\[8pt]
        (-1)^m , & \textup{  if  } \,\, m^2 \cdot x \geq  M^2,
    \end{cases} 
\end{equation}
with the full details of this procedure presented in \cref{lemma:sign_trick}. Truncating the above \gls{lcu} at a finite $M = 2^{n_M}$ gives us a block-encoding error in $\mathcal{O}(1/M)$, so to bound the error by $\epsilon$ it suffices to chose $n_M \in\mathcal{O}(\log(1/\epsilon))$. Crucially, the accuracy of the method improves exponentially with the number of qubits $n_M$ that are utilized. Hence, $U^{(V)}_{12}$ can be implemented via
\begin{align}
    U^{(V)}_{12} = (\mathtt{Had}^{\otimes n_M} \otimes \mathbb{I}) \left( \sum_{m=0}^{M-1} \ketbra{m}{m} \otimes U_m  \right) (\mathtt{Had}^{\otimes n_M} \otimes \mathbb{I}). \label{eq:U_12}
\end{align}
We implement $\sum_{m=0}^{M-1} \ketbra{m}{m} \otimes U_m$ by the following procedure: 
\begin{enumerate}
    \item Compute $\ket{\|\vec q_1 - \vec q_2\|^2}$ in an ancillary register using quantum arithmetic on registers $\ket{\vec q_1 }\ket{\vec q_2}$. This is done by first obtaining the absolute differences $\mathtt{AbsDiff}\ket{q_{1,w}}\ket{q_{2,w}}\rightarrow\ket{|q_{1,w} - q_{2,w}|}$ (\cref{lem:abs_diff}) for all three Cartesian directions $w\in\{x,y,z\}$, and then summing the squares of these three components via $\mathtt{Sum\ of\ Squares}$ (Lemma~8 of Ref.~\cite{su2021fault}).
    \item Compute $\ket{m^2}$ in an ancillary register by squaring the register $\ket{m}$ via $\mathtt{Square}$ (Lemma~6 of Ref.~\cite{su2021fault}).
    \item Flag an ancilla qubit if $m^2\cdot \|\vec q_1 - \vec q_2\|^2 \geq  M^2$ using an inequality test. This is done by first multiplying the registers from steps 1 and 2 to encode $m^2\cdot \|\vec q_1 - \vec q_2\|^2$ via $\mathtt{Mult}$ (\cref{lem:fast_mult}). Flagging for the inequality is encoded in the sign qubit after subtraction by $M^2$ using $\mathtt{Sub}$ (Fig.~15 of Ref.~\cite{berry2019qubitization} with added subunits of Fig.~17b from Ref.~\cite{sanders2020compilation}). We use the fact that $M^2=2^{2n_M}$ has Hamming weight $1$ to reduce the cost of this subtraction.
    \item Apply a Pauli $Z$ on the least significant qubit of the $\ket{m}$ register controlled on the flag qubit.
    \item Uncompute arithmetic from steps 1-3.
\end{enumerate}

The associated circuit implementing $U^{(V)}_{12}$ is shown in \cref{fig:U_V}. Note that in practice we can shift the saturated function $S(r)$ by a constant, changing its spectrum as $[0,1/\Delta]\rightarrow[-1/2\Delta,1/2\Delta]$. This effectively reduces the 1-norm by a factor of two by modifying $\mathcal{U}_{\rm arithmetic}$ with a subdominant cost increase, as detailed in \cref{lem:shifted_V}. Note that this also shifts the block-encoded Hamiltonian by a known constant. However, shifts do not modify the dynamics, and their effect can be undone with trivial post-processing if required. Having presented the general approach for block-encoding $V$, we now present an additional technique for a physically motivated 1-norm reduction.

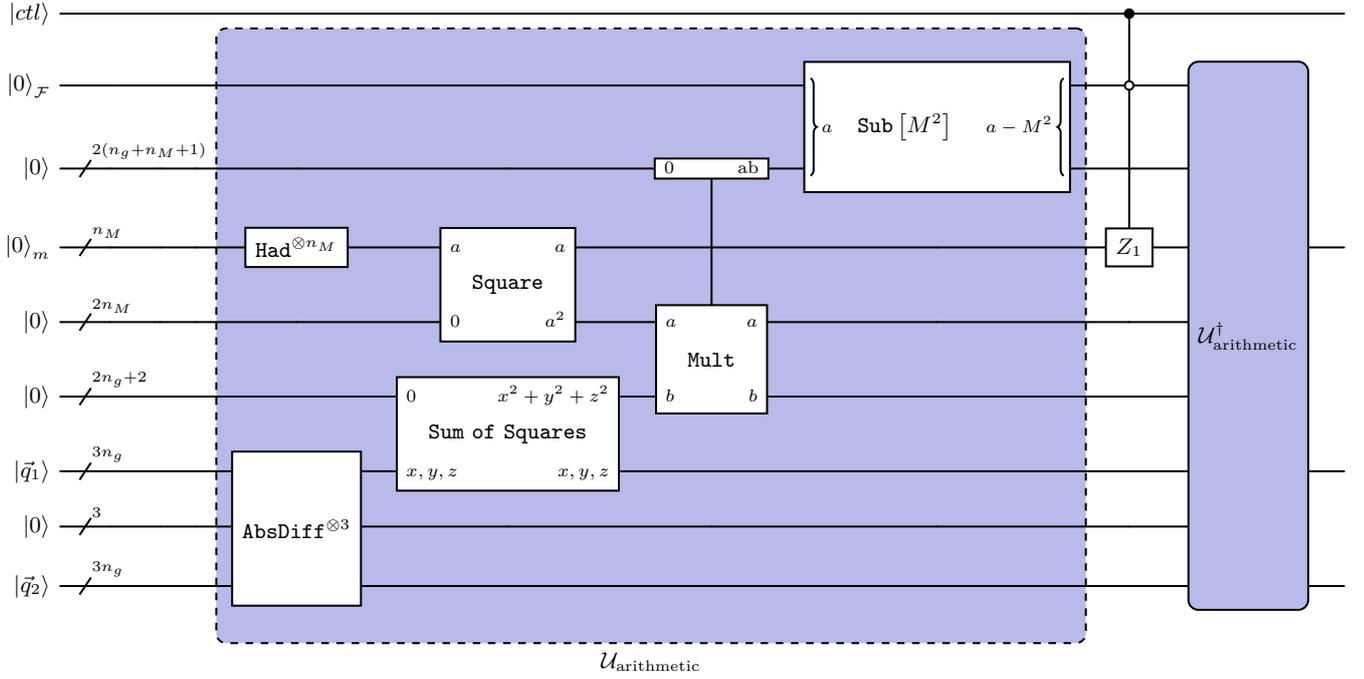
\begin{figure}
    \resizebox{1.0\textwidth}{!}{
    \begin{quantikz}
		\lstick{$\ket{ctl}$} & \phantomgate{nn} &&&  &  &  &  & \ctrl{1}  &  & \\
		\lstick{$\ket{0}_{\mathcal{F}}$} &&&& \gategroup[8,steps=4,style={dashed,rounded
			corners,fill=X1!40, inner
			xsep=3pt, inner ysep=10pt},background,label style={label
			position=below,anchor=north,yshift=-0.2cm,xshift=0cm}]{$\mathcal{U}_{\rm arithmetic}$} &&& \gate[2]{\hspace{0.6cm} \mathtt{Sub}\left[M^2\right] \hspace{1.5cm}} \gateinput[2]{$a$} \gateoutput[2]{$a-M^2$} & \octrl{2} & \gate[8,style={fill=X1!40,rounded corners}]{\mathcal{U}_{\rm arithmetic}^\dagger}  \\
		\lstick{$\ket{0}$} & \qwbundle{2(n_g+n_M+1)} &&&  && \gate{\hspace{1.3cm}} \gateinput{0} \gateoutput{ab} \wire[d][2]{q} &&& \\
		\lstick{$\ket{0}_m$} & \qwbundle{n_M} &&& \gate{\mathtt{Had}^{\otimes n_M}} & \gate[2]{\hspace{0.3cm}\mathtt{Square}\hspace{0.3cm}} \gateinput{$a$} \gateoutput{$a$}  &&& \gate{Z_1} &&  \\
		\lstick{$\ket{0}$} &\qwbundle{2n_M} &&&& \gateinput{$0$} \gateoutput{$a^2$} & \gate[2]{\hspace{0.3cm}\mathtt{Mult}\hspace{0.3cm}} \gateinput{$a$} \gateoutput{$a$} &&& \\
		\lstick{$\ket{0}$} & \qwbundle{2n_g+2} &&&& \gate[2]{\hspace{0.3cm} \mathtt{Sum \ of \ Squares} \hspace{0.3cm}} \gateinput{$0$} \gateoutput{$x^2+y^2+z^2$} & \gateinput{$b$} \gateoutput{$b$} &&& \\
		\lstick{$\ket{\vec q_1}$} & \qwbundle{3n_g} &&& \gate[3]{\mathtt{AbsDiff}^{\otimes 3}} & \gateinput{$x,y,z$} \gateoutput{$x,y,z$} &&&&& \\
		\lstick{$\ket{0}$} & \qwbundle{3} &&&&  &&&& \\
		\lstick{$\ket{\vec q_2}$} & \qwbundle{3n_g} &&&  &&&&&&
	\end{quantikz}}
    \caption{Controlled implementation of Coulomb term $U^{(V)}_{12}$ in Eq.~\eqref{eq:U_12}. Here we have defined all the operations inside the blue box as $\mathcal{U}_{\rm arithmetic}$, with the qubit register labelled as $\ket{0}_{\mathcal{F}}$ carrying the information of the flagging procedure. See \cref{app:circuits} for more details on the circuit directives. Note that the ancilla qubits from the subtraction by $M^2$ are kept alive and uncomputed when applying $\mathcal{U}_{\rm arithmetic}^\dagger$ with measurement and phase. All qubits registers not shown in the output have been returned to all-zeros state.}
    \label{fig:U_V}
\end{figure}

\subsubsection{Physically motivated optimizations}
In the approach outlined above, all interaction types (electron-electron, electron-nucleus, and nucleus-nucleus) are treated on an equal footing. Consequently, the maximum value of any interaction term, $1/\Delta$, set by the grid spacing $\Delta$, must be chosen to accommodate all interaction types. However, in realistic chemical systems, the cutoff distance beyond which the wavefunction has negligible support in configurations where two particles approach closer than the cutoff is substantially larger for nuclear–nuclear pairs than for electron–nuclear or electron–electron pairs. In the following, we describe an algorithmic modification to the above implementation to reduce the 1-norm of our block-encoding by utilizing this fact.\\

Let $\Gamma$ be the multiplicative factor by which we allow this cutoff distance to be larger than the other types of interactions. That is, we know nuclei pairs will not come closer than $\Gamma \cdot \Delta$ to each other during the dynamics (see \cref{app:Gamma} for a more detailed discussion). We achieve this by introducing minor control logic in our SELECT circuit to instead check for the inequality \begin{align}
    m^2\cdot \|\vec q_1 - \vec q_2\|^2 >  \Gamma^2 \cdot M^2,
\end{align}
whenever both particles are nuclei, which requires replacing the subtraction routine in \cref{fig:U_V} by its controlled counterpart in \cref{lemma:hybrid_adder}. We then need to modify the PREPARE circuit for $V$ such that
\begin{align}
    \text{PREP}_{V,\Gamma} \ket{0}\ket{0} = \frac{1}{\sqrt{\tilde\lambda_V}}\sum_{i \neq j}^{\eta} \sqrt{\gamma(i,j)}\, \ket{i}\ket{j}, \quad \quad \tilde\lambda_V = \sum_{i \neq j}^{\eta} \gamma(i,j),
\end{align}
where
\begin{align}
    \gamma(i,j) = \begin{cases}
        \frac{\zeta_i\zeta_j}{\Gamma},& \textrm{ if particle } i \textrm{ and } j \textrm{ are both nuclei }, \\[8pt]
        \zeta_i\zeta_j& \textrm{ otherwise }.
    \end{cases}
\end{align}
In principle, the constant $\Gamma$ can also depend on which species of nuclei are interacting. For example, the average bond distance between carbon atoms is significantly larger than that of hydrogen atoms. However, such specifications require additional modifications for flagging different interaction types and controlling the inequality test, which is left as future work. Finally, we note that this technique can be combined with the previously mentioned shifting technique, reducing the 1-norm by a factor of two, with all details presented in \cref{lem:shifted_V}.

\subsection{Block-encoding the kinetic operator}
Similar to the potential, we block-encode $T$ via an \gls{lcu} of the block-encodings of each term $\displaystyle  T_{i} = \frac{\|\vec p_i\|^2}{2}$, where $\vec p_i$ is the momentum vector associated to the real-space coordinate $\vec q_i$ via a discrete Fourier transform over each Cartesian direction. The kinetic energy is thus written as 
\begin{align}
    T &= \sum_{i=1}^{\eta} \frac{1}{m_i} T_{i},\nonumber\\
    &= \sum_{i=1}^{\eta} \frac{1}{m_i} \left({\textup{SWAP}_{1\leftrightarrow i}}\right)T_{1} \left({\textup{SWAP}_{1\leftrightarrow i}}\right)
\end{align}
Defining $\bra{0} U_{1}^{(T)}\ket{0} = 2 \Delta^2 \cdot  T_{1}$  as the unitary block-encoding $T_{1}$, we can therefore implement the block-encoding $\bra{0}  T_{\text{BE}} \ket{0} =  \Delta^2 \,T/\lambda_T$ via
\begin{align}
     T_{\text{BE}} = \text{PREP}_T^\dagger\cdot \text{SEL}_T \cdot\text{PREP}_T,
\end{align}
where
\begin{align}
    \text{PREP}_T \ket{0} = \frac{1}{\sqrt{\lambda_T}}\sum_{i}^{\eta} \frac{1}{\sqrt{m_i}} \, \ket{i}, \quad \quad \lambda_T = \sum_{i}^{\eta} \frac{1}{m_i},
\end{align}
and
\begin{align}
    \text{SEL}_T =\left(\sum_{i=1}^\eta \,\ketbra{i}{i}\otimes {\textup{SWAP}_{1\leftrightarrow i}}\right) U_{1}^{(T)} \left(\sum_{i=1}^\eta \,\ketbra{i}{i} \otimes {\textup{SWAP}_{1\leftrightarrow i}}\right).
\end{align}
We now describe how the block-encoding $U_{1}^{(T)}$ is implemented. We can write $\Delta^2 \cdot  T_{1}$ as an \gls{lcu} by using a swap network to nest computation over different Cartesian indices $w$ as:
\begin{align}
    U_1^{(T)} & = \sum_{w\in\{x,y,z\}}  (\text{SWAP}_{x\leftrightarrow w}) \qft_{1,x}^\dagger\left(\sum_{q_{1,x}} |q_{1,x}|^2 \ketbra{q_{1,x}}{q_{1,x}}\right) \qft_{1,x}  (\text{SWAP}_{x\leftrightarrow w}), \label{eq:T1}
\end{align}

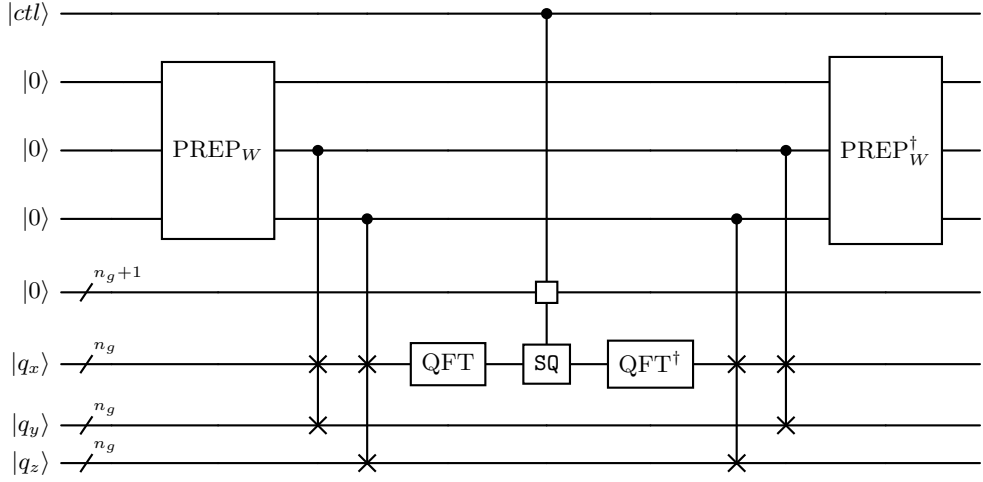
\begin{figure}
    \centering
    \begin{quantikz}
	\lstick{$\ket{ctl}$} & &&&&& \ctrl{4} & &&&& \\
	\lstick{$\ket{0}$} & \phantomgate{12} & \gate[3]{\text{PREP}_W} &&&&&&&& \gate[3]{\text{PREP}^\dagger_W} &  \\
	\lstick{$\ket{0}$} &&& \ctrl{4} &&&&&& \ctrl{4} && \\
	\lstick{$\ket{0}$} &&&& \ctrl{4}  &&&& \ctrl{4} &&& \\
	\lstick{$\ket{0}$} & \qwbundle{n_g+1} &  &&&& \gate{} \vqw{1}  &&&  &  & \\
	\lstick{$\ket{q_{x}}$} & \qwbundle{n_g} && \targX{} & \targX{} & \gate{\text{QFT}} & \gate{\mathtt{SQ}} & \gate{\text{QFT}^\dagger} & \targX{} & \targX{} &&  \\
	\lstick{$\ket{q_{y}}$} & \qwbundle{n_g} && \targX{} &&&&&& \targX{} && \\
	\lstick{$\ket{q_{z}}$} & \qwbundle{n_g} &&& \targX{} &&&& \targX{} &&& \\
\end{quantikz}
    \caption{Controlled block-encoding of kinetic energy of single particle $U_1^{(T)}$ in \cref{eq:T1}, with squaring routine implementing $\mathtt{SQ}\ket{q}\rightarrow (2q^2-1)\ket{q}$. See \cref{app:circuits} for a detailed explanation of all circuit directives.}
    \label{fig:T1}
\end{figure}
where the $(1,x)$ indicates operations being done on the register encoding particle $1$'s Cartesian direction $x$. We now require a routine squaring a two's complement register for implementing the $|q|^2$ component, which is done here by block-encoding the action $\mathtt{SQ}\ket{q}= (2q^2-1)\ket{q}$. This routine has the benefit of effectively reducing the associated 1-norm by a factor of $2$ when compared to a direct block-encoding of $\ket{q}\rightarrow q^2\ket{q}$, while the identity causes a known constant shift in the Hamiltonian that can be effectively ignored. 

The basic idea for implementing this operation is to implement it by two repetitions of a walk operator block-encoding $\mathtt{Amp}\ket{q}= q\ket{q}$. To understand how the block-encoding $\mathtt{Amp}$ works, note that for an $n$-bit unsigned (positive) integer $a$ we can write $a=\sum_{b=0}^{n-1} \bar a_b 2^b$, where $\bar a_b=\{0,1\}$ is the value of the $b$th bit. We can encode the amplitudes of $n$-bit positive integers in a one-hot encoded state $\mathtt{amp_n}\ket 0^n := \ket{\sqrt{\mathtt{amp_n}}} \propto \sum_{b=0}^{n-1}2^{b/2}X_b\ket{0}^n$, where $X_b$ is a Pauli X operator acting on qubit $b$. The idea is to then check the bit-wise product of this state with the bits encoding $a$, appending amplitudes to the block-encoding when the associated bit $\bar a_b=1$ and cancelling amplitudes when $\bar a_b=0$. We thus obtain the block-encoding for unsigned integers:
\begin{equation} \label{eq:square_overview}
    \sum_a a\ketbra{a}{a} \rightarrow (\mathtt{Had}\otimes\mathtt{amp_n}\otimes \mathbb{I})^\dagger \Bigg(\sum_{h=0,1}\sum_{b=0}^{n-2} \ketbra{h,b}{h,b} \otimes U_{h,b} \Bigg),(\mathtt{Had}\otimes\mathtt{amp_n}\otimes \mathbb{I})
\end{equation}
where we have defined
\begin{equation}
    U_{h,b}\ketbra{a}{a} = k_h(\bar{a}_b)\ketbra{a}{a}
\end{equation}
for
\begin{equation}
    k_h(x) = \begin{cases}
        1,\:&\textrm{if}\: x=1\\
        (-1)^h,\:&\textrm{if}\: x=0.
    \end{cases}
\end{equation}
Note that some additional details are required due to how we are dealing with negative numbers, having the full specification of this block-encoding presented in \cref{lemma:walk_square}. On a high level, the circuit implementing $T_1$ is shown in \cref{fig:T1}, corresponding to the following steps:
\begin{enumerate}
    \item Prepare a uniform superposition over three qubits $\ket{W}=(\ket{100}+\ket{010}+\ket{001})/\sqrt{3}$ (\cref{lemma:prep_w}). Controlled on the associated one-hot encoding of $\ket{w}$, swap the Cartesian registers $x\leftrightarrow w$. 
    \item Apply $\qft$ on the register $\ket{q_{1,x}}$, the $x$th Cartesian coordinate of the $1$st particle register that has been swapped in.
    \item Apply the block-encoding for $\mathtt{SQ}\ket{q}\rightarrow(2q^2-1)\ket{q}$ on the $\ket{\cdot}_{1,x}$ register (\cref{lemma:walk_square}).
    \item Uncompute steps 1-2.
\end{enumerate}
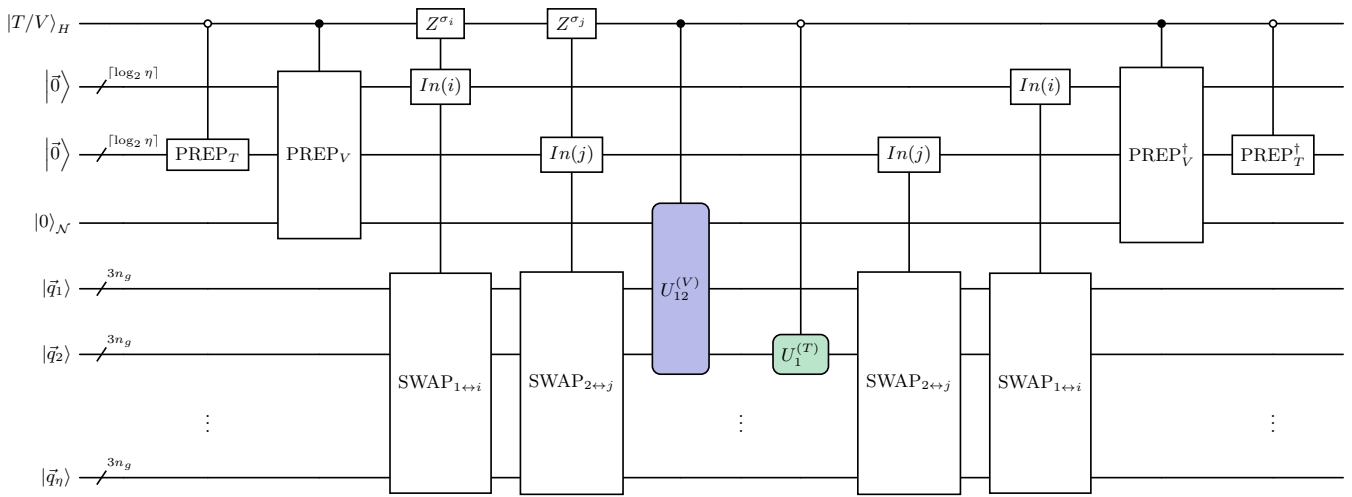
\begin{figure}
    \centering
    \hspace*{-0.0cm}\resizebox{1.0\textwidth}{!}{
    \begin{quantikz}
		\lstick{$\ket{T/V}_{H}$} & \phantomgate{123} & \octrl{2} & \ctrl{1} & \gate[1]{Z^{\sigma_i}} \vqw{1}  & \gate[1]{Z^{\sigma_j}}\vqw{2} & \ctrl{3} && \octrl{5} &  &  & \ctrl{1} & \octrl{2} & \\
		\lstick{$\ket{\vec 0}$} & \qwbundle{\ceil{\log_2\eta}} && \gate[3]{\textup{PREP}_V} & \gate{In(i)} \vqw{3} &&&&&& \gate{In(i)} \vqw{3} & \gate[3]{\textup{PREP}_V^\dagger} && \\
		\lstick{$\ket{\vec 0}$} & \qwbundle{\ceil{\log_2\eta}} & \gate{\textup{PREP}_T} &&& \gate{In(j)} \vqw{2} &&&& \gate{In(j)} \vqw{2} &&& \gate{\textup{PREP}_T^\dagger} & \\
		\lstick{$\ket{0}_{\mathcal{N}}$} &&&&&& \gate[3,style={fill=X1!40,rounded corners}]{U_{12}^{(V)}} &&& &&&& \\
		\lstick{$\ket{\vec q_1}$} & \qwbundle{3n_g} &&& \gate[4]{\textrm{SWAP}_{1\leftrightarrow i}} & \gate[4]{\textrm{SWAP}_{2\leftrightarrow j}} &&&  & \gate[4]{\textrm{SWAP}_{2\leftrightarrow j}} & \gate[4]{\textrm{SWAP}_{1\leftrightarrow i}} &&&  \\
		\lstick{$\ket{\vec q_2}$} & \qwbundle{3n_g} &&&&&&& \gate[1,style={fill=X2!40,rounded corners}]{U_1^{(T)}} &&&&& \\
		 \setwiretype{n} & & \vdots &&&&& \vdots &  &&&& \vdots & \\
		\lstick{$\ket{\vec q_\eta}$} & \qwbundle{3n_g} &&&&&&&&&&&& 
	\end{quantikz}
    }
    \caption{\sel operator for block-encoding full Hamiltonian in \cref{eq:ham_sel}. The block-encoding $U_{12}^{(V)}$ is shown in \cref{fig:U_V}, while $U_1^{(T)}$ corresponds to \cref{fig:T1}. The $\ket{\cdot}_{\mathcal{N}}$ register flags nuclear-nuclear interactions for the optimized potential energy implementation.}
    \label{fig:full_ham}
\end{figure}

\subsection{Block-encoding the full Hamiltonian}
Equipped with a block-encoding of the kinetic and potential terms, we can now combine them with an \gls{lcu} circuit as an outer-loop. This is achieved by summing the two block-encodings with relative weights proportional to their 1-norm, which can be done with the rotation $R_y(\theta_{\alpha})\ket{0}_H = \cos(\theta_{\alpha}/2)\ket{0}_H + \sin(\theta_{\alpha}/2)\ket{1}_H$ for $\theta_{\alpha}:= 2\arccos\sqrt{\alpha_T/(\alpha_V+\alpha_T)}$ up to some specified precision. Here we used $\alpha_T = \lambda_T 2^{2(n_g-1)}3 \pi^2 L^{-2}$ and $\alpha_V = \lambda_V/4\Delta$. We can now define the outer-loop PREPARE and SELECT circuits as follows:
\begin{equation} \label{eq:prep_ham}
    \textup{PREP}\ket{0}_{H} = \sqrt{\frac{\alpha_{T}}{\alpha_{V}+\alpha_{T}}} \ket{0}_H + \sqrt{\frac{\alpha_{V}}{\alpha_{V}+\alpha_{T}}} \ket{1}_H,
\end{equation}

\begin{equation} \label{eq:ham_sel}
    \textup{SEL} = \ketbra{0}{0}_H \otimes  O_{T} + \ketbra{1}{1}_H \otimes  O_V,
\end{equation}
where $ O_X$ represents the block-encoding of the associated operator. Constructing these unitaries allows to sum the block-encodings provided by $ O_{V}$ and $ O_T$ as
\begin{equation}
    \bra{0}_H(\textup{PREP}^\dagger \otimes \mathbb{I}) \textup{SEL} ( \textup{PREP} \otimes \mathbb{I}) \ket{0}_H  = \frac{ T + V}{\alpha_{V}+\alpha_{T}} = \frac{ H}{\alpha_H},
\end{equation}
As expected, the subnormalization of the overall block-encoding is the sum of the 1-norms of each block-encoding in the sum, namely $\alpha_H=\alpha_V+\alpha_T$. We formalize this, as well as the other block-encoding parameters for a sum of distinct block-encodings in \cref{lem:be_sum} in \cref{app:proofs}. Note that the block-encodings for the kinetic and potential operator can here be chosen to be either of the different versions with or without optimizations (e.g. variable saturation, spectral shifts); the associated 1-norms $\alpha_T$ and $\alpha_V$ will depend on what version of the block-encoding is chosen. \\

An additional optimization that can be done when performing this linear combination of block-encodings of $ V$ and $ T$ by noticing that both block-encodings use a swap network. This fact can be exploited to implement the swap network only once for both block-encodings simultaneously, as shown in \cref{fig:full_ham}. The correct action of this circuit follows immediately from a case-by-case analysis over the control qubit, having that controlled on $0(1)$ it implements the block-encoding for $ T( V)$. Note that some control logic has been added to include the electronic phases $(-1)$ for encoding signs of charges in the potential energy, which as shown in \cref{subapp:swap_net} can be done using only CZ operations for no additional non-Clifford cost.

\section{Application} \label{sec:discussion}
Having presented our algorithm for performing pre-Born-Oppenheimer dynamics on a quantum computer, we now discuss potential applications of our approach and report resource requirements for various industrially and scientifically relevant reactions. For chemical processes characterized by strong non-adiabatic effects and pronounced nuclear quantum effects, particularly in regimes where accurate electronic structure is unavailable, nuclear motion is large, or reaction pathways are not known a priori, black-box ab initio approaches such as the one proposed here provide the only systematic means of capturing the underlying physics without introducing uncontrolled approximations. This capability is essential for mechanistic studies and screening workflows, where prior identification of relevant nuclear configurations is often infeasible. By directly evolving the full molecular wavefunction in real time, our approach enables access to virtually any observable of interest throughout the reaction, thereby supporting a comprehensive elucidation of reaction mechanisms.\\

We identify two primary cases where this rigorous treatment is essential for industrial and scientific applications. First are cases where the Born-Oppenheimer approximation fails, which generally happens as multiple electronic states become energetically accessible throughout a reaction and non-adiabatic effects play important roles. This is most often exemplified in photochemistry, with relevant applications in photolithography \cite{lith_1,lith_2,lith_3}, photocatalysis \cite{photocatalysis_1,photocatalysis_2}, light-harvesting processes \cite{light_1,light_2}, and atmospheric chemistry \cite{atmos_1,atmos_2}. Non-adiabatic effects also appear in radical chemistry due to small energy gaps in different electronic states, allowing thermal access to excited states in some biological \cite{radical_1,radical_2,radical_3} and combustion \cite{thermal_1,thermal_2} processes, while also having a pervasive appearance in organic synthesis reactions \cite{organocatalysis_1,organocatalysis_2}. Second, nuclear quantum effects such as zero-point energy and tunnelling are dominant in isotope separation \cite{isotope_1,isotope_2} and hydrogen transfer reactions \cite{nbo_4,hydrogen_1,hydrogen_2}, which play central roles in green energy production \cite{green_1,green_2}. Perhaps most critical are coupled phenomena, such as proton-coupled electron transfer in fuel and solar cells \cite{pcet_1,pcet_2} and symmetry breaking (e.g., Peierls distortion) in organic electronics \cite{solitons_1,solitons_2}, which rely on a concerted interplay between electronic and nuclear quantum effects. \\

We now present resource estimates for the simulation of selected chemical reactions with potential industrial impact. We start by providing a list of the reactions and a short discussion of their industrial relevance. Note that the addition of pseudopotentials could alleviate qubit requirements by both decreasing electron counts and relaxing grid spacings, while diminishing 1-norms and including relativistic effects for heavier elements. This is left as future work, where we expect large portions of the required machinery to be similar to the pseudoions implementations from Ref.~\cite{jornada_comprehensive_2025} while working on a real-space grid. As such, we here focus on chemistry consisting of elements on the first two rows of the periodic table.

\begin{enumerate}
    \item Ammonia and boron trifluoride Lewis acid-base reaction $(\rm NH_3+BF_3)$. This reaction contains electron transfer and bond stabilization, which are fundamental processes in many catalytic reactions. This system was used for benchmarking with Ref.~\cite{jornada_comprehensive_2025}, while still exhibiting stereotypical chemistry of interest.
    
    \item Nitrogen dioxide dimerization $(\rm 2NO_2)$. This reaction is related to nitric acid production in the Ostwald process \cite{nitrogen_ostwald}, as well as rocket propellants \cite{nitrogen_rocket} and atmospheric chemistry \cite{nitrogen_iso_2}. Despite the dimerization itself being well described by Born-Oppenheimer approximation, it has been shown that in industrial and ignition settings the $\rm N_2O_4$ dimer can isomerize while changing electronic character, which leads to non-adiabatic effects \cite{nitorgen_iso_1,nitrogen_iso_2}.

    \item Hydrocarbon combustion $(\rm C_2H_4+O_2)$. The combustion of hydrocarbons remains a fundamental driver of global energy systems, having a large variety of radical species and strong non-adiabatic effects which make their classical simulation difficult. This reaction also appears during the Wacker process in presence of a catalyst, which is used for the production of acetaldehyde with a large impact in the production of specialty chemicals. The extreme conditions under which this process occurs make it extremely challenging to model, with the exact reaction mechanism being a highly debated topic \cite{wacker_1,wacker_2}.

    \item Ozonolysis of ethene $(\rm C_2H_4+O_3)$. The ozone-based cleavage of carbon double bonds is used to produce a variety of high value specialty chemicals in the pharmaceutical industry \cite{ozonolysis_2,ozonolysis_3}, with ozonolysis reactions also appearing in atmospheric chemistry. The diradical character of ozone has a complex electronic structure, with a variety of competing decomposition pathways and low reaction barriers making classical simulations challenging.

    \item Photoactivated double proton-coupled electron transfer $(\rm C_{23}H_{20}N_3O)$. Molecules with similar structures \cite{dpcet_1,dpcet_2} have been proposed for designing next-generation solar harvesting devices, potentially having unparalleled efficiencies by improving the electron transfer through a coupling with protons. Overall, proton-coupled electron transfer plays an important role in solar light harvesting processes. Having access to excited states through the photoactivation requires the inclusion of non-adiabatic effects, while the proton transfer needs a full quantum mechanical treatment of protons.  
\end{enumerate}

Resource estimates for performing 1 femtosecond of time evolution for the discussed systems are presented in \cref{tab:estimates}. Our estimates were obtained by adding the costs of each component in our block-encoding subroutine (\cref{subapp:be_H}), multiplied by the number of calls made to the block-encoding of the Hamiltonian for implementing $e^{-iHt}$ via \gls{qsp} (\cref{app:complexity}). The potential energy was implemented with the most optimal version as discussed in \cref{lem:shifted_V}, which incorporates both a variable saturation for the nuclear-nuclear interactions (\cref{app:Gamma}) alongside a spectral shift that halves the 1-norm of our block-encoding by introducing a global phase. All relevant hyperparameters for our estimations are presented in \cref{app:hyperparams}. A maximum error of $\epsilon=10^{-2}$ was used for the total time evolution, with the error allocation being optimized to minimize the cost while satisfying this total error budget. Note that the subdominant logarithmic dependence of resources with respect to errors translates to minor variations in the algorithmic costs even if $\epsilon$ and the error allocation vary significantly. For instance, reducing the error by $9$ orders of magnitude to $\epsilon=10^{-11}$, only increases the Toffoli count by less than a factor of $2.1$, requiring $1.76\times 10^{10}$ instead of $8.72\times 10^{9}$ Toffolis per femtosecond for the $\rm NH_3+BF_3$ reaction, while requiring $4.13\times 10^{11}$ instead of $2.73\times 10^{11}$ for the $\rm C_{23}H_{20}N_3O$ system, while increasing the ancilla counts from $312$ and $582$ to $612$ and $700$ respectively. Using the uniform allocation of errors in \cref{tab:error_dist} instead of an optimized one modified the Toffoli counts by less than $5\%$ in all reactions. This minor dependence with respect to errors shows how for practical purposes a uniform error allocation is enough for close to optimal results. \\

\begin{figure}
    \centering
    \includegraphics[width=1.0\linewidth]{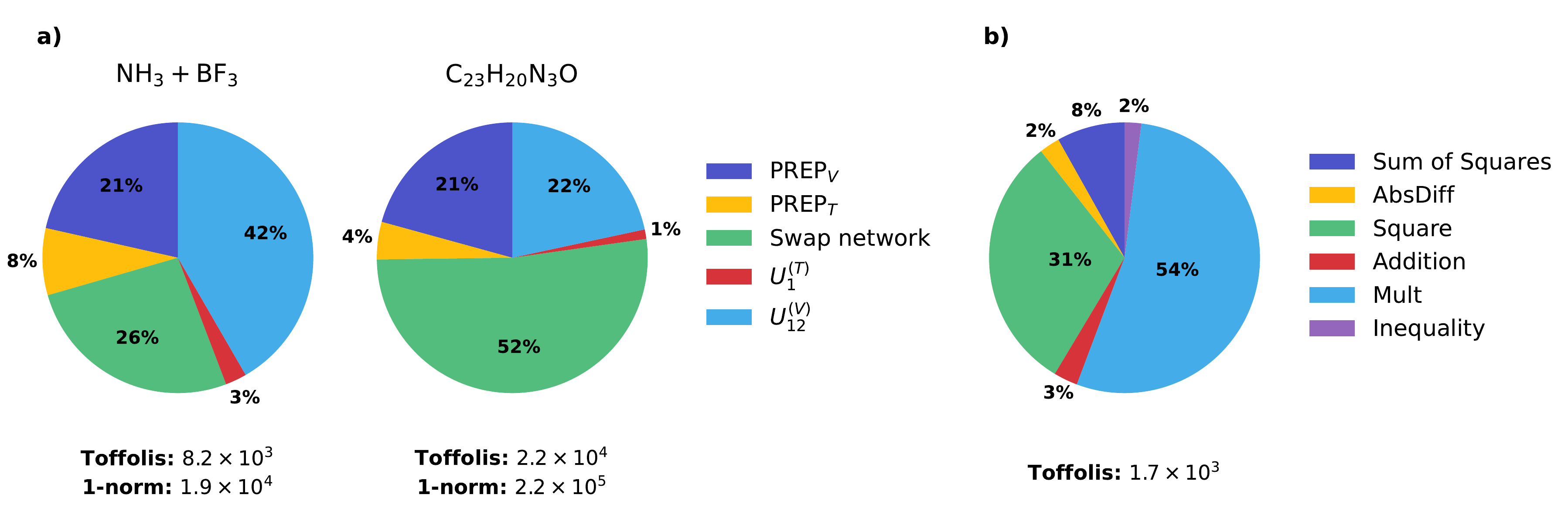}
    \caption{a) Toffoli cost decomposition per walk operator for different reactions. b) Toffoli cost decomposition of $U^{(V)}_{12}$ operation implementing $1/r$ interaction for $\rm NH_3+BF_3$ reaction. Note that cost profiles are extremely similar for different reactions.}
    \label{fig:costs}
\end{figure}

The Toffoli requirements for each walk operator are relatively modest, ranging from $8.1\times 10^3$ to $2.2\times 10^4$ for all systems studied in this work, making the 1-norms the significant amplifier of the algorithmic overhead. This is particularly noticeable for the higher cost of the $\rm C_2H_4+O_2$ reaction, having that the high-temperature required a significantly finer grid, which in turn increased its 1-norm and total Toffoli cost by about an order of magnitude. Overall, we expect for techniques reducing the 1-norms, such as the variable saturations used in this work, to provide a direction for further algorithmic improvements. The inclusion of pseudopotentials could be used both for this purpose and for lowering the qubit counts, with their implementation left as future work. Finding additional ways to further lower the qubit counts will be important for making this type of black-box simulation of chemical reactions feasible for the first generation of fault-tolerant hardware. Note how the main contribution for the logical qubits comes from the first-quantized encoding of the molecular wavefunction, with ancilla qubits consisting of less than $26\%$ of the total qubit counts across all studied reactions. \\

A profiling of the different factors contributing to the overall cost is presented in \cref{fig:costs}a. It can be seen how the main contribution to the cost for smaller reactions comes from the alternating sign implementation of the $1/r$ interaction through the $U_{12}^{(V)}$ operator. However, this cost remains mostly constant as the number of particles increases, having that the swap network becomes more expensive as the system size grows: this is the only component with a dominant contribution that has a linear scaling with respect to the number of particles, namely $\mathcal{O}(\eta)$. For the PREPARE routines, the utilization of QROM with dirty ancillas \cite{berry2019qubitization,low_trading_2024}, often called QROAM, effectively decreases their Toffoli scaling by a quadratic factor as $\tilde{\mathcal{O}}(\eta) \rightarrow \tilde{\mathcal{O}}(\sqrt{\eta})$, having that the optimal number of dirty ancillas for minimizing the Toffoli counts was always available throughout this work. Note that for the smaller reactions (e.g. $\rm NH_3+BF_3$) it is more advantageous to use the PREPARE routine loading a symmetric matrix that directly loads $\mathcal{O}(\eta^2)$ coefficients, while for larger systems (e.g. $\rm C_{23}H_{20}N_3O$) the PREPARE routine based on amplitude amplification that only loads $\mathcal{O}(\eta)$ coefficients becomes more advantageous, as expected from its better scaling. See \cref{subapp:prep} for a more detailed discussion on the different PREPARE routines for the potential energy term. Finally, we note how the Toffoli cost contribution from the kinetic energy routine is completely subdominant for each individual oracle, having that the significant contribution from the kinetic term will only appear through the 1-norm. \\

The cost profile for the operations implementing the $1/r$ interaction corresponding to the $U_{12}^{(V)}$ operator is shown in \cref{fig:costs}b. The main cost of this routine corresponds to the multiplication step before the inequality test, namely $\ket{m^2}\ket{r^2}\rightarrow \ket{m^2}\ket{r^2}\ket{m^2r^2}$, with the other large contribution being the squaring $\ket{m}\rightarrow \ket{m}\ket{m^2}$. The large cost from these routines comes from the number of qubits $n_M$ that is required in the $\ket{m}$ register for implementing the alternating sign technique. Overall, we expect for this kind of profiling of routine costs to provide guidance for focusing further algorithmic improvements. \\

With a simulation cost per femtosecond of $8.72\times 10^{9}$ Toffolis for the $\rm NH_3+BF_3$ reaction, our estimates present the lowest cost so far for a black-box quantum simulation of this chemical reaction. While efficient fully first-principles modelling remains a grand challenge, these results indicate that continued optimization could unlock this application on first-generation fault-tolerant devices. Furthermore, our framework could potentially be extended as to include thermodynamic environments \cite{simon2024improved}, with solvent effects being crucial for high-impact applications such as the modelling of homogeneous catalysis \cite{dyson2016solvent}. 

\glsresetall
\section{Conclusion} \label{sec:conclusion}
In this work we introduced a first-principles quantum algorithm for the direct simulation of chemical dynamics on a real-space grid. By treating electrons and nuclei on an equal footing, our approach goes beyond the Born-Oppenheimer approximation, providing a scalable framework for investigating radical reactions, photo-chemistry, and general non-adiabatic phenomena without the need for uncontrolled approximations. \\

We presented a \gls{qsp}-based framework for simulating pre-Born-Oppenheimer dynamics with significantly reduced resource requirements compared to previous works \cite{jornada_comprehensive_2025}. Central to our approach is the ``swap network'' block-encoding, which enables the implementation of the potential energy operator with linear complexity $\mathcal{O}(\eta)$ despite having $\mathcal{O}(\eta^2)$ pairwise interactions, while reducing the required \gls{qft}s for the kinetic energy implementation from $6\eta$ to only $2$. Furthermore, we address the $1/r$ electrostatic interaction bottleneck through an efficient implementation using the alternating sign technique for block-encoding diagonal operators. In addition, we introduced two 1-norm reduction techniques: a global energy shift that halves the 1-norm without modifying the dynamics, and a saturation technique leveraging statistical bounds on minimal nuclear-nuclear distances. Given a 1-norm that grows quadratically with the number of particles, our framework yields an overall complexity of $\tilde{\mathcal{O}}(\eta^3 t)$ for time evolution $e^{-iHt}$. \\

Through the presented optimizations, we achieved over an order-of-magnitude reduction in Toffoli requirements compared to existing state-of-the-art \cite{jornada_comprehensive_2025}, while also requiring a smaller number of ancilla qubits. For the simulation of the ammonia and boron trifluoride $\rm NH_3+BF_3$ reaction, $8.72\times 10^{9}$ Toffoli gates per femtosecond of simulation time were required, while algorithmic ancillas had a subdominant contribution to the overall qubit requirements: $1050$ qubits were needed for the molecular wavefunction while only using $312$ additional ancillas for the qubitization routine. While the Toffoli cost associated with individual walk operators remains relatively modest, the overall algorithmic complexity continues to be dominated by the magnitude of the 1-norms. This suggests that future improvements for lowering the 1-norms present a promising avenue for further optimizations, alongside the incorporation of pseudopotentials. We note that the preparation of initial states is expected to have a subdominant cost with respect to that of the time evolution, with our initial state preparation routine being identical to that from Ref.~\cite{jornada_comprehensive_2025} up to a \gls{qft}. Detailed construction of initial state preparation routines and associated resource estimates are left as future work. \\

By providing a highly optimized black-box algorithm that is particularly well-suited for organic and photo-organic chemistry, this work narrows the gap between fully first-principles simulations and the requirements of fault-tolerant hardware. Efficient first-principles simulations hold the potential to revolutionize the modelling of chemical reactions that remain intractable for classical methods, positioning quantum computers to address high-impact industrial and scientific challenges.

\acknowledgements
We thank Tyler Kharazi for providing algorithmic suggestions and feedback that improved the quality of the paper. We also thank Diksha Dhawan, Jay Soni, and Torin Stetina for useful discussions, and Tarik El-Khateeb for producing the \cref{fig:hero} graphic. MP acknowledges funding from Mitacs and the NSERC Post-Graduate Doctoral Scholarship. We acknowledge the Applied Quantum Computing Challenge program of the National Research Council of Canada for financial support (grant number AQC-103-2). NW’s work on this project was supported by the U.S. Department of Energy, Office of Science, National Quantum Information Science Research Centers, Co-design Center for Quantum Advantage (C2QA) under contract number DE- SC0012704 (PNNL FWP 76274) and PNNL’s Quantum Algorithms and Architecture for Domain Science (QuAADS) Laboratory Directed Research and Development (LDRD) Initiative. The Pacific Northwest National Laboratory is operated by Battelle for the U.S. Department of Energy under Contract DE-AC05-76RL01830.

\bibliography{bib}

\newpage

\appendix
\tableofcontents
\newpage

\starttocentries
\section{Circuit details} \label{app:circuits}
In this appendix we present the circuits and formal deductions of associated costs for the directives used in this work. These are summarized in \cref{tab:circuits}.

\begin{table}
    \centering
    \setlength\extrarowheight{3pt}
    \begin{tabular}{|c|c|c|} \hline
        \textbf{Directive} & \textbf{Action} & \textbf{Reference} \\ \hline \hline
        $\textup{PREP}_{\rm amp}, \textup{SEL}_{\rm amp}$& $\ket{a}\rightarrow a\ket{a}$& \cref{lemma:amp} \\ \hline
        $\mathtt{SQ}$ & $\ket{a}\rightarrow (2a^2-1)\ket{a}$& \cref{lemma:walk_square} \\ \hline
        $\textup{SWAP}_{1\leftrightarrow k}$ & $\ket{a_1}\ket{a_2}\cdots\ket{a_K}\rightarrow \ket{a_k}\ket{a_2}\cdots \ket{a_{k-1}}\ket{a_1}\ket{a_{k+1}}\cdots\ket{a_K}$ & -- \\ \hline
        $\textup{PREP}_W$ & $\ket{000}\rightarrow \ket{W}=(\ket{001}+\ket{010}+\ket{100})/\sqrt{3}$ & Ref.~\cite{mermin2007quantum} and \cref{fig:w_prep} \\ \hline
        $\mathtt{Abs}$ & $\ket{a} \rightarrow \ket{|a|}$& \cref{lemma:abs}  \\ \hline
        $\mathtt{AbsDiff}$ & $\ket{a}\ket{b} \rightarrow \ket{|a-b|}\ket{b}$&  \cref{lem:abs_diff}  \\ \hline
        $\mathtt{Square}$ & $\ket{a}\rightarrow \ket{a}\ket{a^2}$& Lemma 6 of Ref. \cite{su2021fault} \\ \hline
        $\mathtt{Sum\ of\ Squares}$ & $\ket{a,b,c} \rightarrow \ket{a,b,c} \ket{a^2+b^2+c^2}$& Lemma 8 of Ref. \cite{su2021fault}  \\ \hline
        $\mathtt{Mult}$ & $\ket{a}\ket{b}\rightarrow \ket{a}\ket{b}\ket{ab}$& \cref{lem:fast_mult}  \\ \hline
        $\mathtt{Sub}[2^n]$ & $\ket{a}\rightarrow \ket{a-2^n}$& Fig.~15 of Ref. \cite{berry2019qubitization} \\
        & & with adder subunits of Ref. \cite{sanders2020compilation} Fig 17b \\ \hline
        $\mathtt{IsEq}$ & $\ket{a}\ket{b}\rightarrow \ket{a}\ket{b}\ket{a=b}$& \cref{lemma:is_eq} \\ \hline
        $\mathtt{cSub}[2^{n+\beta}]$ & $\ket{ctl}\ket{a}\rightarrow\ket{ctl}\ket{a-2^{n+ctl\cdot\beta}}$& \cref{lemma:hybrid_adder}  \\ \hline
    \end{tabular}
    \caption{Summary of algorithmic directives and associated actions used in this work. Only qubit registers where relevant action of the operator is encoded are shown for clarity.}
    \label{tab:circuits}
\end{table}

\subsection{Swap networks} \label{subapp:swap_net}
We start by presenting formal theorems for the swap network technique, elucidating how these can be used for building block-encodings recursively. We expect for this generalization to be of general interest when constructing block-encodings of identical operators acting on different registers.

\subsubsection{One-dimensional swap network}
\begin{lemma}[One-dimensional swap network block-encoding] \label{lemma:swap_1}
Let $B\in L(\mathbb{C}^{2^n})$ and let $U_{B}$ be a unitary that provides an $(\alpha_B,n_B,\epsilon_B)$ block-encoding with associated Toffoli cost of $\mathcal{T}_B$. Let $\bar{B}$ be a weighted sum of applications of $B$ acting on independent subsystems 
    \begin{equation}
       \bar B = \sum_{k=1}^K c_k \left( \mathbb{I}_1\otimes \cdots \otimes  \mathbb{I}_{k-1} \otimes  B \otimes  \mathbb{I}_{k+1} \cdots \otimes  \mathbb{I}_{K}\right)
    \end{equation}
    and assume access to a subroutine  $\textup{PREP}(\vec c)\ket{0} = \sum_k \sqrt{|c_k|/\alpha_c} \ket{k}$ that can be implemented with associated Toffoli cost of  $\mathcal{T}_{\vec c}$, requires $n_c$ ancillas in this implementation and has an associated error $\epsilon_{c}$.
The operator $\bar{B}$ then admits an $(\alpha_B\alpha_c,n_B+n_c,\alpha_c\epsilon_B+\alpha_B\epsilon_{c})$ block-encoding, where $\alpha_c=\sum_k |c_k|$, with an associated Toffoli cost of
    \begin{equation}
          2\mathcal{T}_{\vec c} + \mathcal{T}_B + 2(K-1)(1+n)-4
    \end{equation}
    while using $\ceil{\log_2(K-3)}$ temporary carry ancillas.
\end{lemma}

\begin{proof}
    The (one-dimensional) swap network block-encoding works as a regular \gls{lcu} block-encoding over the $c_k$ coefficients, while its SELECT operator corresponds to a conjugation of the block-encoding for $B$ operator with SWAP operations exchanging registers $1$ and $k$. Note that the structure of this operation allows us to freely move the $\textrm{PREP}_B$ oracle for the block-encoding of $B$ either as part of the swap network's PREPARE or SELECT oracle, having that for this part of the proof we consider it a part of PREPARE.  Given that we can write $\textup{PREP}=\textup{PREP}_B\otimes \textup{PREP}(\vec c)$, where  $\textup{PREP}(\vec c)$ is a generic prepare circuit for the coefficient vector $\vec c$, the associated 1-norm becomes $\alpha_B\alpha_C$, with implementation cost of $\mathcal{T}_{\rm PREP}^{(B)} + \mathcal{T}_{\vec c}$. The SELECT operator then requires one call to the $\textup{SEL}_B$ oracle, alongside two calls to a multiplexed $\textup{SWAP}_{1\leftrightarrow k}$ operator. The latter consists of a unary iteration over $K-1$ indices corresponding to $k=2,3,\cdots,K$, alongside $K-1$ swaps of two $n_B$ qubit registers (noting that for the $k=1$ index this operation becomes an identity). The overall cost of SELECT will then be given by $2(K-1)n_B$ Toffolis coming from the swaps and $2(K-3)$ Toffolis coming from the two (non-controlled) unary iterations each over $K-1$ indices, adding up to $2(K-1)(n_B+1)-4$, while requiring $\ceil{\log_2(K-3)}$ temporary carry ancillas coming from the unary iterations, alongside the cost of $\textrm{SEL}_B$. Combining all of these costs together gives the overall Toffoli and ancilla counts shown in this lemma. For completeness
    Noting that the swap network is effectively a linear combination of block-encodings which all share the same 1-norm $\alpha_B$ and error $\epsilon_B$, we can leverage \cref{lem:be_sum} to recover an upper bound for the swap network error $\alpha_c\epsilon_B + \alpha_B\epsilon_c$.
\end{proof}
We note that the multiplexed SWAPs in this routine can also be implemented using no additional ancilla qubits as done in Ref.~\cite{low_trading_2024}. However, the overall Toffoli cost of this procedure will be usually higher when the associated number of registers to be swapped is not a power of $2$. Since the ancilla overheads are logarithmic with respect to the number of particles, we here chose to use a unary-iteration-based approach. Note that this swap network block-encoding can also be considered as a generalization of the technique used to block-encode the action of a second-quantized one-electron operator that acts identically on both spin sectors by acting on a single spin register that has been conjugated by swapping of both spin registers \cite{loaiza2025majorana}.

\begin{figure}
    \centering
   \begin{quantikz}
		\lstick{$\ket{0}$} & \qwbundle{\ceil{\log_2 K}} & \phantomgate{} & \gate{\textup{PREP}(\vec c)} & \rstick{$\ket{k}$}\\ 
	\end{quantikz} \hspace{2cm}
    \begin{quantikz}
		\lstick{$\ket{\vec 0}$} & \qwbundle{n_B} && \gate{}\wire[d][2]{q} && \\
		\lstick{$\ket{k}$} & \qwbundle{\ceil{\log_2 K}} & \gate{In(k)}\wire[d][1]{q} && \gate{In(k)}\wire[d][1]{q} & \\
		\lstick{$\ket{\psi_1}$} & \qwbundle{n} & \gate[4]{\textup{SWAP}_{1\leftrightarrow k}} & \gate{ B} & \gate[4]{\textup{SWAP}_{1\leftrightarrow k}} &  \\
		\lstick{$\ket{\psi_2}$} & \qwbundle{n} & && &  \\
		\setwiretype{n} \vdots & \phantom{12} & & \vdots \\
		\lstick{$\ket{\psi_K}$} & \qwbundle{n} & && & 
	\end{quantikz}
    \caption{\prep (left) and \sel (right) circuits for one-dimensional Swap network block-encoding discussed in \cref{lemma:swap_1}.  These \gls{lcu} oracles block-encode the operator $\sum_{k=1}^K c_k \left( \mathbb{I}_1\otimes \cdots \otimes  \mathbb{I}_{k-1} \otimes  B \otimes  \mathbb{I}_{k+1} \cdots \otimes  \mathbb{I}_{K}\right)$ where $B$ is its block-encoding oracle, and $\textup{PREP}(\vec c)$ prepares the $c_k$ coefficients. The multiplexed SWAP corresponds to a unary iteration \cite{babbush2018encoding} over index encoded in $\ket{k}$, with the associated control over the $k$-th index swapping all qubits as $\ket{\psi_1}\ket{\psi_k}\rightarrow\ket{\psi_k}\ket{\psi_1}$. Note that the PREPARE circuit used for block-encoding $B$ could also be included as part of the swap network's PREPARE instead.}
    \label{fig:swap_1}
\end{figure}
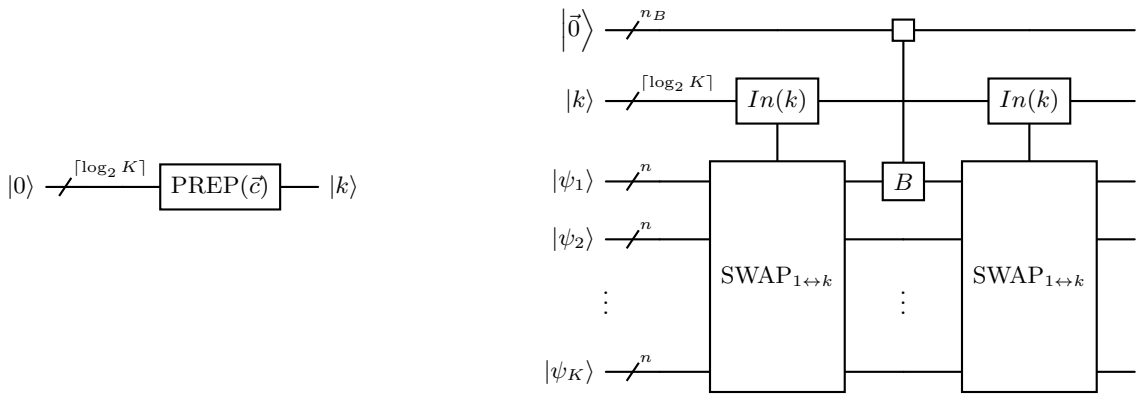

\subsubsection{Two-dimensional swap network}
\begin{lemma}[Two-dimensional swap network block-encoding] \label{lemma:swap_2}
Let $B$ be an operator that acts on two registers, with an $(\alpha_B,n_B,\epsilon_B)$ block-encoding with $\mathcal{T}_{B}$ its associated Toffoli cost. Any operator of the form  
    \begin{equation}
        \sum_{j\neq k=1}^K C_{jk} B_{jk}
    \end{equation}
consisting of a linear combination of $B$ acting simultaneously on $2$ out of $K$ $n$-qubit registers labelled by $j\neq k=1,\cdots,K$, namely $B_{jk}$, then admits an $(\alpha_B\alpha_C,n_B+n_C,\alpha_C\epsilon_B+\alpha_B\epsilon_{C})$ block-encoding, where $\alpha_{C}=\sum_{j\neq k} |C_{jk}|$, with an associated cost of
    \begin{equation}
          2\mathcal{T}_C + \mathcal{T}_{B} + 4(K-1)(1+n)-8
    \end{equation}
Toffolis and using $\ceil{\log_2(K-3)}$ temporary carry ancillas. We have assumed access to an oracle  $\textup{PREP}_{\bm{C}}\ket{0} = \sum_{j\neq k} \sqrt{|C_{jk}|/\alpha_C} \ket{j,k}$ with Toffoli cost of  $\mathcal{T}_C$, requiring $n_C$ ancillas and with associated error $\epsilon_{C}$.
\end{lemma}

\begin{proof}
    The (two-dimensional) swap network block-encoding generalizes the result from the one-dimensional swap network in \cref{lemma:swap_1}. Note that no coefficients corresponding to $j=k$ are prepared in this routine. We now provide a construction for the SELECT circuit (as shown in \cref{fig:swap_2}), corresponding to the following steps:
    \begin{enumerate}
        \item Apply the multiplexed $\textup{SWAP}_{1\leftrightarrow k}$ operator on the $\ket{k}$ register, followed by the same operation in the $\ket{j}$ register. This uses $2(K-1)(1+n_B)-4$ Toffolis and $\ceil{\log_2(K-3)}$ temporary carry ancillas, noting that these can be shared by both multiplexed swaps.
        \item  Apply the operator block-encoding $B$ acting on qubit registers $1$ and $2$ for a cost of $\mathcal{T}_B$.
        \item Uncompute step $1$, for a cost of $2(K-1)(1+n_B)-4$ Toffolis. 
    \end{enumerate}
    Combining all of these costs together gives the overall Toffoli and ancilla counts shown in this lemma. Just as for the one-dimensional swap network, this block-encoding can be considered as a linear combination of block-encodings, each one with a 1-norm $\alpha_B$ and error $\epsilon_B$. Using \cref{lem:be_sum} for linear combination of block-encodings, the total error $\alpha_C\epsilon_B+\alpha_B\epsilon_C$ for the swap network block-encoding is obtained. 
\end{proof}

\begin{figure}
    \centering
   \begin{quantikz}
        \lstick{$\ket{\vec 0}$} & \qwbundle{n_{B}} &&& \gate{}\wire[d][3]{q} &&& \\
		\lstick{$\ket{j}$} & \qwbundle{\ceil{\log_2 K}} & \gate{In(j)}\wire[d][2]{q} &&&& \gate{In(j)}\wire[d][2]{q} & \\
		\lstick{$\ket{k}$} & \qwbundle{\ceil{\log_2 K}} && \gate{In(k)}\wire[d][1]{q} && \gate{In(k)}\wire[d][1]{q} && \\
		\lstick{$\ket{\psi_1}$} & \qwbundle{B} & \gate[4]{\textup{SWAP}_{1\leftrightarrow j}} & \gate[4]{\textup{SWAP}_{2\leftrightarrow k}} & \gate[2]{\hat B} & \gate[4]{\textup{SWAP}_{2\leftrightarrow k}} & \gate[4]{\textup{SWAP}_{1\leftrightarrow j}} & \\
		\lstick{$\ket{\psi_2}$} & \qwbundle{n} & && &&&  \\
		\setwiretype{n} \vdots & \phantom{12} & && \vdots &&& \vdots \\
		\lstick{$\ket{\psi_K}$} & \qwbundle{n} & && & &&
	\end{quantikz}
    \caption{\sel circuit for two-dimensional swap network block-encoding discussed in \cref{lemma:swap_2}.}
    \label{fig:swap_2}
\end{figure}
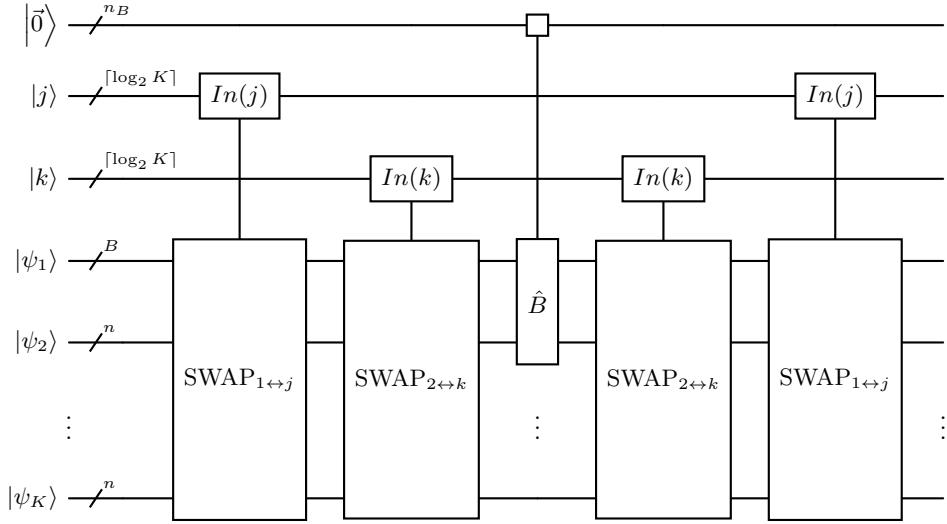

\subsubsection{Swap network with controlled phases} 
For completeness we here detail the implementation for the swap network with controlled phases shown in \cref{fig:full_ham}. We show the detailed circuit in \cref{fig:signed_swap}. The basic idea is to implement the $(-1)$ phase for electronic charges in the Coulomb interaction by adding a Pauli Z gate on the lowermost qubit during the unary iteration sawtooth expansion if the associated index being carried by that branch corresponds to an electron. See Ref. \cite{babbush2018encoding} for a more detailed description on the sawtooth expansion and unary iteration. Since this phase should only be added for the swap network coming for the potential term, the Z gate is controlled on the qubit encoding the control between potential and kinetic energy terms used for block-encoding the full Hamiltonian. Adding the controlled action of electronic phases $(-1)^\sigma$ can thus be done by simply adding controlled Z gates to a regular unary iteration circuit for no additional non-Clifford cost.

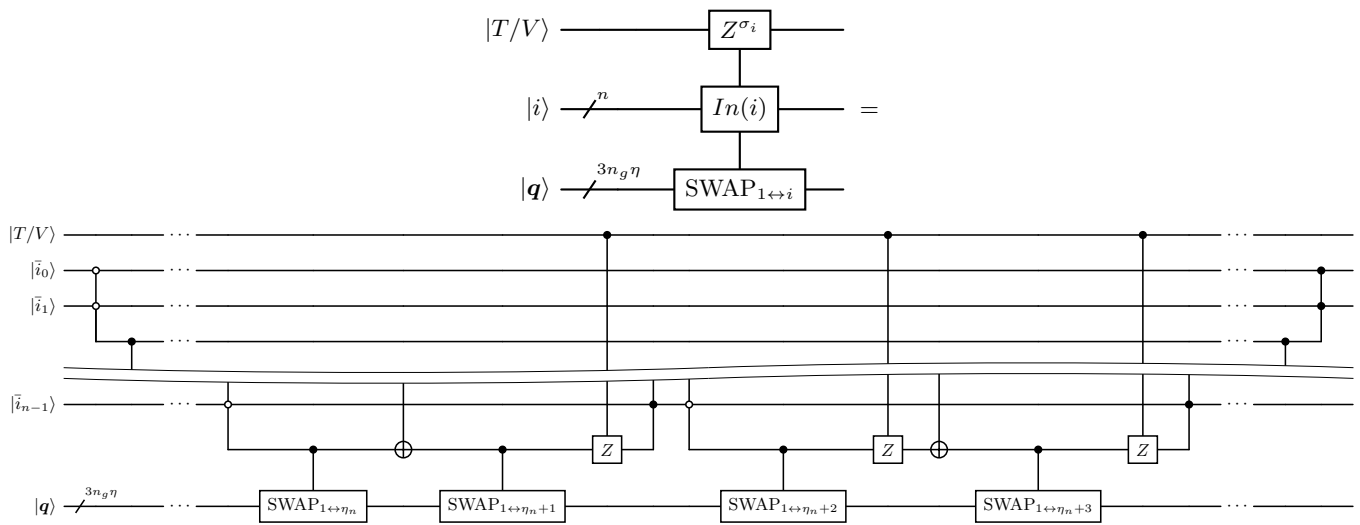
\begin{figure}
    \centering
    \begin{quantikz}
		\lstick{$\ket{T/V}$} & \phantomgate{123} & \gate{Z^{\sigma_i}} \vqw{2} & \\
		\lstick{$\ket{i}$} & \qwbundle{n} & \gate{In(i)} & \\
		\lstick{$\ket{\bm q}$} & \qwbundle{3n_g \eta} & \gate{\text{SWAP}_{1\leftrightarrow i}} &	
	\end{quantikz}
	$=$
    \resizebox{1.0\textwidth}{!}{
\begin{quantikz}[]
	\lstick{$\ket{T/V}$} &&& \ \cdots \  &&&&& \ctrl{6} &&&&\ctrl{6}&&&\ctrl{6}&& \ \cdots \ &&& \\
	\lstick{$\ket{\bar i_{0}}$} & \octrl{2} && \ \cdots \  &&&&&&&&&&&&&& \ \cdots \ && \ctrl{1} & \\
	\lstick{$\ket{\bar i_{1}}$} & \octrl{1} && \ \cdots \  &&&&&&&&&&&&&& \ \cdots \ && \ctrl{1} & \\
	\setwiretype{n} &   & \ctrl{1} \setwiretype{q} & \ \cdots \  &&&&&&&&&&&&&& \ \cdots \ & \ctrl{1} && \setwiretype{n}\\
	\wave&&&&&&&&&&&&&&&&&&&& \\
    \lstick{$\ket{\bar i_{n-1}}$} &  && \ \cdots \ & \octrl{-1} \vqw{1}  &&  &&  & \ctrl{1} \vqw{-1} & \octrl{1} \vqw{-1} &&& &&& \ctrl{1} \vqw{-1} & \ \cdots \ &&& \\
	\setwiretype{n} & &&&&  \ctrl{1} \setwiretype{q} &  \targ{} \vqw{-2} &  \ctrl{1} & \gate{Z} && \setwiretype{n} & \ctrl{1} \setwiretype{q} & \gate{Z} & \targ{} \vqw{-2} & \ctrl{1}  & \gate{Z} && \setwiretype{n}  \\
	\lstick{$\ket{\bm q}$} & \qwbundle{3n_g \eta} && \ \cdots \ && \gate{\text{SWAP}_{1\leftrightarrow \eta_n}} && \gate{\text{SWAP}_{1\leftrightarrow \eta_n+1}} &&&& \gate{\text{SWAP}_{1\leftrightarrow \eta_n+2}} &&& \gate{\text{SWAP}_{1\leftrightarrow \eta_n+3}} &&& \ \cdots\ &&&
\end{quantikz}}
    \caption{Implementation of swap network with controlled phases from \cref{fig:full_ham}. SWAP operations marked with $1\leftrightarrow n$ perform a SWAP between all the qubits of registers $\ket{\vec q_1}$ and $\ket{\vec q_n}$ encoding the system's wavefunction, thus consisting of $3n_g$ SWAPs. Note how operations for electronic indices $(i>\eta_n)$ include a controlled Z in the sawtooth expansion implementing the $(-1)$ phases, while those for nuclear indices $i\leq \eta_n$ perform the standard unary iteration.}
    \label{fig:signed_swap}
\end{figure}

\subsection{Absolute value} \label{subapp:abs}
We now show how to obtain the absolute value of some signed integer $\ket{a}$ in the two's complement representation. This corresponds to the operation
\begin{equation}
    \mathtt{Abs}\ket{a}\ket{0} = \ket{|a|}\ket{\textrm{sign}(a)},
\end{equation}
where one additional qubit is required for making this operation reversible, effectively storing the value of the sign. The circuit for this operation is shown in \cref{fig:abs}, with its costs being summarized via a brief constructive proof in the following lemma.

\begin{lemma}[Absolute value in Two's Complement Representation] \label{lemma:abs}
    The $\mathtt{Abs}$ unitary acting on a $n$-qubit register $\ket{a}$ as described above can be implemented with $n-1$ Toffoli gates, $1$ ancilla, and $n-2$ temporary carry qubits.
\end{lemma}
\begin{proof}
We now present the steps required for implementing this operation, as shown in \cref{fig:abs}. To understand how this works, we note that in general the $n$-bit integer $a$ from the two's complement representation, with associated bits $\{\bar a_i\}_{i=0}^{n-1}$, is encoded as $a=\sum_{i=0}^{n-2} \bar a_i 2^i - \bar a_{n-1}2^{n-1}$. In the case where this is a negative number, namely when the sign bit $\bar a_{n-1}=1$, its absolute value corresponds to $|a|=2^{n-1} - \sum_{i=0}^{n-2} \bar a_i 2^i$. 
\begin{enumerate}
    \item Initialize an additional ancilla qubit, copying the sign of the $\ket{a}$ register with a CNOT.
    \item Controlled on this ancilla qubit, flip all qubits encoding the $\ket{a}$ state. The qubit previously encoding the sign of $a$ will always be returned to $0$ after this operation. In the case where this flip occurs, the remaining $n-1$ qubits will correspond to their one's complement, namely $(2^{n-1}-1)- \sum_{i=0}^{n-2} \bar a_i 2^i = |a|-1$.
    \item If the ancilla qubit carrying the sign information is $1$, then add $1$ to the $(n-1)$-qubit register, thus encoding the value of $|a|$. The qubit that had been returned to $0$ can be used to store the overflow from this operation, associated with the case where we are adding $1$ to the $\ket{11\cdots 1}$ $(n-1)$-qubit state, which was produced for the most negative number $a=-2^{n-1}$. This is done by performing an addition of the ancilla qubit on the $n$-qubit register, requiring $n-1$ Toffolis and $n-2$ temporary carry ancillas \cite{gidney_halving_2018}. One less temporary carry ancilla was needed here since we can directly use the $n$th qubit entering as $0$ instead.
\end{enumerate}
Note that these operations will act as identity for the case where the input number is positive, with the ancilla qubit being returned as $0$: the first $n-1$ qubits from the $\ket{a}$ register are here already encoding the absolute value $\ket{|a|}$. For the case of a negative integer, the above shows how the absolute value is returned by adding $1$ to the one's complement of the first $n-1$ bits.

\end{proof}

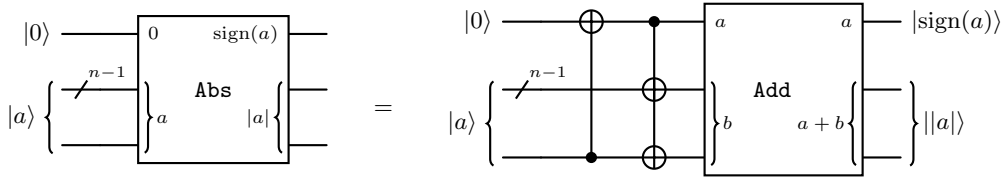
\begin{figure}
    \centering
    \begin{quantikz}
		\lstick{$\ket{0}$} && \gate[3]{\hspace{0.6cm}\mathtt{Abs}\hspace{0.6cm}} \gateinput{$0$} \gateoutput{$\textrm{sign}(a)$} & \\
		\lstick[2]{$\ket{a}$} & \qwbundle{n-1} & \gateinput[2]{$a$} \gateoutput[2]{$|a|$} & \\
		&& & \\
	\end{quantikz}
	\hspace{0.5cm}$=$\hspace{0.5cm}
	\begin{quantikz}
		\lstick{$\ket{0}$} && \targ{} & \ctrl{2} & \gate[3]{\hspace{0.5cm}  \mathtt{Add} \hspace{0.8cm}} \gateinput{$a$} \gateoutput{$a$} & \rstick{$\ket{\textrm{sign}(a)}$} \\
		\lstick[2]{$\ket{a}$} & \qwbundle{n-1} & & \targ{} & \gateinput[2]{$b$} \gateoutput[2]{$a+b$} & \rstick[2]{$\ket{|a|}$}  \\
		& & \ctrl{-2} & \targ{} && \\
	\end{quantikz}
    \caption{Circuit for in-place calculation of absolute value. CNOT gate on multi-qubit register corresponds to applying a CNOT on each of the register's qubits. Bottom qubit for $\ket{a}$ input corresponds to sign qubit from two's complement representation. Note how it only becomes $1$ in the output for overflow case of most negative number $a=-2^{n-1}$.}
    \label{fig:abs}
\end{figure}

\subsection{Absolute value of difference}
We now show the circuit for calculating in-place the absolute value between the difference between two signed integers $q$ and $s$ encoded in two associated $n$-bit registers in the two's complement representation. This translates to performing the following operation:
\begin{equation}
    \mathtt{AbsDiff} \ket{a}\ket{b}\ket{0} = \ket{|a-b|}\ket{b}\ket{\textrm{sign}(a-b)},
\end{equation}
where we note that this operation requires one additional qubit that needs to be passed to its adjoint operation for uncomputation, besides temporary carry ancillas. A notable feature of this operation is that when computing the absolute value of the difference of two $n$-bit signed two's complement integers, storing the absolute (unsigned) value of their difference can be achieved in $n$ bits. This is proven in the following proposition. 

\begin{proposition} \label{prop:abs_diff}
    Let $a$ and $b$ be signed $n$-bit integers represented by two's complement. If an unsigned representation is used to represent $c=|a-b|$, then only $n$ bits are needed to exactly represent $c$, even when accounting for overflow.
\end{proposition}
\begin{proof}
    In two's complement, $a,b \in [-2^{n-1}, 2^{n-1}-1]$. Taking the extremal values for their difference we get $\min(a-b) = -2^n+1$ and $\max(a-b) = 2^n-1$. So $|a-b|\in[0,2^n-1]$. This means that $a-b$ can be encoded in two's complement representation using $n+1$ bits, while its absolute value is encoded in an unsigned representation with $n$ bits (ignoring the leading sign bit which is guaranteed to be 0 as $a-b$ can never attain the most negative value $ -2^{n}$).
\end{proof}

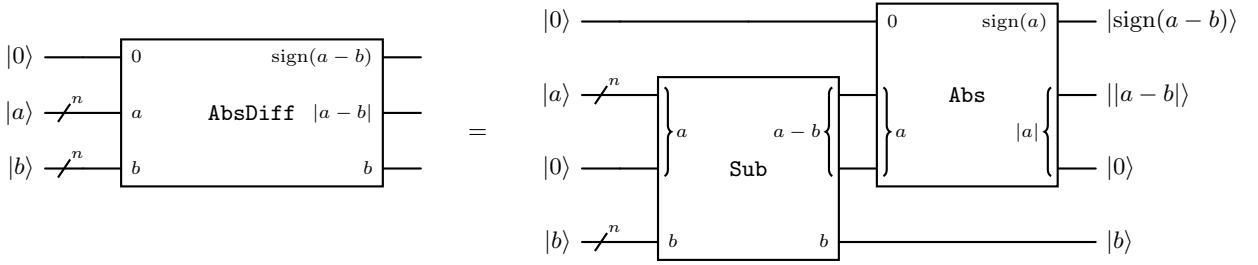
\begin{figure}
    \centering
    \begin{quantikz}
		\lstick{$\ket{0}$} && \gate[3]{\hspace{1.0cm}\mathtt{AbsDiff}\hspace{1.0cm}} \gateinput{$0$} \gateoutput{$\textrm{sign}(a-b)$} & \\
		\lstick{$\ket{a}$} & \qwbundle{n} & \gateinput{$a$} \gateoutput{$|a-b|$} & \\
		\lstick{$\ket{b}$} & \qwbundle{n} & \gateinput{$b$} \gateoutput{$b$} &\\
	\end{quantikz}
	\hspace{0.5cm}$=$\hspace{0.5cm}
	\begin{quantikz}
		\lstick{$\ket{0}$} &&& \gate[3]{\hspace{0.8cm} \mathtt{Abs} \hspace{0.8cm}} \gateinput{$0$} \gateoutput{$\textrm{sign}(a)$} & \rstick{$\ket{\textrm{sign}(a-b)}$} \\
		\lstick{$\ket{a}$} & \qwbundle{n} & \gate[3]{\hspace{0.8cm}\mathtt{Sub}\hspace{0.8cm}} \gateinput[2]{$a$} \gateoutput[2]{$a-b$} & \gateinput[2]{$a$} \gateoutput[2]{$|a|$} & \rstick{$\ket{|a-b|}$}  \\
		\lstick{$\ket{0}$} && && \rstick{$\ket{0}$} \\
		\lstick{$\ket{b}$} & \qwbundle{n} & \gateinput{$b$} \gateoutput{$b$} && \rstick{$\ket{b}$}
	\end{quantikz}
    \caption{Circuit for computing the absolute value of the difference of two two's complement registers, as discussed in \cref{lem:abs_diff}. For notational simplicity we here considered the single qubit register coming out of the subtraction operation to be carrying the sign information of the associated $(n+1)$-qubit register encoding $a-b$. The operation marked as $\mathtt{Abs}$ computes the absolute value of a two's complement register (see \cref{lemma:abs}). Note that as proved in \cref{prop:abs_diff}, the overflow qubit from the absolute value routine will always be returned in $0$ state, which we consider a temporary carry ancilla and do not explicitly include on the $\mathtt{AbsDiff}$ routine on the left.}
    \label{fig:abs_diff}
\end{figure}

The implementation of this $\mathtt{AbsDiff}$ operation is presented via a brief constructive proof in the following lemma.

\begin{lemma} \label{lem:abs_diff}
    The $\mathtt{AbsDiff}$ unitary acting on two $n$-qubit registers $\ket{a} \ket b$ as described above can be implemented with $2n$ Toffoli gates, 1 ancilla, and $n+1$ temporary carry qubits.
\end{lemma}
\begin{proof}
We now present the steps required for implementing this operation, as shown in \cref{fig:abs_diff}.
\begin{enumerate}
    \item Perform in-place subtraction of $\ket{a}$ in the $\ket{b}$ register, while initializing one ancilla to account for overflow. This uses $n$ Toffolis, $n$ temporary carry qubits, and $1$ ancilla qubit to store the overflow. 
    \item Perform in-place operation for obtaining absolute value (\cref{lemma:abs}) over the $(n+1)$-qubit register encoding $a-b$. The Toffoli cost of this operation is $n$, while requiring $n$ temporary carry ancillas and one additional ancilla. Note that after this operation, as show in \cref{prop:abs_diff}, the qubit used for overflow in the $(n+1)$-qubit register (corresponding to the overflow qubit for the absolute value routine) will always be returned to $0$, which can thus instead be considered as a temporary carry ancilla for the $\mathtt{AbsDiff}$ routine. 
\end{enumerate}
\end{proof}

Note that this can easily be repeated over other Cartesian components with the carry qubits being reused when implementing the $\mathtt{AbsDiff}^{\otimes 3}$ operation in \cref{fig:U_V}.

\subsection{Squaring block-encoding} \label{subapp:squaring}
We now discuss how to implement the block-encoding for the squaring operation $\ket{a}\rightarrow a^2 \ket{a}$. We first present how to block-encode the amplitude of a two's complement register, namely $\ket{a}\rightarrow a\ket{a}$. Although this operation could be used twice in a row to implement the square, we here present a qubitization-based implementation which instead implements $\ket{a}\rightarrow T_2(a)\ket{a}=(2a^2-1)\ket{a}$, effectively reducing the 1-norm by a factor of $2$ while introducing a constant shift that can be generally neglected.

\subsubsection{Amplitude encoding}
We start by presenting how to block-encode the amplitude $\ket{a}\rightarrow a\ket{a}$, with the two associated oracles $\textup{PREP}_{\rm amp}$ and $\textup{SEL}_{\rm amp}$ (\cref{fig:amp_prep,fig:amp_sel}). We here consider for the integer $a$ to be encoded using $n$ bits in the two's complement encoding, with its amplitude being expressed as $a=\sum_{b=0}^{n-2} \bar a_b 2^b - \bar a_{n-1}2^{n-1}$, where $\bar a_i\in\{0,1\}$ is the $i$th bit used for encoding $a$. To understand how our block-encoding works, we first consider an unsigned $(n-1)$-bit integer $c$. We then exploit a property of its one's complement $c^*$, namely the associated value we obtain by flipping the values of all of its bits ($\bar c_i=0(1) \leftrightarrow \bar c_i^*=1(0)$):
\begin{equation} \label{eq:ones_comp}
    c = \sum_{b=0}^{n-2} \bar c_i 2^b = 2^{n-1} - 1 - \sum_{b=0}^{n-2} \bar c_i^*2^b = 2^{n-1}-c^* - 1.
\end{equation}
This property is particularly useful as it allows for the calculation of the absolute amplitude of a two's complement integer. To better understand how this works, consider the two following cases:
\begin{enumerate}
    \item The integer $a\geq 0$, meaning its sign qubit is $0$. In this case, the first $n-1$ qubits can be considered as an unsigned $(n-1)$-qubit representation of a positive integer already. 
    \item The integer $a<0$. In this case, it follows from Eq.~\eqref{eq:ones_comp} that $|a|$ can be recovered by doing the one's complement of the first $n-1$ qubits, obtaining the associated amplitude, and adding $1$.
\end{enumerate}

The basic idea for performing this block-encoding is to encode the bit-wise amplitudes through the state
\begin{equation}
    \ket{\sqrt{\mathtt{amp}_n}} := \frac{1}{2^{(n-1)/2}}\left(\ket{0}^{n-1}\ket{1}_s+\sum_{b=0}^{n-2}2^{b/2}X_b\ket{0}\ket{0}_s\right),
\end{equation}
where $X_b$ is a Pauli X acting on qubit $b$, having that this can be thought as a one-hot encoding for the amplitudes of the first $n-1$ bits. Each of these amplitudes is then added constructively (or cancelled destructively) if the associated bit in $\ket{a}$ is $1$ (or $0$). This can be achieved by using an alternating sign register with a single qubit, namely a qubit which is prepared in the $\ket{+}$ state and on which we act with identity (or a Pauli Z) for having constructive (destructive) interference. The $n$th bit, which corresponds to the sign qubit in the two's complement representation, will then be used to control the three following operations: flip to one's complement representation over first $n-1$ qubits, addition of $+1$ value for negative sign branch (encoded through $\ket{\cdot}_s$ qubit), and global multiplication by $-1$ for case where amplitude is negative. We now formally describe this operation and associated costs.

\begin{lemma}[Controlled Amplitude Block-Encoding in Two's Complement Representation]\label{lemma:amp}
    For a given $n$-qubit register encoding a signed integer $\ket{a}$ in two's complement encoding, there exists a $(2^{n-1},n+1,0)$ block encoding $O_{\rm amp}$ of the controlled operator that encodes its amplitude $a\ket{a}$ using
    \begin{equation}
        4n-3
    \end{equation}
    Toffoli gates and $2$ temporary carry qubits.
\end{lemma}
\begin{proof}
 This proof is carried out by constructing the PREPARE and SELECT oracles, as illustrated in \cref{fig:amp_prep,fig:amp_sel}. We now start by describing the PREPARE routine, which corresponds to the operation
\begin{equation}
    \textup{PREP}_{\rm amp} \ket{0} = \ket{+}\ket{\sqrt{\mathtt{amp}_n}} .
\end{equation}
We now provide a detailed explanation of how this is implemented. First we show how to prepare $\ket{\sqrt{\mathtt{amp}_n}}$ following the strategy in Ref. \cite{su2021fault}, and then show how an associated \gls{lcu} block-encodes the amplitude of $a$ in the two's complement representation. 

Preparation of $\ket{\sqrt{\mathtt{amp}_n}}$ involves a series of cascading controlled Hadamard gates, each controlled on the previous single qubit state. This procedure iteratively halves the amplitude on the prior quantum state, effectively yielding the desired amplitudes over $b$th bits in the target $(n-1)$-qubit register. The preparation of this state works as follows:
\begin{align}
    \ket{0}^{n-1} \xrightarrow[]{\mathtt{Had}_0}&\frac{1}{\sqrt 2}(\ket0 +\ket1)\ket{0}^{n-2} \\
    \xrightarrow[]{\text{ctrl}_0 : \mathtt{Had}_1}& \left (\frac{1}{\sqrt 2} \ket{00} + \frac{1}{2}\ket 1 (\ket 0 + \ket 1) \right )\ket{0}^{n-3} \\
    \xrightarrow[]{\text{ctrl}_1 : \mathtt{Had}_2}& \left (\frac{1}{\sqrt 2} \ket{000} + \frac{1}{2}\ket{100} + \frac{1}{2\sqrt 2} \ket{11}(\ket 0 + \ket 1) \right )\ket{0}^{n-4} \\
    \xrightarrow[]{\text{ctrl}_2 : \mathtt{Had}_3}& \left (2^{-1/2} \ket{0000} + 2^{-2/2} \ket{1000} + 2^{-3/2}\ket{1100} + 2^{-4/2} \ket{111}(\ket 0 + \ket 1) \right )\ket{0}^{n-5} \\
    \xrightarrow[]{\: \: ...\: \: } & \sum_{b=0}^{n-2} 2^{-(b+1)/2} \ket{1}^{b}\ket{0}^{n-2-b},
\end{align}
where the last line above can be obtained by completing the cascading controlled Hadamards. Interpreting the state as a unary encoding we effectively have the state $\sum_{b=0}^{n-2} 2^{-(b+1)/2} \ket{b} + 2^{-(n-1)/2} \ket{1}^{n-2}(\ket{0} + \ket{1})$ where the endianness is chosen such that $\ket{1} = \ket{1000...}$, $\ket{2} = \ket{1100...}$ and so on. This normalized state can be re-written
\begin{equation}
    2^{-(n-1)/2} \left (\sum_{b=0}^{n-3} 2^{(n-b-1)/2}\ket b + \ket{1}^{n-2}(\ket 0 + \ket 1) \right). 
\end{equation}
Note that ordering of the amplitudes is opposite what we desire such that the smallest unary encoded integer appears with the greatest amplitude. To give yield the complementary encoding, we flip all of the bits to obtain
\begin{align}\label{eq:pre_cnot}
    2^{-(n-1)/2} \left (\sum_{b=0}^{n-3} 2^{(b+1)/2}\ket{0}^{n-b-3}\ket{1}^{b+2} + \ket{0}^{n-2}(\ket 0 + \ket 1) \right) 
\end{align} 
Next, we use an additional qubit $\ket{\cdot}_s$ to flag for the all-zeros state. Noting that this is the only state for which the last qubit is in $0$, this can be simply done by applying an open-controlled NOT on this last qubit, obtaining the state
\begin{align}\label{eq:pre_cnot}
    2^{-(n-1)/2} \left (\sum_{b=0}^{n-3} 2^{(b+1)/2}\ket{0}^{n-b-3}\ket{1}^{b+2}\ket{0}_s + \ket{0}^{n-2}(\ket{0}\ket{1}_s + \ket{1}\ket{0}_s) \right).
\end{align} 
Next, to avoid checking multiple bits per operation, we desire a one-hot encoding. This can be achieved by applying a cascade of CNOTs such that they act on the $k$th qubit and are controlled on the $(k-1)$th like so: $\rm{ctrl}_{k-1} : X_{k}$. This sequence starts with $k=n$, and progresses by decreasing $k$ by 1, i.e. start on the final qubit pair in the sequence $\rm{ctrl}_{n-1} : X_{n}$ and progress to the last (not considering the $\ket{\cdot}_s$ register). This yields the state:
\begin{align}
    & 2^{-(n-1)/2} \left (\sum_{b=0}^{n-3} 2^{(b+1)/2}\ket{0}^{n-b-3}\ket{1} \ket{0}^{b+1}\ket{0}_s + \ket{0}^{n-2}(\ket 0\ket{1}_s + \ket 1\ket{0}_s) \right) \\
    = & 2^{-(n-1)/2} \left (\sum_{b=0}^{n-2} 2^{b/2}\ket{0}^{n-b-2}\ket{1}\ket{0}^{b}\ket{0}_s + \ket{0}^{n-1}\ket{1}_s \right) \\
    = & 2^{-(n-1)/2} \left (\sum_{b=0}^{n-2} 2^{b/2}\ket{b}\ket{0}_s + \ket{0}^{n-1}\ket{1}_s \right) := \ket{\sqrt{\mathtt{amp}_n}}, 
\end{align}
where in the last line, we reverse the endianness of the register to obtain our one-hot encoding. Here $\ket{0} = \ket{...0001}$, $\ket{1} = \ket{...0010}$, $\ket{2} = \ket{...0100}$ and so on. Zero is interpreted in this way since it is associated to the amplitude $2^0$ and needs to be included in the subsequent control logic, while the all-zeros state flagged by $\ket{1}_s$ will be used for adding the $+1$ amplitude for the case where we are in the negative branch of the sign qubit. Finally, we note that our PREPARE circuit will use one additional qubit in an equal superposition, namely $\ket{+}=\mathtt{Had}\ket{0}$ for performing destructive interference through a Z gate and cancelling amplitudes associated with a bit-value of $0$. Considering the fact that each controlled Hadamard gate can be implemented with a single Toffoli (see Fig. 17 from Ref.~\cite{lee2021even}), and that these are the only non-Clifford operations in this procedure, this PREPARE oracle can then be implemented using $n-2$ Toffoli gates. The associated 1-norm directly corresponds to $2^{n-1}$. Having presented how to implement the PREPARE operation, we now detail the operations of SELECT and derive the total cost of the block-encoding.

\begin{enumerate}
    \item Apply a Z gate on the sign bit of the incoming $n$ qubit register $\ket{a}$. This operation will give an overall phase of $-1$ for negative values, which in tandem with the subsequent implementation of the absolute amplitude $|a|$ will return the signed amplitude $a$.
    \item     Controlled on the sign bit of the incoming $n$ qubit register $\ket a$, compute the one's complement of the other $n-1$ qubits in-place. This can be done by applying a CNOT on each of the $n-1$ qubits controlled on the sign bit for no Toffoli cost.
    \item Controlled on $\ket{b}$ copy the $b$-th bit of the state in the top register into an ancilla ``flag'' qubit. Since $\ket{b}$ is one-hot encoded, this can be done with a Toffoli applied to each single qubit register for a cost of $n - 1$ Toffolis.
    \item For each bit of the state which has a value of $0$, cancel the associated amplitude by applying a Z gate on the $\ket{+}$ register. This can be done by implementing a Z on the $\ket{+}$ register that is controlled on the flag qubit being $0$. Note that this operation should also be controlled on the ``global'' control qubit for the entire block-encoding, for a total cost of $1$ Toffoli gate.
    \item Controlled on the $\ket{\cdot}_s$ qubit flagging the all-zeros state, namely $\ket{1}_s$, cancel the amplitude for the case where the integer $a$ is positive. Note that the previous operations will already have effectively implemented a $Z$ gate on the $\ket{+}$ register. This operation needs to be undone $a$ is a negative number. This can be done by implementing a Toffoli gate on the $\ket{+}$ register, controlled on the sign qubit of $\ket{a}$ being 1 and the all-zeros flag being 1. Controlling this operation with the global control can be done by adding one logical AND operation, increasing the cost to 2 Toffolis and using one temporary carry ancilla.
    \item Uncompute steps 2 and 3, returning the flag qubit to the $\ket{0}$ state and the $\ket{a}$ register to its regular two's complement encoding. This step requires $n-1$ Toffoli gates.
    \item Collecting the costs, this controlled operation requires $2\times(n-1)+3=2n+1$ Toffoli gates and $2$ temporary carry ancillas.
\end{enumerate}

Summing the costs of one call to PREPARE, one call to its hermitian conjugate, and one call to SELECT, we recover the costs outlined in this lemma.
\end{proof}

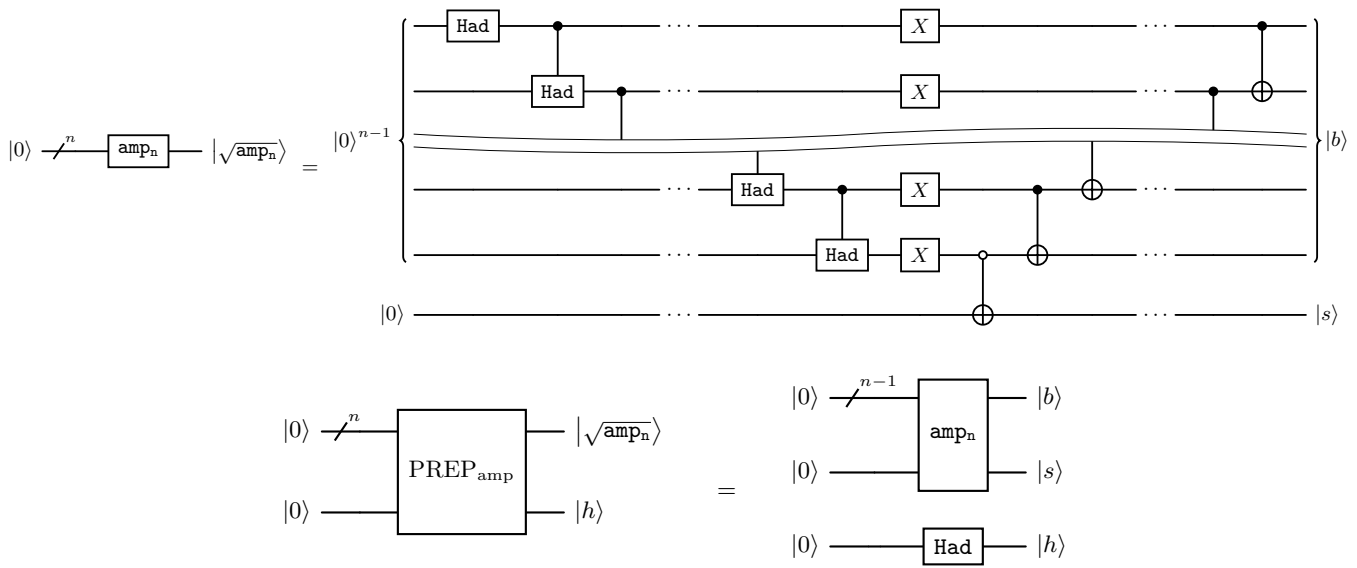
\begin{figure}
    \centering
    \resizebox{1.0\textwidth}{!}{
    \begin{quantikz}
	\lstick{$\ket{0}$} & \qwbundle{n} & \gate{\mathtt{amp_n}} & \rstick{$\ket{\sqrt{\mathtt{amp_n}}}$} \\
\end{quantikz}
$=$
\begin{quantikz}
	\lstick[5]{$\ket{0}^{n-1}$} & \gate{\mathtt{Had}} & \ctrl{1} && \ \cdots \ & && \gate{X} &  &&& \ \cdots\ && \ctrl{1} & \rstick[5]{$\ket{b}$} \\
	&  & \gate{\mathtt{Had}} & \ctrl{1} & \ \cdots \ & && \gate{X} &  &&& \ \cdots\ & \ctrl{1} & \targ{} & \\
	\wave &&&&&&&&&&&&&& \\
	&  & && \ \cdots \ & \gate{\mathtt{Had}} \vqw{-1} & \ctrl{1} & \gate{X} &  & \ctrl{1} & \targ{} \vqw{-1} & \ \cdots\ &&& \\
	&  & && \ \cdots \ &  & \gate{\mathtt{Had}} & \gate{X} &  \octrl{1} & \targ{} && \ \cdots\  & &&  \\
	\lstick{$\ket{0}$} &  & && \ \cdots \ &  &  && \targ{} &  &&\ \cdots \ &&& \rstick{$\ket{s}$}
\end{quantikz}} \vspace{0.5cm}
    \\
    \begin{quantikz}
    \lstick{$\ket{0}$} & \qwbundle{n} & \gate[2]{\textrm{PREP}_{\rm amp}} & \rstick{$\ket{\sqrt{\mathtt{amp_n}}}$} \\
	\lstick{$\ket{0}$} & & & \rstick{$\ket{h}$} \\
\end{quantikz}
\hspace{0.5cm}$=$\hspace{0.5cm}
\begin{quantikz}
   \lstick{$\ket{0}$} & \qwbundle{n-1} & \gate[2]{\mathtt{amp_{n}}} & \rstick{$\ket{b}$} \\
	\lstick{$\ket{0}$} & &  & \rstick{$\ket{s}$}  \\
	\lstick{$\ket{0}$} & \phantomgate{1} & \gate{\mathtt{Had}} &   \rstick{$\ket{h}$}  \\
\end{quantikz}
    \caption{Preparation of $\ket{\sqrt{\mathtt{amp_n}}}$ state (top), and implementation of $\textrm{PREP}_{\rm amp}$ for block-encoding of amplitude $\ket{a}\rightarrow a\ket{a}$ (\cref{lemma:amp}).}
    \label{fig:amp_prep}
\end{figure}

\begin{figure}
    \centering
    \begin{quantikz}
 	\lstick{$\ket{b}$} & \qwbundle{n} & \phase{b} \vqw{2} & \\
 	\lstick{$\ket{0}$} && \targ{} & \\
 	\lstick{$\ket{a}$} & \qwbundle{n} & \phase{b} &
 \end{quantikz} \hspace{0.5cm}$=$\hspace{0.5cm}
 \begin{quantikz}
 	\lstick{$\ket{\bar b_0}$} && \ctrl{5} && \ \cdots\ && \\
 	\lstick{$\ket{\bar b_1}$} &&  & \ctrl{5} & \ \cdots\ && \\
 	\setwiretype{n} & \vdots &&&&&&&&&&&&& \\
 	\lstick{$\ket{\bar b_{n-1}}$} &&  && \ \cdots\ & \ctrl{5} & \\
 	\lstick{$\ket{0}$} && \targ{} & \targ{} & \ \cdots\ & \targ{} & \\
 	\lstick{$\ket{\bar{a}_0}$} && \ctrl{0} && \ \cdots\ && \\
 	\lstick{$\ket{\bar{a}_1}$} &&  & \ctrl{0} & \ \cdots\ && \\
 	\setwiretype{n} & \vdots &&&&&&&&&&&&& \\
 	\lstick{$\ket{\bar{a}_{n-1}}$} &&  && \ \cdots\ & \ctrl{0} & \\
 \end{quantikz} \\
 	\begin{quantikz}
 		\lstick{$\ket{ctl}$} && \ctrl{2} & \\
 		\lstick{$\ket{s}$} && \gate{s} & \\
 		\lstick{$\ket{b}$} & \qwbundle{n-1} & \gate{b} \vqw{2} & \\ 
 		\lstick{$\ket{h}$} && \gate{h} & \\
 		\lstick{$\ket{a}$} & \qwbundle{n} & \gate{\textrm{SEL}_{\rm amp}} &
 	\end{quantikz} \hspace{0.5cm}$=$\hspace{0.5cm}
 	\begin{quantikz}
 		\lstick{$\ket{ctl}$} &&&& \ctrl{3} & \ctrl{3} &&& \\
 		\lstick{$\ket{s}$} &&&&  & \ctrl{0} &&& \\
 		\lstick{$\ket{b}$} & \qwbundle{n-1} && \phase{b}\vqw{4} &&& \phase{b}\vqw{4} && \\
 		\lstick{$\ket{h}$} &&&& \gate{Z} & \gate{Z} &&& \\
 		\lstick{$\ket{0}$} &&& \targ{} & \octrl{-1} && \targ{} && \rstick{$\ket{0}$} \\
 		\lstick{$\ket{\bar a_{n-1}}$} & \gate{Z} & \ctrl{1} &  & & \ctrl{-2} &  & \ctrl{1} & \\
 		\lstick{$\ket{\bar a_{0},\cdots,\bar a_{n-2}}$} & \qwbundle{n-1} & \targ{} & \phase{b} & && \phase{b}  & \targ{} & \\
 	\end{quantikz}
    \caption{Implementation of Toffoli cascade (top) and controlled $\textrm{SEL}_{\rm amp}$ (bottom) routines for controlled block-encoding of amplitude $\ket{a}\rightarrow a\ket{a}$, as discussed in \cref{lemma:amp}.}
    \label{fig:amp_sel}
\end{figure}
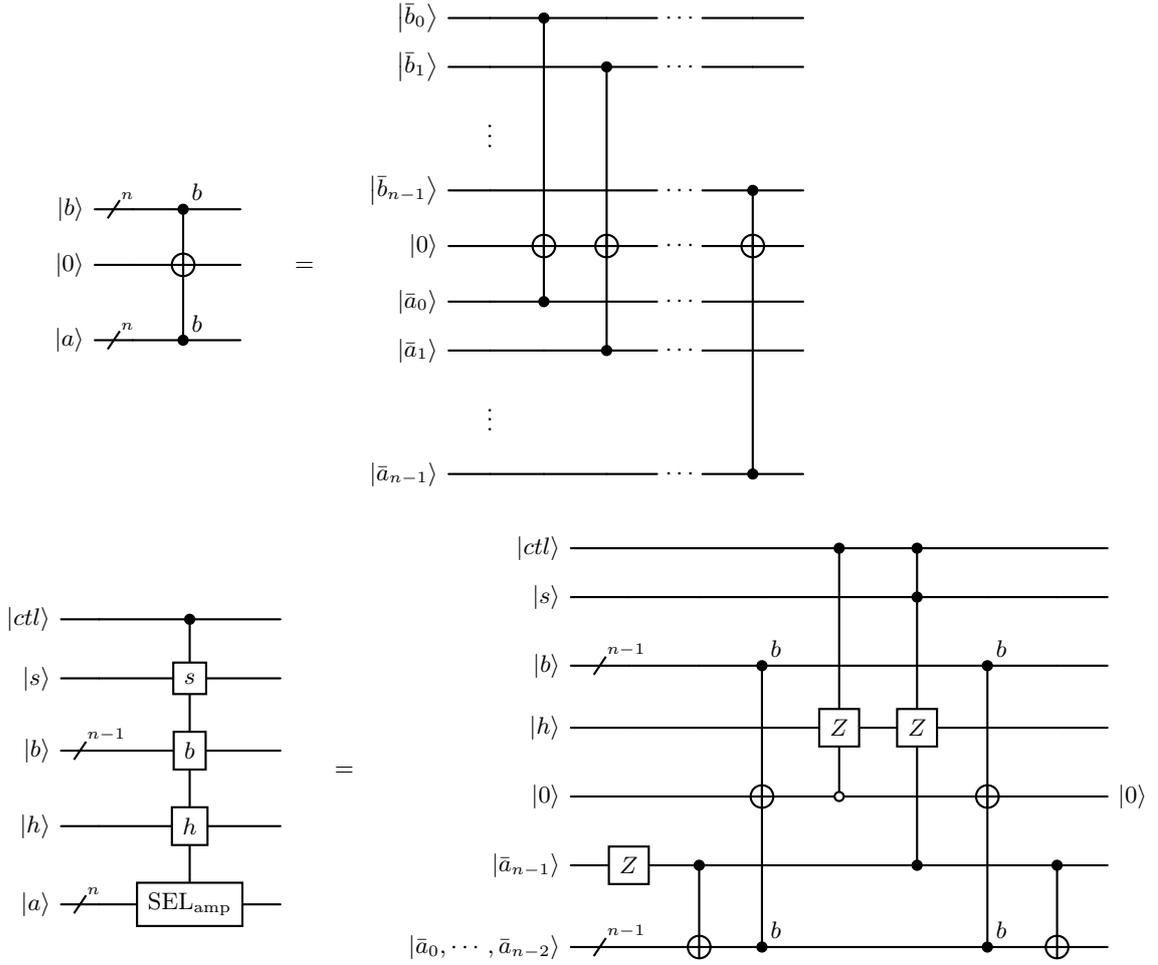

\subsubsection{Walk-based implementation}
We now present an implementation of the squaring operation that uses a qubitization approach to implement it as a second degree Chebyshev polynomial $\ket{a}\rightarrow T_2(a)\ket{a}=(2a^2-1)\ket{a}$. The benefit of this approach is that it reduces the 1-norm by a factor of 2, since it implements $\ket{a}\rightarrow 2a^2\ket{a}$ up to a constant shift. This improvement comes from effectively shifting the spectrum with a constant as $[0,a_{\rm max}^2]\rightarrow [-a^2_{\rm max}/2, a^2_{\rm max}/2]$. This operation is built by using the block-encoding oracles for the amplitude encoding $\ket{a}\rightarrow a\ket{a}$ from \cref{lemma:amp}, which are then used to build an associated qubitized walk operator \cite{low2019hamiltonian,linlin_notes} and applied two times, recovering the target Chebyshev polynomial. 

\begin{lemma}[Walk-based Controlled Block Encoding of Square in Two's Complement Representation] \label{lemma:walk_square}
    Given an $n$-qubit register encoding a signed integer $\ket{a}$ in the two's complement encoding, there exists a $(2^{2(n-1)},n+1,0)$ block-encoding of the controlled operator that computes $T_2(a)\ket{a}=(2a^2-1)\ket{a}$ using
    \begin{equation}
        10n-6\in\mathcal{O}(n)
    \end{equation}
    Toffoli gates and $n-1$ temporary carry qubits.
\end{lemma}
\begin{proof}
    Using the PREPARE and SELECT oracles presented in \cref{lemma:amp} implementing $O_{\rm amp}\ket{a}\rightarrow a\ket{a}$, it follows immediately from the qubitized walk operator construction \cite{low2019hamiltonian,linlin_notes} that the construction in \cref{fig:walk_square} block-encodes the second-degree Chebyshev polynomial $\ket{a}\rightarrow T_2(a)\ket{a}=(2a^2-1)\ket{a}$. The costs for this procedure then become twice the Toffoli costs from the block-encoding $O_{\rm abs}$ (i.e. $2\times (4n-3)$), plus the cost of the reflection logic for building the walk operator, corresponding to two (separate) multi-controlled $Z$ each over $n+1$ control qubits (where each can be implemented via temporary ANDs using $n-1$ temporary carry qubits and $n$ Toffolis). Note that the temporary carry qubits for is dictated by those coming from the temporary ANDs, namely $n-1$. Adding up the Toffoli counts, we recover the costs shown in this lemma.
\end{proof}

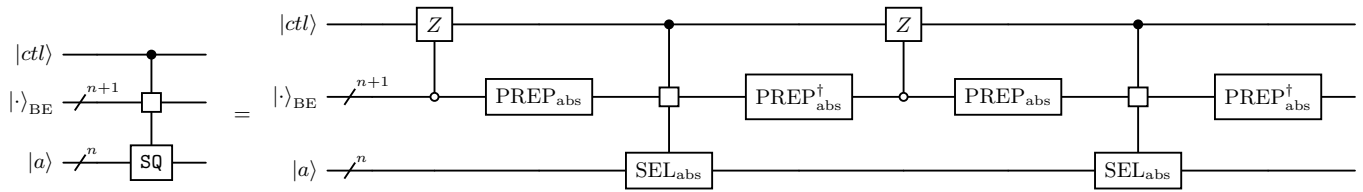
\begin{figure}
    \centering
    \resizebox{1.0\textwidth}{!}{
    \begin{quantikz}
	\lstick{$\ket{ctl}$} && \ctrl{2} & \\
	\lstick{$\ket{\cdot}_{\rm BE}$} & \qwbundle{n+1} & \gate{} & \\
	\lstick{$\ket{a}$} & \qwbundle{n} & \gate{\mathtt{SQ}} &
\end{quantikz} \hspace{0.2cm}$=$\hspace{0.1cm}
\begin{quantikz}
	\lstick{$\ket{ctl}$} & \phantomgate{12} &\gate{Z} &  & \ctrl{2} & & \gate{Z} && \ctrl{2} && \\
	\lstick{$\ket{\cdot}_{\rm BE}$} & \qwbundle{n+1} & \octrl{-1} & \gate{\textup{PREP}_{\rm abs}} & \gate{} & \gate{\textup{PREP}_{\rm abs}^\dagger} & \octrl{-1} & \gate{\textup{PREP}_{\rm abs}} & \gate{} & \gate{\textup{PREP}_{\rm abs}^\dagger} & \\
	\lstick{$\ket{a}$} & \qwbundle{n} &  && \gate{\textup{SEL}_{\rm abs}} & & &  & \gate{\textup{SEL}_{\rm abs}} &  & \\
\end{quantikz}}
    \caption{Walk-based implementation for controlled block-encoding of square, up to a constant shift (see \cref{lemma:walk_square}). Block-encoding register $\ket{\cdot}_{\rm BE}$ encodes $\ket{s}\ket{b}\ket{h}$ registers for block-encoding amplitude from \cref{fig:amp_prep}.}
    \label{fig:walk_square}
\end{figure}

\subsection{Improved multiplication} 
We now discuss our improved approach for out-of-place multiplication of two positive integers, which acts as
\begin{equation}
    \mathtt{Mult}\ket{a}\ket{b}\ket{0} = \ket{a}\ket{b}\ket{ab}.
\end{equation}
One of the most expensive parts of block-encoding the potential operator $V$ using the methods in this work is the multiplication step that occurs in the inequality test. This is because we are required to compute 
\[\mathtt{Mult}\ket{r^2}\ket{m^2}\ket{0} = \ket{r^2}\ket{m^2}\ket{r^2m^2},\]
and the $\ket{m^2}$ register is particularly large due to the fact that many qubits are used to perform the calculation of the inverse root to high precision. Optimizing this multiplication routine therefore becomes a key interest. The cost of so called \textit{schoolbook} multiplication that is often cited is $2n_an_b$ Toffoli gates for the operation described above. This comes from the general procedure of performing multiplications as a series of controlled additions, and a controlled addition requires an extra Toffoli per qubit control qubit, hence the factor of 2. Optimized asymptotic schemes based on Karatsuba multiplication also exist \cite{parent2017improved, gidney2019asymptotically}, however, these only become advantageous for very large bitstrings. Litinski recently improved the scaling of quantum schoolbook multiplication by performing a series of add-subtract operations with corrections to remove the controlled adders, eliminating the factor of 2 \cite{litinski2024quantum}. Here, we slightly improve the constants stemming from said corrections by performing an initial copy operation. This is detailed in the Lemma \ref{lem:fast_mult}.

\begin{lemma}[Faster Integer Multiplication] \label{lem:fast_mult}
    Let $a, \:b \in \mathbb{N}$ encoded in $n_a$ and $n_b$ qubits respectfully, then assuming $n_a<n_b$, it is possible to perform out-multiplication and compute the product $ab$ in an $(n_a+n_b)$ qubit register using 
    \begin{equation}
        n_an_b + 2n_b + n_a +3
    \end{equation}
    Toffoli gates with $n_a+n_b+1$ temporary ancilla qubits. If $a$ and $b$ are both $n$-bit integers, then we obtain a Toffoli count of 
    \begin{equation}
        n^2 + 3n + 3,
    \end{equation}
    which is an improvement over Ref. \cite{litinski2024quantum} by an additive factor of $n$.
\end{lemma}
\begin{proof}
    To complete this proof, we leverage the results of Ref. \cite{litinski2024quantum}. Specifically, we use the main innovation of the \textit{controlled add-subtract} arithmetic unit (Fig.~1g) for improved schoolbook multiplication (Fig.~2b). The key is that one can save by one of the addition circuits, since we are applying \textit{add-subtract} to an all zero register. However, this will affect the corrections that need be applied later. We proceed with schoolbook multiplication by applying add-subtract of $b$ controlled off the bits of $a$. Controlled on $a_0$, we begin by copying $b$ into the output register. However, to account for required corrections later, we copy $b$ into a register that is left to be arbitrary using CNOT gates. Letting $c \in \mathbb{Z}$ and $c \geq 0$, we currently have 
    \begin{equation}
        a_0b 2^c,
    \end{equation}
    where if $c=0$ we copy it into the first $n_b$ bits of the output register. Now, we apply the controlled add-subtract art of Litinski controlled on all but the least significant bit of $a$ to obtain:
    \begin{align}
        &\sum_{k=1}^{n_a-1} (1-a_k)2^{n_a+k} + (2a_k -1)2^kb + a_0b2^c \\
        & 2^{n_a}(2^{n_a}-2-(a-a_0)) + 2b(a-a_0) - b(2^{n_a}-n_a) + a_0b2^c \\
        & 2ab + 2^{2n_a} + 2b - 2^{n_a}(2+a-a_0+b) - 2a_0b + 2b + a_0b2^c, \: \text{choose $c=1$}\\
        & 2ab + (2^{2n_a} + 2b) - 2^{n_a}(a+1-a_0) - 2^{n_a}(a+1). \\
    \end{align}
    By choosing $c=1$ we have copied the $b$ with CNOTs into the output register such that $b_0$ is in the $2^1$ position. Now we have a very similar equation to that obtained by Litinski, and need to apply the like corrections to obtain $ab$. Adding $2^{n_a}(a+1-a_0)$ can be done by similarly applying an adder with carry-in set to 1, and applying a CNOT on the carry-in controlled on $a_0$. This controls whether we add $2^{n_a}a$ or $2^{n_a}(a+1)$ as in Ref. \cite{litinski2024quantum}. This addition is done on the $n_b-n_a+1$ high bits of the output, costing $n_a+1$ Toffolis. Next we add $2^{n_a}(b+1)$ using an adder (again with the carry-in set to 1) on the $n_b+1$ high bits of the output, costing $n_b+1$ Toffoli. Next, subtract $(2^{2n_a} + 2b)$. First, using the fact that $2^{2n_a}$ has Hamming weight 1, we subtract only on the $n_a + n_b + 1 -2n_a$ high bits of the string, using $n_b-n_a+1$ Toffoli. Then we subtract $2b$ on all but the bottom bit using $n_a+n_b$ Toffoli. The remaining $2ab$ can then be divided by 2 using bit shifts with no non-Clifford cost, yielding the final cost result. The ancilla come from the largest necessary subtraction/addition, which is done in the final correction, plus an overflow qubit.
\end{proof}

Despite this minor speedup, the $n^2$ term cannot be improved asymptotically for schoolbook multiplication, and it seems unlikely that a prefactor of 1 can be improved upon (given that the Hamming weights of the strings is unknown), simply due to the number of logical manipulations that need be carried out. Therefore, further improvements would likely only further improve upon the linear constant. 

One minor limitation of the result is that the output register needs to be initialized in the all zero state, where as the construction by Litinski in principle allows one to add the product $ab$ directly into a $2n$ bit-string $\textup{mod} \: 2^{n+1}$. For our purposes, the following operation is subtraction by a classical constant with Hamming weight one, so this in place addition of the product does not lead to space-time savings. 

\subsection{Equality check}
The equality check circuit [\cref{fig:is_eq}] uses one qubit to flag whether the values in two registers of the same size are the same, namely
\begin{equation}
    \mathtt{IsEq}\ket{a}\ket{b}\ket{0} = \ket{a}\ket{b}\ket{a=b},
\end{equation}
where we have defined
\begin{equation}
    \ket{a=b} = \begin{cases}
        \ket{0},&\textrm{if }a\neq b \\
        \ket{1},&\textrm{if }a=b.
    \end{cases}
\end{equation}
We detail the cost of this operation in \cref{lemma:is_eq}. While this is not necessarily a novel contribution, we provide detailed costs of this operation for readability and completeness.  

\begin{proposition}[Equality Check] \label{lemma:is_eq}
    Given two $n$-bit strings stored in the registers $\ket{a}$ and $\ket{b}$, it is possible to check whether $a=b$ using $n-1$ Toffoli gates and temporary carry qubits, and a single ancilla bit to store $\ket{a=b}$. 
\end{proposition}
\begin{proof}
    First check if $a=b$ by applying a CNOT controlled on each bit of $\ket{a}$ with target on each associated bit of $\ket{b}$. Then, via temporary AND, check if $b$ is the all zero state with $n-1$ Toffoli and store the result in an ancilla qubit, un-computing temporary ancilla with measurement and Clifford operations as done with the Gidney adder \cite{gidney_halving_2018}.
\end{proof}

\begin{figure}
    \centering
    \begin{quantikz}
		\lstick{$\ket{a}$} & \qwbundle{n} & \gate{a} \vqw{2} & \\
		\lstick{$\ket{b}$} & \qwbundle{n} & \gate{b} & \\
		\lstick{$\ket{0}$} & & \gate{\mathtt{IsEq}} & \\
	\end{quantikz}
	\hspace{0.5cm} $=$ \hspace{0.5cm} 
	\begin{quantikz}
		\lstick{$\ket{\bar{a}_{0}}$} & \ctrl{4}  &  & \ \ldots\ &&&& \ \ldots\ && \ctrl{4} & \rstick{$\ket{\bar{a}_{0}}$} \\
		\lstick{$\ket{\bar{a}_{1}}$} &  & \ctrl{4}  & \ \ldots\ &&&&  \ \ldots\ & \ctrl{4} && \rstick{$\ket{\bar{a}_{1}}$} \\
		\setwiretype{n} \vdots &  && \ddots &&  && \iddots &&& \vdots \\
		\lstick{$\ket{\bar{a}_{n-1}}$} &  &  & \ \ldots\ & \ctrl{4} && \ctrl{4}  & \ \ldots\ &&& \rstick{$\ket{\bar{a}_{n-1}}$}\\
		\lstick{$\ket{\bar{b}_{0}}$} & \targ{}  &  & \ \ldots\ && \octrl{4} &&  \ \ldots\ && \targ{} & \rstick{$\ket{\bar{b}_{0}}$} \\
		\lstick{$\ket{\bar{b}_{1}}$} &  & \targ{} & \ \ldots\ && \octrl{0} && \ \ldots\ &\targ{} & & \rstick{$\ket{\bar{b}_{1}}$} \\
		\setwiretype{n} & \vdots &&&  &&&& \vdots \\
		\lstick{$\ket{\bar{b}_{n-1}}$} &  &  & \ \ldots\ & \targ{} & \octrl{0}& \targ{} & \ \ldots\ &  && \rstick{$\ket{\bar{b}_{n-1}}$} \\
		\lstick{$\ket{0}$} & &  & \ \ldots\ && \targ{} && \ \ldots \ &&& \rstick{$\ket{a=b}$} \\
	\end{quantikz}
    \caption{Circuit implementation of equality check routine discussed in \cref{lemma:is_eq}.}
    \label{fig:is_eq}
\end{figure}
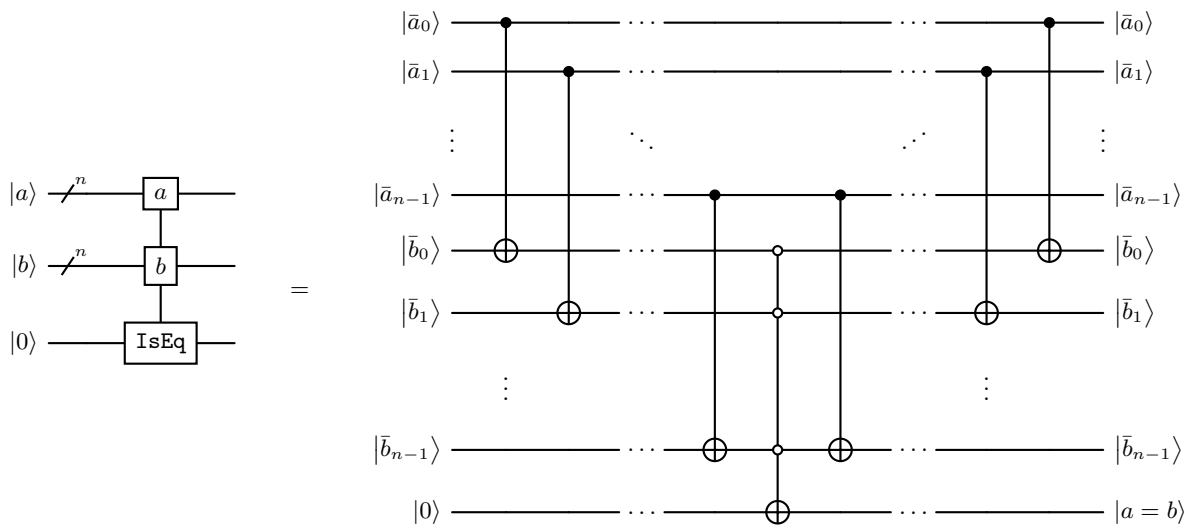

\subsection{Controlled classical subtraction circuit}
In this section we show how to add a number that stores both classical and quantum information to a quantum register. The motivation for this comes from saturating the potential in Section \ref{sec:algorithm}. This involves modifying the inequality test such that we are required to subtract by a different classical number depending on a quantum control. This can trivially be done by using a QROM-like strategy and loading the classical number in a quantum register, but for large numbers, this can be qubit intensive. Here we show how this can be achieved (under certain conditions) without ever loading the classical number at all. We first prove our construction for a more general case than necessary in our algorithm in \cref{lemma:hybrid_adder}, provide an example circuit construction in \cref{fig:hybrid_adder}, and then explain how it applies to our inequality test. 

\begin{figure}[htbp!]
    \centering
    \includegraphics[width = \textwidth]{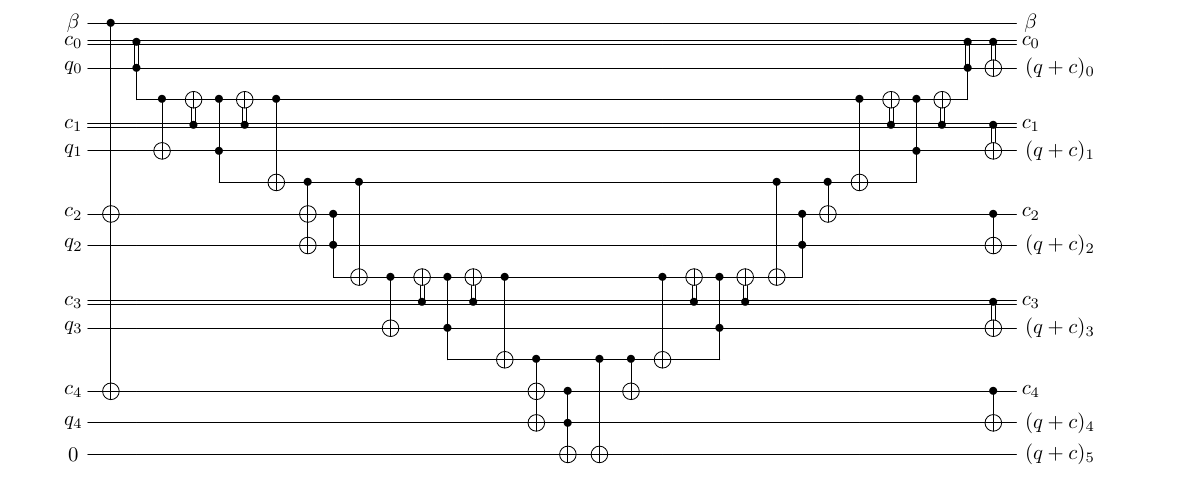}
    \caption{Circuit implementation of the adder leveraged in the proof of \cref{lemma:hybrid_adder}.}
    \label{fig:hybrid_adder}
\end{figure}

\begin{lemma}[In-place addition of controlled power of 2] \label{lemma:hybrid_adder} Given $Q\in \mathbb{N}$ encoded in $q$ quantum bits, a classical number $C \in \mathbb{N}$ encoded in $c \leq q$ bits, and a quantum control bit $\beta \in \{0,1\}$, then it is possible to compute the addition $Q + C + \beta[(-1)^{c_j} 2^{c_j} + (-1)^{c_k} 2^{c_k}]$ using $q-1$ Toffoli gates, $q-2$ temporary carry qubits, and 3 ancilla qubits; $1$ to store the overflow and $2$ needed to encode the classical constant $C$, specifically the bits $c_j$ and $c_k$ whose values are dictated by the single qubit quantum control register $\ket \beta$. 
\end{lemma}
\begin{proof}
    Considering first the qubit counts, observing the example circuit in \cref{fig:hybrid_adder} it is clear that only the bits of $C$ controlled by $\beta$ need a quantum encoding, specifically those of $c_j$ and $c_k$, and one quantum bit is needed to store overflow. Then the rest of $C$ is encoded via classical controls and $q$ additional bits are required to store $Q$. The circuit works by replacing quantum adder units with classical adder subunits introduced in Fig.~17 of Ref. \cite{sanders2020compilation}. It is clear from the example that under the lemma conditions this can be done for a bit-string of arbitrary length. In terms of the gate cost, 1 Toffoli comes from each adder unit with the exception of the first one (whereby we have effectively assumed $j\neq k \neq 0$). 
\end{proof}
Note that the control register $\ket \beta$ changes the Hamming weight of $C$ by at most the number of registers it acts on, which in the case above is 2. This circuit can be further generalized to cases where the quantum control performs more intricate manipulations to the classical constant at the price of encoding more bits in quantum registers. \\

We now explain how this can be used in our algorithm. The circuit in \cref{fig:hybrid_adder} can clearly be used as a subtractor, similarly as shown in the circuit of Figs. 13 and 15 of Ref. \cite{berry2019qubitization}, all components  of the adder can be reversed to be used as a subtractor circuit (with the exception of the overflow carry-out logic that remains unchanged). In the inequality test of \cref{fig:U_V} where we have added a control for modifying the saturation of nuclear-nuclear interactions (see \cref{app:Gamma}), we are required to check $r^2m^2 - 2^{2n_M+n_\Gamma}<0$. The key here is to notice that $2^{2n_M}$ is a Hamming weight 1 bitstring, whereby the location of the bit can be swapped/shifted by $n_\Gamma$ bits with a single Toffoli gate. These properties allow $2^{2n_M+n_\Gamma}$ to be written in the form $C + n_\Gamma[(-1)^{c_j} 2^{c_j} + (-1)^{c_k} 2^{c_k}]$ from \cref{lemma:hybrid_adder}, where in our case, all but one of $c_j = 2n_M$ or $c_k = 2n_M + n_\Gamma$ are zero, making much of the classical control logic in \cref{fig:hybrid_adder} unnecessary.

\subsection{State preparation circuits} \label{subapp:prep}
In this section we discuss how to implement all \prep circuits appearing in the main text. Before showing these circuits, it is useful to present a lemma summarizing the cost of implementing rotations with some target accuracy using a phase kickback register.
\begin{lemma}[Single qubit rotation using phase kickback register] \label{lemma:toffoli_rot}
Any single qubit rotation operator, namely $R_x$, $R_y$, or $R_z$, can be implemented with $\epsilon_{rot}$ accuracy by using a phase gradient register consisting of at least
\begin{equation}
    n_{R}(\epsilon_R) = \ceil{\log_2 \frac{\pi}{\epsilon_{R}}}
\end{equation}
qubits, with an associated Toffoli cost of
\begin{equation}
    \mathcal{T}_R(\epsilon_R) = n_{R}(\epsilon_R).
\end{equation}
\end{lemma}
\begin{proof}
    We start by noting that all rotation operators are equivalent to e.g. $R_z$ under a Clifford transformation, having that all will then have the same Toffoli cost. Considering that this rotation is implemented with Toffoli gates by using a controlled classical adder on an $n_{R}$-qubit phase kickback register, the associated angle is then implemented with an accuracy of $\delta_{R}=2\pi/2^{n_{R}}$, from which we can bound the accuracy for the $R_z$ rotation as $\epsilon_{R}\leq \delta_{R}/2$ using the fact that $\delta_{R}\ll 1$ \cite{nielsen2000quantum,sanders2015bounding}. It follows immediately that by choosing the $n_{R}(\epsilon_R)$ considered in this lemma, the overall accuracy for the rotation will be upper bounded by $\epsilon_R$. The Toffoli cost of this operation then corresponds to the cost of doing one controlled adder on $n_R$ qubits, with a Toffoli cost of $n_R$.
\end{proof}

\subsubsection{W state preparation}
The $\textup{PREP}_W$ circuit shown in \cref{fig:T1} corresponds to the preparation of the $\ket{W}$ state, which is defined as
\begin{equation} \label{eq:w_state}
    \ket{W}=(\ket{001}+\ket{010}+\ket{100})/\sqrt{3}.
\end{equation}
The corresponding circuit is shown in \cref{fig:w_prep}, with the cost costs being summarized in the following lemma.
\begin{lemma}[W state preparation]\label{lemma:prep_w}
    The W state in Eq.~\eqref{eq:w_state} can be prepared to $\epsilon_W$ accuracy by using
    \begin{equation}
        \mathcal{T}_W(\epsilon_W) = \ceil{\log_2\frac{\pi}{\epsilon_W}}+1
    \end{equation}
    Toffoli gates, assuming we have access to an appropriate phase gradient state.
\end{lemma}
\begin{proof}
    We start by considering how the the only error in this operation comes from a single gate, namely $R_y$. Using the triangle inequality, it follows immediately that the error in this operation ($\epsilon_W$) can be upper bounded by the error in implementing $R_y$. Using \cref{lemma:toffoli_rot} for implementing this $R_y$ using a phase gradient register with $\ceil{\log_2 \pi/\epsilon_W}$ qubits, it follows immediately that an accuracy of $\epsilon_W$ is achieved. The only other non-Clifford gate in this operation corresponds to a controlled Hadamard, which can be implemented for one Toffoli gate. Combining this cost with that of the rotation, we obtain the cost outlined in this lemma. The right action of this circuit follows immediately after the action of $R_y(2\arccos 1/\sqrt3)\ket{0} = (\ket{0}+\sqrt{2}\ket{1})/\sqrt{3}$ after propagating through the rest of the circuit.
\end{proof}

\begin{figure}
    \centering
    \begin{quantikz}
		\lstick{$\ket{0}$} & \gate[3]{\textup{PREP}_W} & \\
		\lstick{$\ket{0}$} & & \\
		\lstick{$\ket{0}$} & &
	\end{quantikz}
$=$
	\begin{quantikz}
		\lstick{$\ket{0}$} & \gate{R_y\left(2\arccos\frac{1}{\sqrt{3}}\right)} & \ctrl{1} && \ctrl{1} & \gate{X} & \rstick[3]{$\ket{W}$} \\
		\lstick{$\ket{0}$} & & \gate{H} & \ctrl{1} & \targ{} && \\
		\lstick{$\ket{0}$} & && \targ{} &&& 
	\end{quantikz}
    \caption{Circuit preparing the $\ket{W}=(\ket{001}+\ket{010}+\ket{100})/\sqrt{3}$ state.}
    \label{fig:w_prep}
\end{figure}

\subsubsection{Controlled coherent alias sampling with flagging} \label{subsubapp:gen_prep}
We show how to implement the controlled version of the coherent alias sampling routine from Ref.~\cite{babbush2018encoding}, which can be used for preparing a state for arbitrary coefficients. These controlled versions of the coherent alias sampling circuits have been presented in Ref.~\cite{loaiza2025majorana}, however they are missing a minor correction for ensuring that the ``junk'' registers are also returned on an all-zeros state for the case where the control qubit is not activated. We also consider a general case where a flag is added to the returned states, namely for a target coefficient vector $(c_1,\cdots,c_K)$ with associated 1-norm $\lambda_{\vec c}=\sum_k |c_k|$, we return the state
\begin{equation} \label{eq:flag_prep}
    \textup{PREP}^{(n_F)}(\vec c) = \frac{1}{\sqrt{\lambda_{\vec c}}} \sum_k \sqrt{|c_k|} \ket{k}\ket{\vec b_k} \ket{\textrm{temp}(k)},
\end{equation}
for $\vec b_k$ a binary vector of length $n_F$. The controlled implementation of this operation is presented in \cref{fig:control_prep}, with the following lemma summarizing its costs.

\begin{lemma}[Generic state preparation with flagging]\label{lemma:gen_prep}
    Consider a coefficient vector $\vec c=(c_1,\cdots,c_K)$ with associated 1-norm $\lambda_{\vec c}:= \sum_k |c_k|$, and some number $n_F\geq 0$ corresponding to the length of binary flags to include (with $n_F=0$ corresponding to adding no flags). Implementing the operation outlined in Eq.~\eqref{eq:flag_prep} to accuracy $\epsilon$ can be implemented with a Toffoli cost of
    \begin{equation}
        \mathcal{T}_{\vec c}(\epsilon) = b_K+n_F+2l_K+\mathcal{Q}(K,b)+2\mathcal{T}_R(\epsilon/4)+\aleph +[l_K+k_K+\aleph+1]
    \end{equation}
    using $n_c = b_K + n_F$ qubits, $\tilde n_c=b_K+2\aleph+n_F+1$ ``junk' '  ancillas that need to be passed for uncomputation, and 
    \begin{equation}
        \tilde n_c = \max\{l_K-1+[1],n(K,b),\aleph-1\}
    \end{equation}
   temporary carry ancillas. Here we have defined $\mathcal{Q}(K,b)$ and $n(K,b)$ as the Toffoli and temporary carry qubits cost of loading $K$ bitstrings of length $b=2n_F+\aleph+\ceil{\log_2 K}$ (which could be done with a QROM using $\mathcal{Q}(K,b)=K-1$ and $n(K,b)=\ceil{\log_2(K)}-1$), noting that there are approaches for reducing the Toffoli count in exchange for using more qubits. We have also defined $b_K:=\ceil{\log_2 K}$,  $\displaystyle k_K=\max_{k\in\mathbb N}\{k \ :\ K/2^k\in\mathbb N\}$ the largest power of $2$ dividing $K$, $l_K=\ceil{\log_2(K/2^{k_K})}$, and $\displaystyle \aleph=\ceil{\log_2\frac{2}{K\epsilon}}$. Numbers in brackets correspond to modifications when this operation is controlled.
    
\end{lemma}
\begin{proof}
    The action of this circuit is completely analogous to that from the original coherent alias sampling implementation of Ref.~\cite{babbush2018encoding}, while having that the flagging procedure is a straightforward generalization from the one-qubit flags introduced in Refs.~\cite{berry2019qubitization,loaiza2025majorana}. We start by noting that the error from this operation will come from two components: first, the approximation of the angle for the $R_z$ rotation inside of the UNIFORM routine, and second the discretization error from using $\aleph$ qubits for the $keep$ register. Using the triangle inequality, it follows immediately that by allocating an error of $\epsilon/2$ to each of these will yield an overall accuracy for our procedure which is upper bounded by $\epsilon$, having that each $R_z$ rotation is then implemented to $\epsilon/4$ accuracy. We note that tighter error estimates could yield better bounds. However, the logarithmic dependence would make the associated resources extremely similar to those presented here. Using Eq.(35) from Ref.~\cite{babbush2018encoding}, we can make the error on each of the loaded coefficients $\rho_k:=\sqrt{|c_k|/\lambda_{\vec c}}$ smaller than $\epsilon/2$ by choosing $\displaystyle \aleph=\ceil{\log_2\frac{2}{K\epsilon}}$. We now start by deducing the costs for the uncontrolled version of this routine. We now show the costs for each part of this routine:
    \begin{enumerate}
        \item One [controlled] call to UNIFORM$(K)$, requiring $2l_K+2\mathcal{T}_R(\epsilon/2)+[l_K+k_K+1]$ Toffolis and $l_K-1+[1]$ temporary carry ancillas, where $\mathcal{T}_R(\epsilon/2)$ is the associated cost of performing an $R_z$ rotation (see \cref{lemma:toffoli_rot}).
        \item One [controlled] call to $\aleph$ Hadamard gates, requiring $[\aleph]$ Toffoli gates.
        \item One data loading routine over $K$ coefficients with data length $b_K+\aleph+2n_F$, with $\mathcal{Q}(K,n_b)$ Toffoli cost and requiring $n(K,b_K+\aleph+2n_F)$ temporary carry ancillas.
        \item One call to a comparison circuit between two $\aleph$ bit registers, which can be done as a subtraction \cite{berry2019qubitization} requiring $\aleph$ Toffolis and $\aleph-1$ temporary carry ancillas and one ancilla for storing the result.
        \item One [open controlled] series of NOT gates for no Toffoli cost.
        \item  A series of $b_K+n_F$ controlled SWAPs, requiring $b_K+n_F$ Toffolis.
    \end{enumerate}
Combination of these steps returns the costs outlined in this lemma.
\end{proof}

\begin{figure}
    \centering
    \begin{quantikz}[row sep=0.4cm]
	\lstick{$\ket{ctl}$} &\phantomgate{12345} & \ctrl{3} &  && \octrl{2} &&& \\
	\lstick{$\ket{\vec 0}$} & \qwbundle{\ceil{\log_2 K}} & \gate{\textrm{UNIFORM}(K)} & \gate{In(k)} \wire[d][6]{q} &&& \targX{}\vqw{1} && \\
	\lstick{$\ket{\vec 0}$} & \qwbundle{\ceil{\log_2 K}} && \gate{\textrm{data : alt}_k} & & \gate[6]{\vec X(s=0)} &  \swap{3} &&   \rstick{$\ket{k}$} \\
	\lstick{$\ket{\vec 0}$} & \qwbundle{\aleph} & \gate{H^{\otimes{\aleph}}} && \gate{a \wire[d][2]{q}} &&&& \\
	\lstick{$\ket{\vec 0}$} & \qwbundle{\aleph} && \gate{\textrm{data : keep}_k} & \gate{b} &&&& \\
	\lstick{$\ket{0}$} &&& & \gate{\mathtt{Comp}(a\leq b)} && \ctrl{0} & \ctrl{2} &   \\
	\lstick{$\ket{0}$} &\qwbundle{n_F} && \gate{\textrm{data : }\vec b_k} &&&& \targX{} & \rstick{$\ket{\vec b_{k}}$} \\
	\lstick{$\ket{0}$} &\qwbundle{n_F}&& \gate{\textrm{data : alt}_{\vec b_k}} &&&& \targX{} & \\
\end{quantikz}
    \caption{Controlled coherent alias sampling circuit with $n_F$ flagging qubits.}
    \label{fig:control_prep}
\end{figure}

\subsubsection{Alias sampling for symmetric matrices} \label{subsubapp:sym_prep}
We now show an extension of the coherent alias sampling PREPARE circuit \cite{babbush2018encoding,loaiza2025majorana} that loads data corresponding to a two-dimensional symmetric matrix. This circuit exploits the fact that for a $\bm C$ matrix with dimensions $K\times K$, not all $K^2$ coefficients need to be loaded, saving some costs on the QROM of the coherent alias sampling routine, following the ideas from the sparse data loading routine in Ref.~\cite{berry2019qubitization}. The associated action of this circuit with $F$ flags then corresponds to

\begin{equation}
    \textup{PREP}^{(F)}(\bm{C}) \ket{\vec 0} = \sum_{j,k=1}^K \sqrt{\frac{|C_{jk}|}{\lambda_c}} \ket{j}\ket{k}  \ket{\textrm{junk}(j,k)} \bigotimes_{f=1}^F\ket{b^{(f)}_{jk}},
\end{equation}
where $b_{jk}^{(f)}$ is a binary variable flagging some property of interest associated with coefficients $j,k$ and we have defined the 1-norm
\begin{equation}
    \lambda_C = \sum_{j,k=1}^K |C_{jk}|.
\end{equation}
The associated circuit is presented in \cref{fig:gen_prep}, which closely follows the implementation of the regular coherent alias sampling routine from \cref{lemma:gen_prep} with two key additions: utilization of an additional register for only loading unique non-zero coefficients, as done for sparse data loading \cite{berry2019qubitization}, and swapping of registers encoding $j,k$ coefficients in the end for recovering a symmetric matrix. Note that as mentioned in Refs.~\cite{berry2019qubitization,loaiza2025majorana}, the data to be loaded by the register encoding $S$ non-zero coefficients will run through at most $K(K+1)/2$ coefficients as only the upper-diagonal is required to fully specify a symmetric matrix. Defining a bijection $s\leftrightarrow(j_s,k_s)$, the coefficients to be loaded then become the non-zero entries of
\begin{equation}
    C_{s} = \begin{cases}
        \sqrt{2}C_{j_sk_s} &\text{if} \:j_s < k_s \\
        C_{j_s j_s} & \text{if} \:j_s = k_s\\
        0 & \text{otherwise}.
    \end{cases}
\end{equation}

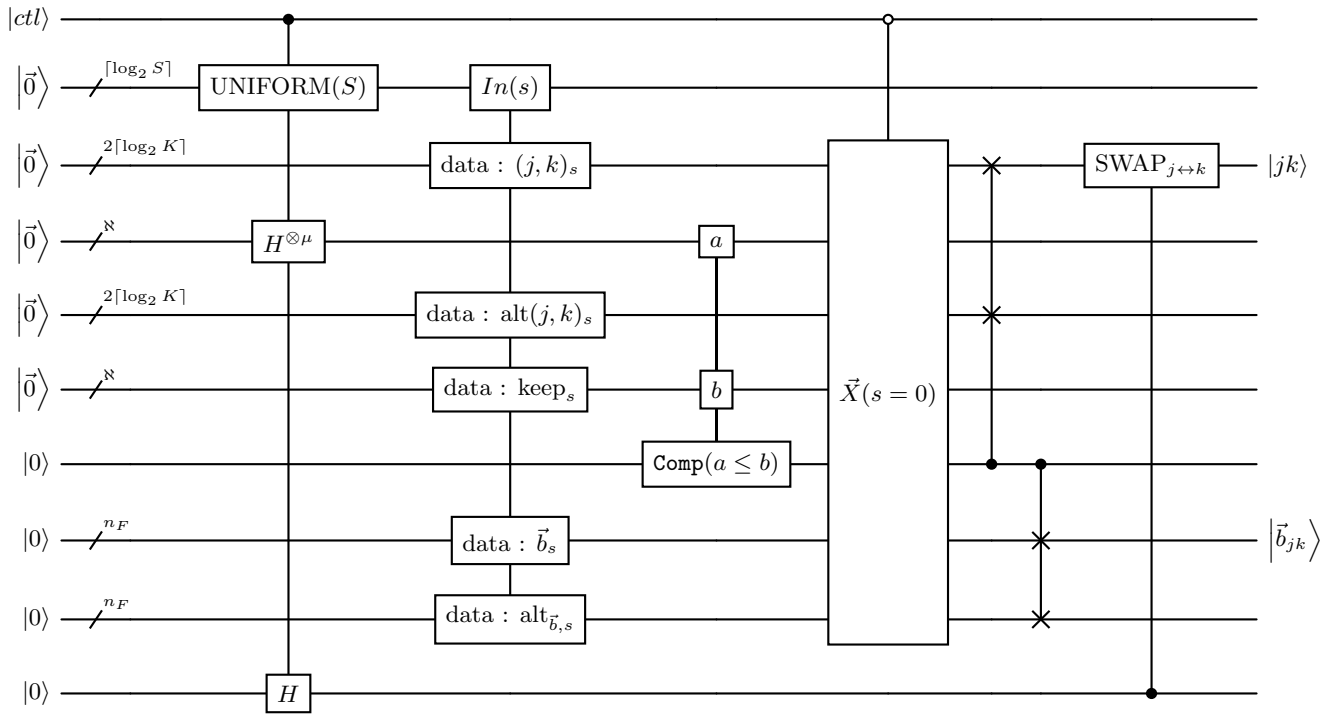
\begin{figure}
    \centering
    \begin{quantikz}[row sep=0.4cm]
		\lstick{$\ket{ctl}$} &\phantomgate{12345} & \ctrl{9} &  && \octrl{8} &&&& \\
		\lstick{$\ket{\vec 0}$} & \qwbundle{\ceil{\log_2 S}} & \gate{\textrm{UNIFORM}(S)} & \gate{In(s)} \wire[d][7]{q} &&&&&& \\
		\lstick{$\ket{\vec 0}$} & \qwbundle{2\ceil{\log_2 K}} && \gate{\textrm{data : }(j,k)_s} & & \gate[7]{\vec X(s=0)} &  \swap{4} && \gate{\textrm{SWAP}_{j\leftrightarrow k}} & \rstick{$\ket{jk}$} \\
		\lstick{$\ket{\vec 0}$} & \qwbundle{\aleph} & \gate{H^{\otimes{\mu}}} && \gate{a \wire[d][3]{q}} &&&&& \\
		\lstick{$\ket{\vec 0}$} & \qwbundle{2\ceil{\log_2 K}} && \gate{\textrm{data : alt}(j,k)_s} &&& \targX{} &&& \\
		\lstick{$\ket{\vec 0}$} & \qwbundle{\aleph} && \gate{\textrm{data : keep}_s} & \gate{b} &&&&& \\
		\lstick{$\ket{0}$} &&& & \gate{\mathtt{Comp}(a\leq b)} && \ctrl{0} & \ctrl{2} &  & \\
		\lstick{$\ket{0}$} &\qwbundle{n_F} && \gate{\textrm{data : }\vec b_s} &&&& \targX{} && \rstick{$\ket{\vec b_{jk}}$} \\
		\lstick{$\ket{0}$} &\qwbundle{n_F}&& \gate{\textrm{data : alt}_{\vec b,s}} &&&& \targX{} && \\
		\lstick{$\ket{0}$} &&\gate{H}&&&&  && \ctrl{-7} &
	\end{quantikz}
    \caption{Circuit for loading positive symmetric matrix $\bm{C}\in \mathbb{R}^{K\times K}$ with $S$ unique non-zero coefficients and $n_F$ flagging registers. The $\vec X(s=0)$ routine here corresponds to a series of X gates which reset the registers to the all-zeros state for the case where the control qubit is not activated. $\mathtt{Comp}$ routine compares two qubit registers of the same size and flags the result in an additional qubit \cite{berry2018improved}, while the SWAP over $j\leftrightarrow k$ swaps qubits from register encoding index $j$ with those of index $k$, consisting $\ceil{\log_2 K}$ controlled SWAPs.}
    \label{fig:gen_prep}
\end{figure}

\subsubsection{Amplitude-amplified alias sampling} \label{subsubapp:aa_prep}
We now show how to implement the state preparation routine which acts as
\begin{align} \label{eq:prep_targ}
    \textup{PREP}_{V,\Gamma}\ket{0}\ket{0} &= \frac{1}{\sqrt{\lambda_V(\Gamma)}}\sum_{i\neq j}^\eta \sqrt{\gamma_\Gamma(i,j)} \ket{i}\ket{j} \\
    &:= \ket{\zeta_{target}(\Gamma)}
\end{align}
where we consider
\begin{align}
    \gamma_\Gamma(i,j) = \begin{cases}
        \frac{\zeta_i\zeta_j}{\Gamma},& \textrm{ if particle } i \textrm{ and } j \textrm{ are both nuclei }, \\[8pt]
        \zeta_i\zeta_j& \textrm{ otherwise },
    \end{cases}
\end{align}
and have defined the 1-norm
\begin{equation}
    \lambda_V(\Gamma) = \sum_{i\neq j} \gamma_\Gamma(i,j).
\end{equation}
Note that we here have made the dependence with respect to the nuclear-nuclear factor $\Gamma\geq 1$ explicit as to make this discussion general for both oracles $\textup{PREP}_V$ and $\textup{PREP}_{V,\Gamma}$ presented in the main text, having that the former will correspond to choosing $\Gamma=1$. We start by considering a product state consisting of two separate alias sampling routines, each one preparing a coefficient vector $(\zeta_1,\zeta_2,\cdots,\zeta_\eta)$, associated with the state
\begin{align}
    \ket{\zeta_{prod}} &:= \left(\textup{PREP}(\vec\zeta)\ket{0}\right)\otimes \left(\textup{PREP}(\vec\zeta)\ket{0}\right) \\
    &= \frac{1}{\lambda_\zeta} \left(\sum_{i} \sqrt{\zeta_i} \ket{i}\right) \otimes \left(\sum_{j} \sqrt{\zeta_j} \ket{j}\right), \label{eq:prod_prep}
\end{align}
where we have defined $\lambda_\zeta = \sum_{i} \zeta_i$. The idea of this approach is to use amplitude amplification \cite{brassard2000quantum} as to recover the target state in \cref{eq:prep_targ}. Each of the coherent alias sampling routines $\textup{PREP}(\vec\zeta)$ has a Toffoli cost $\mathcal{T}_{\vec\zeta}\in \mathcal{O}(\eta)$ (see \cref{lemma:gen_prep}). Below we show that for most chemical reactions a single call to amplitude amplification is sufficient for obtaining the target state with a success probability of $1$, making the overall cost of our state preparation routine $6\mathcal{T}_{\vec\zeta} + \mathcal{O}(1) \in \mathcal{O}(\eta)$:  one call to amplitude amplification requires $3$ calls to the associated oracle, with each oracle here requiring $2$ calls to $\textup{PREP}(\vec\zeta)$. \\

Deducing the number of calls to amplitude amplification can be done by calculating the associated success probability
\begin{align}
    P(\Gamma,\vec\zeta) &:=|\langle\zeta_{target}(\Gamma)|\zeta_{prod}\rangle|^2  \label{eq:p_succ} \\
    &= \frac{1}{\lambda_\zeta^2 \lambda_V(\Gamma)} \left(\sum_{i\neq j} \sqrt{\zeta_i\zeta_j\gamma_\Gamma(i,j)}\right)^2 \\
    &\geq \frac{\lambda_V(\Gamma)}{\lambda_\zeta^2}.
\end{align}
Defining $\mathcal{M}$ as the set of electron-nuclear and electron-electron interactions and noting that this overlap will be minimized in the $\Gamma\rightarrow\infty$ limit, we get to
\begin{align}
    P(\Gamma,\vec\zeta)  &\geq  \frac{\sum_{i,j\in\mathcal{M}}\zeta_i\zeta_j}{\left(\sum_i \zeta_i\right)^2 } \\
    &= \frac{2\eta_{e} Z_{tot} + \eta_{e}(\eta_{e}-1)}{Z_{tot}^2+2\eta_{e}Z_{tot}+\eta_{e}^2},
\end{align}
where we have defined the total nuclear charge $Z_{tot}=\sum_{i\in nuc} \zeta_i$. Since for chemical reactions we can generally consider for the net charge to be close to zero, we can replace $Z_{tot}\approx \eta_{e}$, arriving to
\begin{align}
    P(\Gamma,\vec\zeta)  &\gtrsim \frac{3}{4} - \frac{1}{4\eta_{e}} \geq \frac{1}{2},
\end{align}
where we have considered to have at least one electron in the simulation for the last inequality. Overall, this shows that the success probability is greater than $1/4$, from which a single round of amplitude amplification will suffice to obtain the target state deterministically (Theorem 4 from Ref.~\cite{brassard2000quantum}). The associated circuit is presented in \cref{fig:aa_prep}, which acts as follows:
\begin{enumerate}
    \item Apply the coherent alias sampling technique associated with the coefficients $(\zeta_1,\cdots,\zeta_\eta)$ in two separate registers. One flagging qubit is used in these routines to mark whether the associated index corresponds to a nuclei or an electron. This can be done for almost  the same cost $\mathcal{T}_{\vec\zeta}$ of implementing the controlled generic alias sampling state preparation, requiring one additional Toffoli gate and two extra qubits \cite{loaiza2025majorana}.
    \item Apply an $R_y$ gate with a rotation angle encoding the success probability $P(\Gamma,\zeta)$ controlled on both coefficient registers encoding a nuclear index. This can be done with a cost of $\mathcal{T}_R+1$ Toffolis using one additional qubit to encode whether both indices are nuclear $\ket{i,j\in\mathcal{N}}$, where $\mathcal{T}_R$ is the Toffoli cost of implementing a controlled rotation.
    \item Perform an equality test (\cref{lemma:is_eq}) over the registers encoding the particle indices, flagging the result in one additional qubit $\ket{i=j}$. This routine requires $n_\eta-1$ Toffoli gates and temporary carry qubits, where we have defined $n_\eta:=\ceil{\log_2\eta}$.
    \item Controlled on the equality test qubit $\ket{i=j}$ from step 3 being false and the rotation angle qubit from step 2, flag an additional qubit for the overall success of the state preparation routine, using an additional qubit and one Toffoli gate.
    \item Steps 1-4 constitute the oracle $\mathcal{A}$ for a Grover iteration. This oracle will be called a total of three times for doing one round of amplitude amplification, which as discussed in Theorem 4 of Ref.~\cite{brassard2000quantum} will return the target state with certainty by choosing adequate rotation angles in between applications of $\mathcal{A}$ and $\mathcal{A}^\dagger$. The rotation around the successful application can be done with a single $R_z$ gate on the qubit flagged in step 4 for $\mathcal{T}_R$ cost. The rotation around the initial state requires an $R_z$ gate controlled over $2n_\eta+8$ qubits, requiring $2n_\eta+7$ temporary carry ancillas and $2n_\eta+7+\mathcal{T}_R$ Toffoli gates.
\end{enumerate}
Overall, this routine has a total Toffoli cost of 
\begin{align}
    \mathcal{T}_{{\rm PREP}_V} &= 3\cdot(2(\mathcal{T}_{\vec\zeta}+1)+(\mathcal{T}_R+1)+(n_\eta-1)+1)+2n_\eta+7+2\mathcal{T}_R \\
    &= 6\mathcal{T}_{\vec\zeta} + 2n_\eta +16+5\mathcal{T}_R.
\end{align}

\begin{figure}
    \centering
   \begin{quantikz}
		\lstick{$\ket{ctl}$} && \ctrl{3} &&&&&& \ctrl{1} & \gate{R_z\left(\varphi\right)} \vqw{8} & \ctrl{1} &&& \\
		\lstick{$\ket{0}$} & \qwbundle{} & \gate[2]{\textup{PREP}(\vec\zeta)} \gategroup[8,steps=5,style={dashed,rounded
			corners,fill=X5!40, inner
			xsep=0pt},background,label style={label
			position=below,anchor=north,yshift=-0.3cm, xshift=0cm}]{$\mathcal{A}$} && \gate{i} \vqw{4} &&&& \gate[8,style={fill=X5!40,rounded corners}]{\mathcal{A}^\dagger} & \octrl{0} & \gate[8,style={fill=X5!40,rounded corners}]{\mathcal{A}} && \rstick{$\ket{i}$} \\
		\lstick{$\ket{0}$} &&&\ctrl{4} &&  &&&& \octrl{0} &&& \rstick{$\ket{i\in nuc}$} \\
		\lstick{$\ket{0}$} & \qwbundle{} & \gate[2]{\textup{PREP}(\vec\zeta)} && \gate{j} &&&&& \octrl{0} &&& \rstick{$\ket{j}$} \\
		\lstick{$\ket{0}$} &&& \ctrl{0} &&  &&&& \octrl{0} &&& \rstick{$\ket{j\in nuc}$} \\  
		\lstick{$\ket{0}$} &&&& \gate{\mathtt{IsEq}} && \octrl{2} &&& \octrl{0} &&& \rstick{$\ket{i=j}$} \\
		\lstick{$\ket{0}$} &&& \targ{} & \ctrl{1} &  &  &&& \octrl{0} &&& \rstick{$\ket{i,j\in \mathcal{N}}$} \\
		\lstick{$\ket{0}$} &&&& \gate{R_y\left(\theta\right)} &  & \ctrl{1} &&& \octrl{0} &&&      \\
		\lstick{$\ket{0}$} &&&&&& \targ{} & \gate{R_z\left(\phi\right)} && \octrl{0} && \gate{X} & \rstick{$\ket{0}$} 
	\end{quantikz}
    \caption{Controlled $\textup{PREP}_V(\Gamma,\vec\zeta)$ circuit implementing operation in \cref{eq:prep_targ} using one round of amplitude amplification. $\textup{PREP}(\vec\zeta)$ routines correspond to controlled coherent alias sampling with one flag qubit marking whether $i$ index corresponds to a nuclear term, as shown in Fig. 7 of Ref. \cite{loaiza2025majorana}. $\mathtt{IsEq}$ routine is shown in \cref{lemma:is_eq}, flagging the ancilla qubit if indices $i$ and $j$ are equal. Rotation angles $\theta,\phi$ and $\varphi$ are determined by overlap $P(\Gamma,\vec\zeta)$ defined in \cref{eq:p_succ}. Logic inside of $\mathcal{A}$ operation in yellow rectangle following application of PREPs flags successful preparation of $\ket{\zeta_{target}(\Gamma)}$ with $\ket{1}_{succ}$ qubit, which is then used to rotate the Grover iterate for performing amplitude amplification \cite{brassard2000quantum}.}
    \label{fig:aa_prep}
\end{figure}
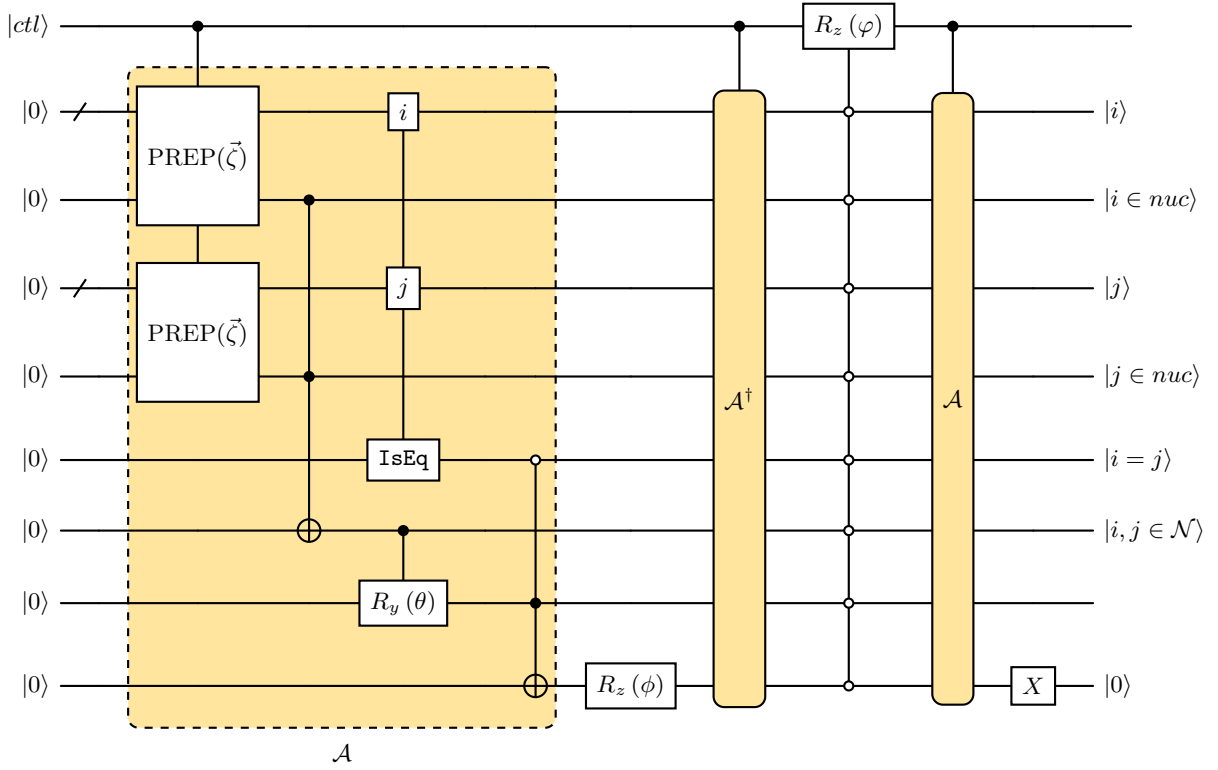

\section{Cost and error analysis} \label{app:proofs}
In this appendix we show how the action of our circuits correctly implements the associated operations, while also analyzing the total error from these operations. We start by showing the action of our circuits for block-encoding $V$ and $T$ while deducing associated costs and errors, building up to proving \cref{theo:ham} which outlines the cost of block-encoding the molecular Coulomb Hamiltonian. We note that we generally consider arithmetic routines to be implemented without any error.  Before showing the action of our operators, it is useful to establish some general notation, as well as a general theorem for linear combinations of block-encodings. \\

Let's first consider the block-encoding error of an operator $A$. We here use the notation $\tilde{A}$ to indicate an imperfect construction of the operator. For example, we say $\widetilde{O}_A$ imperfectly block-encodes the operator $A$ up to accuracy $\epsilon_A$ if 
\begin{equation}
  \|\alpha_A(\bra{0}^{m_i} \otimes \mathbb{I}) \widetilde{O}_A (\ket{0}^{m_i} \otimes \mathbb{I}) - A\| \leq \epsilon_A.
\end{equation}

Following the notation of Ref.~\cite{gilyen_quantum_2019}, we call this a $(\alpha_A, n_A, \epsilon_A)$ block-encoding of the operator $A$. In the context of \gls{lcu}, $n_A$ would be the number of qubits used to represent the PREPARE register. We also use the notation of an $(\beta,Q,\epsilon)$ state preparation pair $(P_L,P_R)$ to denote a PREPARE unitary $P_L$ and its associated uncomputation $P_R$ (noting that often $P_R=P_L^\dagger)$ that prepare a state with coefficients encoded in $Q$ qubits given by $\vec y$ with a 1-norm of $\beta=\sum_j |y_j|$ such that the sum of errors over coefficients $\sum_j |y_j-\beta \tilde y_j/\tilde\beta|\leq\epsilon$, where $\tilde y_j$ are the effective coefficients being implemented and $\tilde\beta$ their associated 1-norm. We now present a lemma summarizing how to block-encode a linear combination of block-encodings. Lemma 52 of Ref. \cite{gilyen_quantum_2019} provides the means for analyzing a linear combination of block-encoded operators, however, under the assumption that all block-encodings in the linear combination have the same parameters $(\alpha, Q, \epsilon)$. We provide a minor generalization in the following lemma and then leverage it to conclude the block-encoding parameters of our Hamiltonian from this algorithm.

\begin{lemma}[Linear combination of distinct block-encodings (generalization of (Lemma 52, Ref. \cite{gilyen_quantum_2019})] \label{lem:be_sum}
    Let $ A=\sum_{j=0}^{m-1} y_j  A_j$ be an $s$-qubit operator. Suppose $( P_L,  P_R)$ is an $(\beta, Q, \epsilon_y)$ state preparation pair for the vector $\vec{y} \in \mathbb{C}^m$. If $\forall \: j \in \{0, ..., m-1\}$ we have $ O_j$ an $(\alpha_j, Q_j, \epsilon_j)$ block-encoding of $ A_j$, and 
    \begin{equation}
        \textup{SEL}_A= \sum_{j=0}^{m-1}\ketbra{j}{j} \otimes  \mathbb{I}_{(a-a_j)}\otimes  O_j + \left(({\mathbb{I}} -  \sum_{j=0}^{m-1} \ketbra{j}{j}) \otimes {\mathbb{I}}_i \otimes {\mathbb{I}}_s\right ), \nonumber
    \end{equation}
    where $Q_J := \max_j Q_j$, and $\alpha := \sum_j \alpha_j$, then we can construct an $(\alpha \beta , Q+Q_J, \epsilon_j \max_j \alpha_j + \beta \max_j \epsilon_j)$ block-encoding of $ A$ with one query to each of $ P_L$,$  P_R$, and $\textup{SEL}_A$.
\end{lemma}
\begin{proof}
    Observe that $\tilde{O}_A:= ( P_L^\dagger \otimes {1}_i\otimes {1}_s)\textup{SEL}( P_R \otimes {1}_i\otimes {1}_s)$ is a $(\alpha \beta,Q+Q_J,\epsilon_j \max_j \alpha_j + \beta \max_j \epsilon_j)$ block-encoding of $ A$:
    \begin{align}
        & \left \|  A - \alpha \beta (\bra{0}^Q\otimes \bra{0}^{Q_J}\otimes \mathbb{I}) \tilde{O}_A (\ket{0}^Q\otimes \ket{0}^{Q_J}\otimes {\mathbb{I}})\right \| \\ 
        &= \left \| A - \sum_{j=0}^{m-1} \beta \frac{\tilde{y}_j}{\tilde{\beta}} \alpha_j \left ((\bra{0}^{Q_J} \otimes {\mathbb{I}})({\mathbb{I}}_{({Q_J}-Q_j)}\otimes \tilde{O}_j) (\ket{0}^{Q_J} \otimes {\mathbb{I}}) 
        \right)\right \| \\
        & \leq \sum_{j=0}^{m-1} \left |\beta \frac{\tilde{y}_j}{\tilde{\beta}} - y_j\right |\alpha_j + \left \| A - \sum_{j=0}^{m-1} y_j \alpha_j \left ((\bra{0}^{Q_J} \otimes {\mathbb{I}})({\mathbb{I}}_{({Q_J}-Q_j)}\otimes \tilde{O}_j) (\ket{0}^{Q_J} \otimes {\mathbb{I}}) 
        \right)\right \| \\
        & \leq \sum_{j=0}^{m-1} \left |\beta \frac{\tilde{y}_j}{\tilde{\beta}} - y_j\right |\alpha_j + \sum_{j=0}^{m-1} y_j\left \| A_j - \alpha_j \left ((\bra{0}^{Q_J} \otimes {\mathbb{I}})({\mathbb{I}}_{({Q_J}-Q_j)}\otimes \tilde{O}_j) (\ket{0}^{Q_J} \otimes {\mathbb{I}}) 
        \right)\right \| \\
        & = \underbrace{\sum_{j=0}^{m-1} \left |\beta \frac{\tilde{y}_j}{\tilde{\beta}} - y_j\right |}_{\epsilon_y} \alpha_j + \sum_{j=0}^{m-1}y_j \epsilon_j \\
        & \leq \epsilon_y \max_j \alpha_j + \beta \max_j \epsilon_j.
    \end{align}
    Here $\tilde{\beta}$ is an imperfect normalization factor, and the ${\mathbb{I}}_{(Q_J-Q_j)}$ term is to ``pad'' the operators requiring less than $Q_J$ qubits to block-encode, therefore this ensures consistent dimensionality.  
\end{proof}

Note that when the block-encoding parameters are chosen to be identical, we recover the result of Lemma 52 up to a factor of $\alpha$ in the second term. We believe this to be an erratum in Ref \cite{gilyen_quantum_2019}, given that the error as stated is a sum of terms with different units. The way that the block-encoding error is defined therein is to have units of $\alpha$ so a term in an error bound derived from a sum of such operators cannot have a term that is a product $\alpha \epsilon$ unless the block-encoding error is defined in terms of the normalized operator $ A/\alpha$, which it is not as stated. 

In addition, also defined in Ref.~\cite{gilyen_quantum_2019} is the notion of a state-prep-pair, which loosely speaking is a pair of unitaries that prepare a normalized state with arbitrary coefficients. We leverage this construction throughout the analysis where we define errors in terms of the coefficients as done in the lemma above. For example, for a state that encodes reciprocal mass coefficients, we define $\epsilon_m =\sum_{j=1}^\eta \left |\lambda_T \frac{\tilde{m}_j^{-1}}{\tilde{\lambda}_T} - m_j\right |$, and this error is then passed into and propagated through our prepare routines in the form of discretization and rotation errors.   

\subsection{Block-encoding $V$}
We now start by showing how the circuit in \cref{fig:U_V} effectively block-encodes the saturated Coulomb interaction [Eq.~\eqref{eq:sat_inv}] corresponding to $1/2r$, alongside the incurred error in this operation coming from the alternating sign discretization. 

\begin{lemma}[Alternating sign implementation of controlled Coulomb interaction] \label{lemma:sign_trick}
    The construction in \cref{fig:U_V} corresponding to an alternating sign trick encoding of $\|\vec q_1-\vec q_2\|^{-1}$ effectively block-encodes the controlled saturated Coulomb interaction in Eq.~\eqref{eq:sat_inv} using $n_M$ ancillas with a 1-norm of
    \begin{equation}
        \alpha_M = \frac{1}{2\Delta},
    \end{equation}
    while using 
    \begin{equation}
        \mathcal{T}_M = 2n_M^2+8n_Mn_g+16n_M+6n_g^2+16n_g+8
    \end{equation}
    Toffolis,
    \begin{equation}
        \tilde n_M=n_g+4 + \max\{3n_g^2, \;4n_M+5n_g+6\}
    \end{equation}
    temporary carry ancillas, and incurring an error
    \begin{equation} \label{eq:M_error}
        \epsilon_M  \leq \frac{3}{2M\Delta}.
    \end{equation}
\end{lemma}
\begin{proof}
    The construction in \cref{fig:U_V} clearly corresponds to an \gls{lcu} which is a uniform superposition of the unitaries $U_m$ summing over all the values of $m=0,\cdots,M-1$, having that the $\mathtt{Had}^{\otimes n_M}$ operation done on the $\ket{m}$ register effectively implements the PREPARE circuit corresponding to a uniform superposition over $M$ unitaries, where each unitary then has an associated weight of $1/M$. Here we defined $M:=2^{n_M}$. The action from the arithmetic part of the circuit trivially encodes the result of the inequality test $m^2 \|\vec q_1-\vec q_2\|^2<M^2$ in the flag register $\ket{\cdot}_{\mathcal{F}}$. The action of the CZ operation on the less significant digit of $\ket{m}$ will thus effectively implement the $u_m(\|\vec q_1-\vec q_2\|^2)$ function in Eq.~\eqref{eq:u_m}, corresponding to the action of $U_m$ on the state $\ket{\vec q_1,\vec q_2}$. Note that the constant $2\Delta$ here appears as a global constant multiplying our block-encoding as we defined $\bra{0}U_{12}^{(V)}\ket{0}=2\Delta\cdot V_{12}$, from which we are effectively implementing
    \begin{equation}
        V_{12} \ketbra{\vec q_1,\vec q_2}{\vec q_1,\vec q_2} = \frac{1}{2M\Delta} \sum_{m=0}^{M-1} u_m(\|\vec q_1-\vec q_2\|^2) \ketbra{\vec q_1,\vec q_2}{\vec q_1,\vec q_2}.
    \end{equation}
We now show that the sum over $m$ of these $m$-dependent signs, namely the \gls{lcu} over $U_m$'s, recovers the saturated $\|\vec q_1-\vec q_2\|^{-1}$ up to a small error. For clarity we use the notation $q:= \|\vec q_1-\vec q_2\|$, with the associated distance $r=q\Delta$. First consider the case $r\geq \Delta$ (i.e. $q\geq 1$) corresponding to the unsaturated branch, where we effectively implement $1/2r$ up to a discretization error:
    \begin{align}
    \epsilon_{M}^{(q>0)} &= \left| \frac{1}{2r} -\frac{1}{2M\Delta} \sum_{m=0}^{M-1} u_m(q) \right| \\
         &= \frac{1}{2}\left| \frac{1}{q\Delta} -\frac{1}{M\Delta} \left(\sum_{m=0}^{\floor{\frac{M}{q}}}(1) + \sum_{m=\ceil{\frac{M}{q}}}^{M-1} (-1)^m\right)\right| \\
        &= \frac{1}{2}\left| \frac{1}{q\Delta} - \frac{1}{M\Delta}\left(\floor{\frac{M}{q}} + 1 + \sum_{m=\ceil{\frac{M}{q}}}^{M-1} (-1)^m\right) \right| \\
        &\leq \frac{1}{2M\Delta}\left| \frac{M}{q} - \floor{\frac{M}{q}} +2\right| \\
        &\leq \frac{3}{2M\Delta}.
    \end{align}
    Here we have used the fact that $m^2q^2<M^2$ is equivalent to $m<M/q$, noting that $q$ is a positive integer. For the saturated branch with $q=0$, for which our target is $1/\Delta$, we have an exact implementation:
    \begin{align}
        \epsilon_M^{(q=0)} &= \left| \frac{1}{2\Delta} -\frac{1}{2M\Delta} \sum_{m=0}^{M-1} u_m(0) \right| \\
         &= \frac{1}{2}\left| \frac{1}{\Delta} -\frac{1}{M\Delta} \sum_{m=0}^{M-1}(1)\right| \\
         &= \frac{1}{2}\left| \frac{1}{\Delta} - \frac{1}{\Delta} \right| \\
         &= 0.
    \end{align}
    Thus, in all cases we have shown that this \gls{lcu} recovers the target saturated function for $1/2r$ with the maximum error $\epsilon_M$ stated in this lemma. The associated 1-norm of this encoding corresponds to
    \begin{align}
        \alpha_M &= \frac{1}{2\Delta} \sum_{m=0}^{M-1} \frac{1}{M} \\
        &= \frac{1}{2\Delta}
    \end{align}
    Having shown the correct action of our operation, we now deduce the associated cost by following the implementation in \cref{fig:U_V}. At each step we calculate the qubit overhead, which corresponds to the cumulative sum of the required ancillas at previous steps plus the number of temporary carry ancillas at the given step in parenthesis.
    \begin{enumerate}
        \item One $n_M$-qubit Hadamard transform for no Toffoli cost. This step can be considered the PREPARE routine, while the remaining steps corresponds to SELECT.
        \item Three calls to the $\mathtt{AbsDiff}$ operation, each one acting on two $n_g$-qubit registers (\cref{lem:abs_diff}). This requires $3\times(2n_g)=6n_g$ Toffolis, $3$ ancillas, and $n_g+1$ temporary carry qubits, noting that the latter can be reused for each of the three applications. The qubit overhead is $\tilde n_M^{(2)}=3+(n_g+1)$.
        \item  One call to the $\mathtt{Sum\ of\ Squares}$ oracle (Lemma 8 of Ref.~\cite{su2021fault}) acting on three $n_g$-qubit registers, requiring a total of $3n_g^2-n_g-1$ Toffolis, $2n_g+2$ additional ancillas for encoding the result and $3n_g^2-n_g-1$ temporary carry qubits. The qubit overhead is $\tilde n_M^{(3)}=2n_g+5+(3n^2_g-n_g-1)$.
        \item One $\mathtt{Square}$ routine on the $n_M$-qubit register (Lemma 6 of Ref.~\cite{su2021fault}), requiring $n_M^2-2$ Toffolis, $2n_M$ ancillas for storing the result, and $2n_M+2$ temporary carry ancillas. The qubit overhead is $\tilde n_M^{(4)}=2n_g+2n_M+5+(2n_M+2)$.
        \item One $\mathtt{Mult}$ of two registers consisting of $2n_M$ and $2n_g+2$ qubits (\cref{lem:fast_mult}). Assuming $n_M>n_g+1$ (noting that in this work $n_g\leq 9$ and $n_M\geq 23$ in all systems), the Toffoli cost is $4n_M(n_g+1) + 4n_M+2(n_g+1)+3=4n_Mn_g + 8n_M+2n_g+5$, while requiring $2(n_g+n_M+1)$ ancillas for storing the result and $2n_M+2n_g+3$ temporary carry ancillas. The qubit overhead is $\tilde n_M^{(5)}=4n_g+2n_M+7+(2n_M+2n_g+3)$.
        \item One $\mathtt{Sub}$ routine by a classical constant $M^2=2^{2n_M}$ with Hamming weight one (in bit $2n_M$) from a $(2n_M+2n_g+2)$-qubit register. This can be done by combining the inequality test routine in Fig. 15 of Ref.~\cite{berry2019qubitization} with the adder subunits of Fig. 17b from Ref.~\cite{sanders2020compilation}, with an associated cost of $2n_g+1$ Toffolis, one additional ancilla for carrying the inequality test result, and $2n_g-1$ temporary carry qubits. Note that these carry qubits are kept alive for uncomputing this operation with measurement and phase for no Toffoli cost. The qubit overhead is $\tilde n_M^{(6)}=4n_g+2n_M+8+(2n_g-1)$.
        \item Application of a CCZ gate for one Toffoli cost, with the same qubit overhead as in step 6 (considering how carry qubits from subtraction are kept alive for uncomputation).
        \item Uncomputation of steps 1-7, which effectively doubles the Toffoli costs from steps 2-6 (noting that step 7 is uncomputed for no Toffoli cost). This requires $6n_g+3n_g^2-n_g-1+n^2_m-2+4n_Mn_g+8n_M+2n_g+5=n_M^2+4n_Mn_g+8n_M+3n_g^2+7n_g+3$ Toffolis. 
    \end{enumerate}
    After step 8, all ancillas used in steps 2-7 are returned to $0$. The number of temporary carry ancillas used overall thus becomes $\tilde n_M=\max_{i} \tilde n_M^{(i)}$, which will generally be given by either step 3 or 5:
    \begin{align}
        \tilde n_M &= \max\{3n^2_g+n_g+4;4n_M+6n_g+10\} \\
        &= n_g+4 + \max\{3n_g^2,4n_M+5n_g+6\},
    \end{align}
    while the total number of Toffolis is twice the cost reported in step 8, plus the costs from steps 6 and 7:
    \begin{align}
        \mathcal{T}_M &= 2\times(n_M^2+4n_Mn_g+8n_M+3n_g^2+7n_g+3) + 2n_g+1+1 \\
        &= 2n_M^2+8n_Mn_g+16n_M+6n_g^2+16n_g+8,
    \end{align}
    concluding this proof.
\end{proof}

Initially, the factor of $3/2$ in the bound on $\epsilon_M$ seems peculiar. The reason for this is due to the subtlety that the $\ket m$ state produced by applying $\mathtt{Had}^{\otimes n_M} \ket{\vec{0}}$ produces a superposition over stings $m\in\{0,M-1\}$, whereas the inequality test is defined in terms of $M$, not $M-1$. Therefore, when $m=0$ amplitude is always added, leading to over counting by a grid point for boundary cases as a trade-off for cheaper arithmetic. 

\begin{lemma}[Controlled block-encoding of potential energy]\label{lemma:V_be}
    Block-encoding of the controlled potential energy operator $V$ in Eq.~\eqref{eq:V} can be done using 
    \begin{equation}
        \mathcal{T}_V = 2\mathcal T_{\textrm{PREP}_V} + 4(\eta-1)(1+3n_g)-8 + \mathcal T_{M}
    \end{equation}
    Toffolis, with a 1-norm of
    \begin{equation}
        \alpha_V = \frac{\lambda_{V}}{2\Delta},
    \end{equation}
    using
    \begin{equation}
        n_V = n_M+2\ceil{\log_2\eta}
    \end{equation}
    block-encoding qubits, 
    \begin{equation}
        \tilde n_V= \tilde n_M + \tilde n_{\textup{PREP}_V}
    \end{equation}
    temporary carry ancillas, and incurring an error of
    \begin{equation}
        \epsilon_V \leq \frac{\epsilon_{\zeta}}{2\Delta} + \frac{3\lambda_{V}}{2^{n_M+1}\Delta}.
    \end{equation}
    Here we have defined $\mathcal{T}_{\textup{PREP}_V}\in\tilde{\mathcal{O}}(\sqrt{\eta})$ as the cost of loading the $\sum_{i\neq j}\sqrt{\zeta_i\zeta_j}\ket{i,j}$ state with 1-norm $\lambda_{V}:= \sum_{i\neq j} \zeta_i\zeta_j$, which corresponds to the cheapest routine between the amplitude-amplified approach (\cref{subsubapp:aa_prep}) and coherent alias sampling for symmetric matrices (\cref{subsubapp:sym_prep}), requiring an associated number of temporary qubits that need to be passed for uncomputation $\tilde n_{\textup{PREP}_V}$ and incurring a maximum error $\epsilon_{\zeta}$ per coefficient. $\mathcal{T}_M$ is the Toffoli count for the alternating sign block-encoding in \cref{lemma:sign_trick}. 
\end{lemma}
\begin{proof}
    The numbers from this lemma follow immediately from combining the two-dimensional swap network technique in \cref{lemma:swap_2} with the alternating sign block-encoding for the saturated Coulomb potential in \cref{lemma:sign_trick}, considering that in this case $K=\eta$, $n=3n_g$, $n_B=n_M$, $n_c=2\ceil{\log_2\eta}$. Note that the swap network with controlled phases is used here (\cref{fig:signed_swap}), which has the same Toffoli and qubit requirements as its regular version. Finally, note that the temporary carry ancillas required for the multiplexed swaps in the unary iteration can be generally taken from the temporary carry ancillas $\tilde n_M$, from which they do not have to be considered in this cost.
\end{proof}

We now present the costs for the optimized approach with variable saturation $\Gamma$. 
\begin{lemma}[Controlled block-encoding of potential energy with optimized saturation]\label{lemma:V_be_opt}
    The controlled block-encoding of the potential energy operator $V$ in Eq.~\eqref{eq:V} with the variable nuclear-nuclear saturation governed by $\Gamma:= 2^{n_\Gamma/2} \: | \: n_\Gamma \in \mathbb{N}$ can be done using 
    \begin{equation}
        \mathcal{T}_{V,\Gamma} = 2\mathcal T_{\textrm{PREP}_V} + 4(\eta-1)(1+3n_g)-8 + \mathcal T_{M} + 3
    \end{equation}
    Toffolis, with a 1-norm of
    \begin{equation}
        \alpha_{V,\Gamma} = \frac{\lambda_V(\Gamma)}{2\Delta},
    \end{equation}
    using
    \begin{equation}
        n_V = n_M+2\ceil{\log_2\eta}
    \end{equation}
    block-encoding qubits, 
    \begin{equation}
        \tilde n_{V,\Gamma}= \tilde n_M + \tilde n_{\textrm{PREP}_V}+3
    \end{equation}
    temporary carry ancillas, and incurring an error of
    \begin{equation}
        \epsilon_{V,\Gamma} \leq \frac{\epsilon_{\zeta}}{2\Delta} + \frac{3\lambda_{V}(\Gamma)}{2^{n_M+1}\Delta}.
    \end{equation}
    All definitions are the same as in \cref{lemma:V_be}, while we also used the optimized 1-norm $\lambda_V(\Gamma):= \sum_{i,j\in\mathcal N} \zeta_i\zeta_j/\Gamma + 2\eta_e Z_{tot}+\eta_e(\eta_e-1)$ for $\mathcal{N}$ the set of nuclear-nuclear interactions and $Z_{tot}=\sum_{i\in nuc} \zeta_i$.
\end{lemma}
\begin{proof}
    The proof of this lemma is analogous to that of \cref{lemma:V_be}, with the only difference that the cost of the alternating sign block-encoding needs to be modified. The cost of modifying this block-encoding corresponds to replacing the $\mathtt{Sub}$ routine for the inequality test by a controlled counterpart $\mathtt{cSub}$, which using \cref{lemma:hybrid_adder} will increase the cost by a total of three extra Toffolis and three extra temporary carry qubits: two Toffolis are coming from the quantum control changing the saturation constant with two associated quantum single-qubit registers, while one Toffoli and one extra qubit come from an additional temporary AND operation. For these associated costs, $\Gamma$ must be defined as the square root of a power of 2, with the reason being that in the inequality test, $\Gamma^2$ appears, which then has Hamming weight 1 (an important property in the efficiency of \cref{lemma:hybrid_adder}). Also, note that the $\mathcal{T}_{\textrm{PREP}_V}$ and $\tilde n_{\textrm{PREP}_V}$ costs used here will be higher by a couple of Toffolis/qubits when compared to those appearing in \cref{lemma:V_be}, given how the PREPARE routine now also needs to flag a register for nuclear-nuclear interactions. 
\end{proof}

Finally, we note that the potential operator can be block-encoded with a reduced 1-norm by shifting the spectrum. The idea here is that currently the spectral range of the operator is $(0, \|V\|)$, and this can be shifted to $(-\|V\|/2, \|V\|/2)$ to reduce $\alpha_V$ as realized in Ref.~\cite{loaiza2023reducing}, having that this procedure effectively adds a constant shift by an identity that for all purposes can be neglected. This can be done by testing a more sophisticated inequality. First, rewrite the prior inequality and shift by half the maximum interaction energy $1/4\Delta$ recalling the fact that the maximum energy comes from particles separated by 1 grid space $\Delta$ and the factor of 2 comes from double-counting interactions:
\begin{align}
    \frac{m}{M}\frac{1}{2\Delta}&<\frac{1}{2\Delta\|\vec q\|}-\frac{1}{4\Delta} \\
    \frac{m}{M}&<\frac{1}{\|\vec q\|}-\frac{1}{2} \\
    \|\vec q\|^2(4m^2 + M(4m+M)) &<4M^2.
\end{align}
We see from the second line that this inequality is violated if $\|\vec q\| \geq 2 \: |\: \forall \: m$ or for $\|\vec q\|^2\geq 4$. In the last line we rearrange the inequality into a form that is efficient to compute with arithmetic. This leads us to define the following operator:
\begin{align} \label{eq:V_ss}
    V^{\rm{ss}}_{12}
    &= \frac{1}{\Delta M} \sum_{m=0}^{M-1} \sum_{\vec q_1, \vec q_2} v_m\left(\|\vec q_1 - \vec q_2\|^2\right) \,\,\ketbra{\vec q_1}{\vec q_1} \otimes \ketbra{\vec q_2}{\vec q_2},
\end{align}
with 
\begin{equation} \label{eq:shifted_V}
    v_m(x)=
    \begin{cases}
        1, &\textup{if} \: \left(\frac{m}{M}<\frac{1}{x}-\frac{1}{2}\right ) \land (x\leq 2),\\
        -1, &\textup{if} \: \left(\frac{m}{M}>\frac{1}{x}-\frac{1}{2}\right ) \land (x>2),\\
        (-1)^m, &\textup{otherwise}.
    \end{cases}
\end{equation}
To generalize this to the case where the nuclear-nuclear interaction can be further saturated by $\Gamma$, this inequality undergoes the following modification:
\begin{align}
    \frac{m}{M}\frac{1}{2\Delta \Gamma} &<\frac{1}{2\Delta \|\vec q\|} - \frac{1}{4\Delta \Gamma} \\
    \|\vec q\|^2(4m^2+M(4m+M))&<4\Gamma^2M^2 \\
    \|\vec q\|^2(4m^2+2^{n_M}(4m+2^{n_M}))&<\Gamma^22^{2n_M+2}, \label{eq:shifted_ineq_test}
\end{align}
in which case we have the function 
\begin{equation} \label{eq:shifted_V_nuc}
    v_m^{(nuc)}(x)=
    \begin{cases}
        1, &\textup{if} \: \left(\frac{m}{M}<\frac{\Gamma}{x}-\frac{1}{2}\right ) \land (x\leq 2\Gamma),\\
        -1, &\textup{if} \: \left(\frac{m}{M}>\frac{\Gamma}{x}-\frac{1}{2}\right ) \land (x>2\Gamma),\\
        (-1)^m, &\textup{otherwise}.
    \end{cases}
\end{equation}
As with the previous construction, by testing this inequality the $V^{ss}$ operator can be efficiently block encoded, with the following lemma summarizing the costs, as shown in \cref{fig:shifted_V}.

\begin{figure}
    \centering
    \resizebox{1.0\textwidth}{!}{
 \begin{quantikz}
 	\lstick{$\ket{i,j\in\mathcal{N}}$} &&&&&&&&&&& \gate{\hspace{3.2cm}} \gateinput{b} \gateoutput{b} \vqw{9} & \gate[3]{\hspace{0.9cm} \mathtt{cSub}[2^{2(n_M+1+b\cdot n_\Gamma)}] \hspace{0.9cm}} \gateinput{b} \gateoutput{b} &&  \\
 	\lstick{$\ket{0}_{\mathcal{F}}$} &&&&  &&&&&&&& \gateinput{$0$} \gateoutput{$a\leq$} & \targ{}\vqw{8} & \rstick{$\ket{\mathcal{F}}$}   \\
 	\lstick{$\ket{0}$} & \qwbundle{2(n_g+n_M+3)} &  &&&&&&&& \gate{\hspace{1.3cm}} \gateinput{0} \gateoutput{ab} \vqw{2}  && \gateinput{$a$} && \\
 	\lstick{$\ket{0}_m$} & \qwbundle{n_M} & \gate{\mathtt{Had}^{\otimes n_M}} & \gate[2]{\hspace{0.3cm}\mathtt{Square}\hspace{0.3cm}} \gateinput{$a$} \gateoutput{$a$}  & \gate{\mathtt{Mult}[2^2]}  & \qwbundle{+2} & \gate{\mathtt{Add}[2^{n_M}]} & \qwbundle{+2}  & \gate[2]{\hspace{0.3cm}\mathtt{Add}[\times 2^{n_M}]\hspace{0.3cm}} \gateinput{a} \gateoutput{a} &&&&&&   \\
 	\lstick{$\ket{0}$} &\qwbundle{2n_M} && \gateinput{$0$} \gateoutput{$a^2$} & \gate{\mathtt{Mult}[2^2]} & \qwbundle{+2} &&& \gateinput{$b$} \gateoutput{$b+a\cdot2^{n_M}$} & \qwbundle{+2} & \gate[2]{\hspace{0.3cm}\mathtt{Mult}\hspace{0.3cm}} \gateinput{$a$} \gateoutput{$a$} &&&& \\
 	\lstick{$\ket{0}$} & \qwbundle{2n_g+2} && \gate[2]{\hspace{0.3cm} \mathtt{Sum \ of \ Squares} \hspace{0.3cm}} \gateinput{$0$} \gateoutput{$x^2+y^2+z^2$} & &&&&&& \gateinput{$b$} \gateoutput{$b$} & \gate{\hspace{0.5cm}\mathtt{cSub}[2^{2+b\cdot2n_\Gamma}]\hspace{0.5cm}} \gateinput{a}  &&& \\
 	\lstick{$\ket{\vec q_1}$} & \qwbundle{3n_g} & \gate[3]{\mathtt{AbsDiff}^{\otimes 3}} & \gateinput{$x,y,z$} \gateoutput{$x,y,z$} &&&&&&&&&&& \\
 	\lstick{$\ket{0}$} & \qwbundle{3} &&&&&&&&&&&&& \\
 	\lstick{$\ket{\vec q_2}$} & \qwbundle{3n_g} &&&&&&&&&&&&& \\
 	\lstick{$\ket{0}$} &&&&&&&&&&& \gate[style={inner ysep=5pt}]{\hspace{3.2cm}}  \gateinput{$0$} \gateoutput{$a\leq$} & \gate{Z} & \ctrl{0} & 
 \end{quantikz}}
    \caption{Implementation for spectrally shifted and variably saturated Coulomb interaction flagging routine (replacing flagging routine $\mathcal{U}_{\rm arithmetic}$ in \cref{fig:U_V}), as specified by Eqs.(\ref{eq:V_ss}-\ref{eq:shifted_V_nuc}) and summarized by \cref{lem:shifted_V}. Single-register $\mathtt{Mult}[2^n]$ routines simply append $n$ qubits to the end of the register for encoding multiplication by power of $2$. $\mathtt{Add}[\times 2^n]$ routine corresponds to adder where the ``$a$'' register is being multiplied by a power of $2$, which can be done with usual adder circuit with shifted $a$ register by $n$ bits, having a classical control for value of $0$ in associated $n$ bits. Note that Pauli Z gate in bottom register does not get applied during uncomputation, as it adds a global $-1$ sign for negative values below median appearing from spectral shifting.}
    \label{fig:shifted_V}
\end{figure}
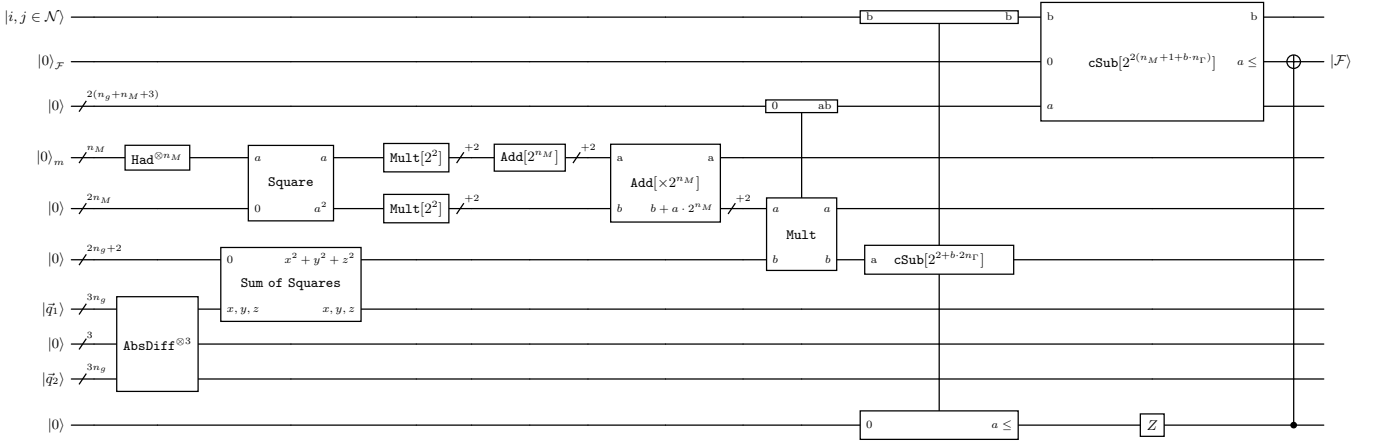

\begin{lemma}[Controlled block-encoding of spectral shifted $V$ operator with saturation] \label{lem:shifted_V} There exists a $(\alpha_V^{\rm ss}, n_V^{\rm ss}, \epsilon_{V,\Gamma})$ controlled block-encoding of $V^{\rm{ss}}_{12}$ defined in Eq.~\eqref{eq:shifted_V} that requires 
\begin{equation}
    \mathcal{T}_{V,\Gamma}^{\rm ss} = 6n_g^2+2n_M^2+8n_Mn_g+36n_g+20n_M+46+2\mathcal{T}_{\textrm{PREP}_V}+4(\eta-1)(1+3n_g)
\end{equation}
Toffoli gates to construct, with a 1-norm of 
\begin{equation}
    \alpha_V^{\rm ss} = \frac{\lambda_V(\Gamma)}{4\Delta}
\end{equation}
using 
\begin{equation}
    n_{V, \Gamma}^{\rm ss} = n_M+2\ceil{\log_2\eta}
\end{equation}
block-encoding qubits, 
\begin{equation}
    \tilde{n}_{V, \Gamma}^{\rm ss} = \max\{(6n_g + 6n_M+25),\; (3n_g^2+n_g+4) \} + \tilde n_{\textrm{PREP}_V}
\end{equation}
temporary carry ancillas, and incurring an error of
\end{lemma}
    \begin{equation}
        \epsilon_{V,\Gamma} \leq \frac{\epsilon_{\zeta}}{2\Delta} + \frac{\lambda_{V}(\Gamma)}{2^{n_M}\Delta}.
    \end{equation}
\begin{proof}
    In perfect analogy to \cref{lemma:V_be}, we here use a swap network over an alternating sign implementation for the Coulomb interaction, with the difference that this implementation also includes the discussed spectral shift, alongside the saturation $\Gamma$. We now review the cost of the shifted and saturated alternating sign implementation for the Coulomb potential, with the remainder of the costs following exactly the same line of reasoning as in \cref{lemma:V_be}.
    We need to use quantum arithmetic to compute the inequality test of Eq.~\eqref{eq:shifted_ineq_test}, which can be done with the following steps that correspond to circuit \ref{fig:shifted_V}. We will again calculate the qubit overhead by noting the maximum ancilla at any point in time using the notation for temporary carry ancillas $(\cdot )$:
    \begin{enumerate}
        \item Three calls to the $\mathtt{AbsDiff}$ operation, each one acting on two $n_g$-qubit registers (\cref{lem:abs_diff}). This requires $3\times(2n_g)=6n_g$ Toffolis, $3$ ancillas, and $n_g+1$ temporary carry qubits, noting that the latter can be reused for each of the three applications. The qubit overhead is $\tilde n_M^{(1)}=3+(n_g+1)$.
        \item  One call to the $\mathtt{Sum\ of\ Squares}$ oracle (Lemma 8 of Ref.~\cite{su2021fault}) acting on three $n_g$-qubit registers, requiring a total of $3n_g^2-n_g-1$ Toffolis, $2n_g+2$ additional ancillas for encoding the result and $3n_g^2-n_g-1$ temporary carry qubits. The qubit overhead is $\tilde n_M^{(2)}=2n_g+5+(3n^2_g-n_g-1)$.
        \item One $\mathtt{Square}$ routine on the $n_M$-qubit register (Lemma 6 of Ref.~\cite{su2021fault}), requiring $n_M^2-2$ Toffolis, $2n_M$ ancillas for storing the result, and $2n_M+2$ temporary carry ancillas. The qubit overhead is $\tilde n_M^{(3)}=2n_g+2n_M+5+(2n_M+2)$.
        \item Multiply the registers $\ket m$ and $\ket{m^2}$ each by 4 with a bit shift. This has no Toffoli cost, but requires 2 qubits each. This has no Toffoli cost, and the qubit overhead is $\tilde{n}_M^{(4)} = 2n_g+2n_M+9$
        \item Add $M=2^{n_M}$ to the $\ket{4m}$ state. Since this is addition by a power of 2, we only need an adder on the top 2 bits with a carry-out. This has a cost of 2 Toffoli. The qubit overhead is $\tilde{n}_M^{(5)} = 2n_g+2n_M+10 + (2)$.
        \item Multiply the register storing $\ket{4m +M}$ by $M$, which can be done with bit shifts. This shift can be done be done by reinterpreting the indices of the bits. Since $\ket{M(4m+M)}$ will be used in an adder in the next step, and since the bottom $n_M$ bits will be deterministically 0 after the shift, classical controls can be used on those bits using a hybrid quantum-classical adder as in \cref{fig:hybrid_adder}. In this step, prior to using this adder, there is no cost.
        \item Add $\ket{M(4m+M)}$ to $\ket{4m^2}$ in place with the hybrid quantum-classical adder. The resulting sum is $2n_M+4$ bit state, 2 more bits than needed to store $4m^2$ (so the end of said state is padded prior to addition), and we can remove the Toffoli on the lowest bit with a classical control, giving a cost of $2n_M+3$ Toffoli. The qubit overhead is $\tilde{n}_M^{(7)} =  2n_g+2n_M+12 + (2n_M+3)$.
        \item Multiply the two registers $\ket{\|q\|^2}$ and $\ket{4m^2 + M(4m+M)}$ using \cref{lem:fast_mult}. This gives a Toffoli cost of $4n_Mn_g+10n_g+8n_M+21$, with a qubit overhead of $\tilde{n}_M^{(8)} = 4n_g+4n_M+18+(2n_g+2n_M+7)$.
        \item Perform an inequality test to check if $\|q\|^2 \leq 4\Gamma^2$ using the subtraction from Lemma \ref{lemma:hybrid_adder}, where we subtract by either $4\Gamma^2$ or $4$ controlled on the register flagging nuclear-nuclear interactions ($\Gamma=1$ when the interaction is not nuclear-nuclear). This has cost $2n_g+1$. The qubit overhead is $\tilde{n}_M^{(9)} = 4n_g+4n_M+22  + (2n_g)$.
        \item Perform the overall inequality test using $\mathtt{cSub}$ controlled on whether the reaction is between two nuclei or not. This is performed on the state $\ket{\|q\|^2(4m^2+M(m+M))}$, which is stored in $2n_M+2n_g+6$ qubits. As stated previously, since we are subtracting powers of 2, the bottom $2n_M$ bits do not change. Performing the subtraction on the $2n_g+6$ high bits then gives a cost of $2n_g+5$. The qubit overhead is $\tilde{n}_M^{(10)} = 4n_g+4n_M+23 + (2n_g+4)$. 
        \item One Toffoli is needed for the controlled $Z$.
        \item Uncompute prior arithmetic other than the step 10, which can be uncomputed for free by holding on to dirty ancillae. 
    \end{enumerate}
    
    In terms of the error, the sign-trick inequality test is begin performed with the same precision as the non-shifted version of the block encoding, so the error follows from the prior lemmas. 
\end{proof}

\subsection{Block-encoding $T$}
We now show how block-encoding of the kinetic energy $T$ can be done simply by combining the one-dimensional swap network technique in \cref{lemma:swap_1} with the squaring block-encoding from \cref{lemma:walk_square}. Note that the kinetic energy is block-encoded using the swap network recursively: first it is done over Cartesian coordinates, effectively block-encoding $T_1$ [Eq.~\eqref{eq:T1}], which is then used inside another swap network running over all $\eta$ particles. 
\begin{lemma}[Controlled block-encoding of kinetic energy operator]\label{lemma:T_be}
Controlled block-encoding of the kinetic energy operator $T$ in Eq.~\eqref{eq:V} can be done using 
    \begin{equation}
        \mathcal{T}_T =2n_g^2 + 14n_g-3+2\mathcal{T}_R(\epsilon_W)+2\mathcal{T}_m(\epsilon_m) +[ 2(\eta-1)(1+3n_g)-4]
    \end{equation}
    Toffolis, with a 1-norm of
    \begin{equation}
        \alpha_T = \frac{3\pi^2 2^{2(n_g-1)}}{L^2}\lambda_T,
    \end{equation}
    using
    \begin{equation}
        n_T = \ceil{\log_2\eta} + n_g+4
    \end{equation}
    ancilla qubits, 
    \begin{equation}
        \tilde n_T= \ceil{\log_2\eta} +n_g-1+ \tilde{n}_{mass}(\epsilon_m)
    \end{equation}
    temporary carry ancillas, and incurring an error of
    \begin{equation}
        \epsilon_T \leq \alpha_T\left(\epsilon_W+\frac{\epsilon_m}{\lambda_T}\right).
    \end{equation}
    Here we have defined the mass 1-norm $\lambda_{T}:= \sum_{i} m_i^{-1}$, alongside the number of temporary qubits for the state associated PREPARE routine $\tilde n_{mass}(\epsilon_m)$, incurring an error $\epsilon_{m}$ and using $\mathcal{T}_m(\epsilon_m)$ Toffolis. $\epsilon_W$ corresponds to the error in the W state preparation in \cref{fig:w_prep}.
\end{lemma}
\begin{proof}
    We first start by noting that the full block-encoding of $T$ is done by using a swap network over the different mass coefficients for different particles $m_i^{-1}$ for the single-particle kernel $T_1$. This $T_1$ operator in turn corresponds to a swap network block-encoding of the squaring operator in \cref{lemma:walk_square} swapped over $3$ Cartesian coordinates, with the associated 1-norm
    \begin{align}
        \alpha_{T_1} &= \frac{2\pi^2}{L^2}\cdot 3\cdot\alpha_{\rm SQ}(n_g) \\
        &= \frac{3\pi^2 2^{2(n_g-1)}}{L^2},
    \end{align}
    where the $2\pi^2/L^2$ comes from the physical momentum units (see Eq.~\eqref{eq:T}), the factor of $3$ corresponds to the 1-norm for the linear combination over Cartesian coordinates, and $\alpha_{\rm SQ}(n_g)=(2^{2(n_g-1)})/2$ is the 1-norm from the squaring block-encoding, noting that there is a factor of $1/2$ appearing from the fact that the squaring routine encodes $2q^2-1$. The 1-norm of the full kinetic energy thus becomes
    \begin{align}
        \alpha_T &= \alpha_{T_1} \cdot \sum_i \frac{1}{m_i} \\
        &= \frac{3\pi^2 2^{2(n_g-1)}}{L^2}\lambda_T.
    \end{align}
    Noting that no error is incurred by the squaring routine, the error for $T_1$ can be upper bounded using \cref{lemma:swap_1} as
    \begin{align}
        \epsilon_{T_1} &\leq \alpha_{T_1}\cdot\epsilon_W \\
        &= \frac{3\pi^2 2^{2(n_g-1)}}{L^2} \epsilon_W,
    \end{align}
    where $\epsilon_W$ is the error associated with the $W$ state preparation (\cref{lemma:prep_w}). Considering an error of $\epsilon_m$ for the mass coefficients and using \cref{lemma:swap_1} again, we recover the total error for this routine
    \begin{align}
        \epsilon_T &\leq \lambda_T \epsilon_{T_1} + \alpha_{T_1}\epsilon_m \\
        &= \frac{3\pi^2 2^{2(n_g-1)}}{L^2} \lambda_T\epsilon_W + \frac{3\pi^2 2^{2(n_g-1)}}{L^2}\epsilon_m \\
        &= \alpha_T\left(\epsilon_W+\frac{\epsilon_m}{\lambda_T}\right).
    \end{align}
    The number of ancilla qubits for this block-encoding corresponds to the sum of those used for the mass coefficients ($\ceil{\log_2\eta}$), those used for the Cartesian coordinates ($3$), and those used by the squaring routine ($n_g+1$), yielding a total of
    \begin{equation}
        n_T = \ceil{\log_2\eta} + n_g+4
    \end{equation}
    block-encoding ancillas. For the number of temporary carry ancillas, we need to sum those required by the squaring block-encoding ($n_g-1$), those used by the PREPARE routine for the mass coefficients ($\tilde n_{mass}$), and those coming from the unary iteration for the swap network ($\ceil{\log_2 \eta}$), recovering
    \begin{equation}
        \tilde n_{T} = \ceil{\log_2\eta}+n_g-1+\tilde n_{mass}.
    \end{equation}
    Finally, we have that the Toffoli cost of this routine can be obtained by adding the cost of the squaring block-encoding ($10n_g-6$), two times cost of the W state preparation ($2\times(\mathcal{T}_R(\epsilon_W)+1)$), two times the cost of the mass state preparation ($2\mathcal{T}_m$), two times the cost of the SWAPs over Cartesian coordinates ($2\times 2n_g$), two times the cost of performing a \gls{qft} over $n_g$ qubits ($2\times n_g(n_g+1)$), and two times the cost of the multiplexed SWAPs over particle registers ($2(\eta-1)(1+3n_g)-4$ from \cref{lemma:swap_1}) obtaining a total cost of
    \begin{equation}
        \mathcal T_T = 2n_g^2 + 16n_g-4+2\mathcal{T}_R(\epsilon_W)+2\mathcal{T}_m + [2(\eta-1)(1+3n_g)-4],
    \end{equation}
    where we have left the swap network cost inside of brackets for clarity in our following deductions.
\end{proof}

\subsection{Block-encoding $H$} \label{subapp:be_H}
Having shown how our routines correctly block-encode the kinetic and potential energy operators, we are now ready to provide the proof of \cref{theo:ham,theo:time_evolution} summarizing our costs, having that the block-encoding of the full Hamiltonian is done as a linear combination of block-encodings for $T$ and $V$. 

\theobe*
\begin{proof}
    This theorem follows from building the block-encodings of $T$ and $V$ with their most optimized versions (\cref{lemma:T_be,lem:shifted_V}), alongside \cref{lem:be_sum} for implementing the sum of both block-encodings as described by \cref{eq:prep_ham,eq:ham_sel}. Note that as shown in \cref{fig:full_ham}, the swap networks for both block-encodings have been combined, which effectively removes the Toffoli cost from performing the multiplexed SWAPs in the kinetic energy routine. Overall, this yields the 1-norm
    \begin{align}
        \alpha_H &= \alpha_V+\alpha_T \\
        &= \frac{\lambda_{V}(\Gamma)}{4\Delta} + \frac{3\pi^2 2^{2(n_g-1)}}{L^2}\lambda_T \\
                &\in \mathcal{O}\left(\frac{\eta^2}{\Delta}+\frac{\eta}{\Delta^2}\right),
    \end{align}
    where we have used $\lambda_{V}(\Gamma) \in \mathcal{O}(\eta^2)$, $\lambda_T\in\mathcal{O}(\eta)$, and $L\in\mathcal{O}(\Delta\cdot2^{n_g})$ (see Eq.~\eqref{eq:grid_delta}). The total number of Toffolis is
    \begin{align}
        \mathcal T_{H} &= (2\mathcal{T}_{\textup{PREP}_V} + 4(\eta-1)(1+3n_g) + 2n^2_m+8n_Mn_g+20n_M+6n^2_g+36n_g+46)+(2n_g^2+14n_g-3+2\mathcal{T}_R(\epsilon_W)+2\mathcal{T}_m) \\
        &= 12\eta n_g + 4\eta + 2n_M^2 + 8n_g^2 + 8n_Mn_g + 40n_g + 20n_M+39+2(\mathcal{T}_{\textup{PREP}_V}+\mathcal{T}_R(\epsilon_W)+\mathcal T_m) \\
        &\in \mathcal{O}(\eta n_g + n_M^2+n_g^2),
    \end{align}
       noting that $\mathcal{T}_{\textup{PREP}_V}\in\tilde{\mathcal{O}}(\sqrt{\eta})$, while requiring
    \begin{align}
        n_H &=1+\max\{n_V,n_T\} \\
        &= n_M+2\ceil{\log_2\eta}+3
    \end{align}
    ancillas for the block-encoding. The number of temporary block-encoding ancillas corresponds to the maximum of required by $T$ and $V$, which corresponds to
    \begin{align}
        \tilde n_H &= \tilde n_V+n_R \\
        &= n_g+4+\max\{3n^2_g;4n_M+5n_g+6\}+\tilde n_{\textup{PREP}_V}.
    \end{align}
    Finally, the error of this routine will be given by (using \cref{lem:be_sum}):
    \begin{align}
        \epsilon_H &\leq \alpha_H\epsilon_r + \frac{\alpha_V}{\alpha_H}\epsilon_V + \frac{\alpha_T}{\alpha_H}\epsilon_T \\
        &= \alpha_H\epsilon_r +\frac{\alpha_V}{\alpha_H}\left(\frac{\epsilon_{\zeta}}{2\Delta} + \lambda_{V}\epsilon_M\right) +\frac{\alpha_T}{\alpha_H}\left(\epsilon_W+\frac{\epsilon_m}{\lambda_T}\right).
    \end{align}
    We now deduce the associated scalings, where we use $\epsilon_V\in\mathcal{O}(\epsilon_H)$, from which we can write $\mathcal{O}(n_M)=\mathcal{O}(\log\eta/\epsilon_V)=\mathcal{O}(\log \eta/\epsilon_H)$. Considering a simulation box with a volume that grows linearly with the number of particles, we have that $\mathcal{O}(L^3)=\mathcal{O}(\eta)$. We also assume that the grid spacing does not scale with system size, namely $\mathcal{O}(\Delta)=\mathcal{O}(1)$. We thus arrive to $\mathcal{O}(n_g)=\mathcal{O}(\log \eta^{\frac{1}{3}})=\mathcal{O}(\log\eta)$, from which we get
    \begin{align}
        \mathcal{T}_H &\in \mathcal{O}\left(\eta\log\eta + \log^2\frac{\eta}{\epsilon_H}+\log^2\eta\right) \\
        &\subset \tilde{\mathcal{O}}\left(\eta+\log^2\frac{1}{\epsilon_H}\right),
    \end{align}
    for the Toffoli complexity, and
    \begin{align}
        n_H &\in \mathcal{O}\left(\log\frac{\eta}{\epsilon}\right)+\mathcal{O}\left(\log\eta\right)+\mathcal{O}(1) \\
        &= \mathcal{O}\left(\log\eta+\log\frac{1}{\epsilon}\right).
    \end{align}
\end{proof}

\section{Simulation complexity} \label{app:complexity}
In this section we calculate the overall Toffoli gates required for implementing an $\epsilon$-accurate time-evolution operator $e^{-iHt}$ using our algorithm. Thus far we have calculated the total block-encoding 1-norm $\alpha_H$ and accuracy $\epsilon_{H}$, which provides the means for calculating the necessary degree of polynomial approximation $d$ to implement via \gls{qsp} for a qubitized simulation of the time-evolution operator. Since the block-encoding will be iteratively called in the \gls{qsp} sequence (of which has a length dependent on $t$) the error in the block-encoding will naturally propagate in time. In order to make the guarantee a total simulation error $\epsilon$, we use  Lemma 20 from \cite{pocrnic2025constant}, which accounts for error propagation and simultaneously give the simulation complexity in terms of queries to the block-encoding to perform an $\epsilon$-accurate implementation of the time-evolution operator $e^{-iHt}$ via the qubitization algorithm using generalized quantum signal processing \cite{motlagh2024generalized}. Note that a similar result is also presented in Ref. \cite{jornada_comprehensive_2025}. The lemma states that given access to an $(\alpha_H, n_H, \epsilon_{H})$ block-encoding of the Hamiltonian $\tilde{O}_H$, then one can construct $f(\tilde{O}_H)$ which is a $(1, n_H+2, \epsilon)$ block-encoding of $e^{-iHt}$, where $f(x)$ is effectively a polynomial approximation of $\exp(-ix)$ such that $|f(x) - \exp(-ix)| \leq \epsilon_{\mathrm{exp}}$. In order for the entire simulation error to be bound by $\epsilon$, the following inequality must be satisfied
\begin{equation}
    t \epsilon_{H} + \epsilon_{\mathrm{exp}} \leq \epsilon.
\end{equation}
Given our prior analysis in the proof of \cref{theo:ham}, the total error must satisfy:
\begin{align}
     t \epsilon_{H} + \epsilon_{\mathrm{exp}} &= t\left (\alpha_H\epsilon_r + \frac{\alpha_V}{\alpha_H}\epsilon_V + \frac{\alpha_T}{\alpha_H}\epsilon_T \right ) + \epsilon_{\mathrm{exp}} \\
    &\leq t\left(\alpha_H\epsilon_r +\frac{\alpha_V}{\alpha_H}\left(\frac{\epsilon_{\zeta}}{2 \Delta} + \lambda_{V}\epsilon_M\right) +\frac{\alpha_T^2}{\alpha_H}\left(\epsilon_W+\frac{\epsilon_m}{\lambda_T}\right)\right) +\epsilon_{\rm exp} \leq \epsilon, \label{eq:tot_error}
\end{align}
where we have used $\Delta \geq L/2^{n_g}$ to arrive to the second line. As expected, for long time simulations the block-encoding of the Hamiltonian will need to be achieved with high precision. The most expensive part of doing so will likely be in the inequality test to calculate the electrostatic potential. As the desired accuracy of block-encoding is increased, the number of qubits to perform said inequality test will grow (albeit logarithmically), which then increases the Toffoli cost of arithmetic. Given that the algorithm complexity scales logarithmically with respect to all errors considered, finding the ``optimal'' means for distributing error will likely lead to at most a small constant-factor improvement over the naive choice of evenly distributing the error. Let us consider a set of positive numbers $f_r,f_{\rm exp},f_M,f_{\zeta},f_m$, and $f_W$ such that their sum is equal to $1$, associated with portion of the error is allowed to come from each component, as seen in \ref{tab:error_dist}. A uniform distribution of error would then correspond to choosing $f_r=f_{\rm exp}=f_M=f_{\zeta}=f_m=f_W=1/6$. Note how we have also deduced the scaling of each error with respect to the main parameters used in this work by considering all $f$'s to be $\mathcal{O}(1)$, while also considering that both the number of grid points and length per Cartesian direction grow as $L,2^{n_g}\in\mathcal{O}(\eta^{\frac{1}{3}})$, from which $\Delta\in\mathcal{O}(1)$.

\begin{table}
    \centering
    \setlength\extrarowheight{12pt}
    \begin{tabular}{|c|c|c|}
        \hline
        \textbf{Error} & \textbf{Expression} & \textbf{Scaling} \\
        \hline
         Rotation error in $\textup{PREP}_H$ & $\displaystyle \epsilon_r = \epsilon f_r\cdot\frac{1}{t\alpha_H}$ & $\displaystyle \mathcal{O}\left(\frac{\epsilon}{t}\cdot\frac{1}{\eta^2}\right)$\\
         \hline
         Polynomial approximation error & $\epsilon_{\mathrm{exp}} = \epsilon f_{\rm exp}$ & $\displaystyle \mathcal{O}\left(\epsilon\right)$ \\
         \hline
         Error in approximating the potential & $\displaystyle \epsilon_M = \epsilon f_M\cdot \frac{\alpha_H} {t\lambda_{V} \alpha_V}$ & $\displaystyle \mathcal{O}\left(\frac{\epsilon}{t}\cdot\frac{1}{\eta^2}\right)$ \\
         \hline
         Error in charge coefficients & $\displaystyle \epsilon_{\zeta} = \epsilon f_{\zeta} \cdot \frac{2\alpha_H \Delta}{t \alpha_V}$& $\displaystyle \mathcal{O}\left(\frac{\epsilon}{t}\right)$ \\
         \hline
         Error in mass coefficients & $\displaystyle \epsilon_m = \epsilon f_m \cdot \frac{\alpha_H\lambda_T}{t\alpha_T^2}$  & $\displaystyle \mathcal{O}\left(\frac{\epsilon}{t}\cdot\eta \right)$ \\ \hline
         Error in W state preparation & $\displaystyle \epsilon_W =\epsilon f_W \cdot \frac{\alpha_H}{t\alpha_T^2}$ & $\displaystyle \mathcal{O}\left(\frac{\epsilon}{t}\right)$ \\
         \hline
    \end{tabular}
    \caption{Subroutine errors in relation to total error $\epsilon$ where the grid spacing is taken to be constant for a fixed precision. Note that each contains a factor $f_\cdot$ which can be replaced by an arbitrary positive constant given that the constants sum to 1, with the uniform allocation corresponding to taking all of them to be $1/6$. As discussed in the main text, the logarithmic cost dependence with respect to errors makes the costs mostly insensitive to the error allocation.}
    \label{tab:error_dist}
\end{table}


Lemma 20 of Ref. \cite{pocrnic2025constant} leverages recent improvements from generalized quantum signal processing (GQSP)\cite{berry2024doubling, motlagh2024generalized} to derive a constant factor bound on the degree of this polynomial $D$. The bound from [Ref.~\cite{pocrnic2025constant}, Lemma 21] is stated below: 
\begin{equation} \label{eq:qubitization_queries}
    D(t)= \ceil{\frac{e}{2}\alpha t +  \log \left (\frac{2c}{\epsilon_{\mathrm{exp}}}\right ) }
\end{equation}
where $c = 4(\sqrt{2\pi}e^{1/13})^{-1} \approx 1.47762$. 

The art of GQSP then dictates that $D+2$ queries to the block-encoding of the Hamiltonian guarantee an $\epsilon$-accurate implementation of the time-evolution operator. To calculate the necessary Toffoli gates for our simulation, we can then simply multiply the number of Toffoli gates required to block-encode the Hamiltonian (calculated in prior sections) by $(D+2)$. The total Toffoli cost $\mathcal{T}[e^{-i H t}]$ can then be approximated as
\begin{equation} \label{eq:tot_cost}
    \mathcal{T}[e^{-i H t}] = D(t)\times \mathcal{T}_H.
\end{equation}
Note that the additive factor of 2 is a trivial addition to the complexity, and we omit it from our constant factor formulas for simplicity. In addition, for an $n$-qubit block-encoding the \gls{qsp} sequence requires an additional $D(n-1)$ zero controlled Toffoli gates (to control on the entire prep register being $\ket{0}$) and $D$ rotations. However, since these costs are additive to $D(t)\times \mathcal{T}_H$, which will be orders of magnitude larger, they will be a negligible contribution to the cost. Having presented these results, we are now ready to present the proof of \cref{theo:time_evolution}.

\theodyn*
\begin{proof}
    We start by combining the results from \cref{theo:ham}, alongside with Eq.~\eqref{eq:tot_cost}. The associated number of queries to the block-encoding oracle for implementing a time evolution with time $t$ is
    \begin{align}
        \mathcal{T}[e^{-iHt}] &\in \mathcal O(D(t)\cdot \mathcal T_H) \\
        &= \mathcal O\left(\left(\alpha_H t + \log \frac{1}{\epsilon}\right)\cdot\left(\eta n_g+n_M^2+n_g^2\right)\right) \\
        &= \mathcal O\left(\left[\left(\frac{\eta^2}{\Delta}+\frac{\eta}{\Delta^2}\right)t+\log\frac{1}{\epsilon}\right]\cdot\left[\eta n_g+n^2_m+n^2_g\right]\right).
    \end{align}
    Inserting the scalings $n_g\in\mathcal{O}(\log\eta)$ and $n_M\in\mathcal{O}(\log\eta/\epsilon_V)=\mathcal{O}(\log \eta t/\epsilon)$ into the Toffoli complexity above, we recover
    \begin{align}
        \mathcal T[e^{-iHt}] &\in \mathcal{O}\left(\left[\eta^2t+\log\frac{1}{\epsilon}\right]\cdot\left[\eta\log\eta+\log^2\frac{\eta t}{\epsilon}+\log^2\eta\right]\right) \\
        &= \mathcal{O}\left(\eta^3t\log\eta + \eta\log(\eta)\log\frac{1}{\epsilon}+\textup{polylog}\frac{\eta t}{\epsilon}\right) \\
        &\subset \tilde{\mathcal{O}}\left(\eta^3t + \eta\log\frac{1}{\epsilon}\right). 
    \end{align}
Noting that the error for the block-encoding of the Hamiltonian scales as $\epsilon_H\in\mathcal{O}(\epsilon/t)$, we recover the required number of ancillas $\displaystyle \mathcal{O}\left(\log\eta + \log\frac{t}{\epsilon}\right)$.
\end{proof}

\section{Grid discretization for faithful wavefunctions} \label{app:grid_spacing}
In this appendix we show an approach to determine the Cartesian grid spacing for having a faithful representation of the molecular wavefunction. We start from the well-known result when that $2$ to $4$ grid points per de Broglie wavelength are required for a faithful representation of a quantum system. This is a direct cause of the Nyquist-Shannon sampling theorem \cite{macridin2018digital,macridin2018electron}, having that this result has been obtained in different places throughout the quantum chemistry literature \cite{dvr_1,dvr_2}. The two different types of particles that we consider in this work are nuclei and electrons. We first start by considering what this criterion implies for nuclei, and then consider the electronic case. 

\subsection{Nuclear grid spacing}
We can estimate the nuclear de Broglie wavelengths by considering a system of non-interacting nuclei, where for a given temperature $T$ we can use the equipartition theorem for deducing the average kinetic energy of each nuclei
\begin{equation}
    \langle K \rangle_{nuc} = \frac{3}{2}k_B T.
\end{equation}
Using the formula for the de Broglie wavelength $\lambda_{dB}$ of a particle with momentum $p$ (using atomic units):
\begin{equation}
    \lambda_{dB} = \frac{2\pi}{\tilde p}
\end{equation}
we can then use the average momentum $\tilde p = \sqrt{2m\langle K \rangle_{nuc}}$, where $m$ is the particle's mass, to arrive to a formula for the grid spacing:
\begin{align}
    \Delta_r^{(nuc)} &= \frac{\tilde\lambda_{dB}^{(nuc)}}{2} \\
    &= \frac{\pi}{\sqrt{3m k_B T}}. \label{eq:grid_spacing}
\end{align}
Using room temperature $T=300\ \rm K$, a single proton would then require a grid spacing of $0.79$ \AA, while a heavier nuclei (e.g. chlorine) would need $0.12$ \AA, with more atoms and associated grid spacings being shown in \cref{fig:grid_spacing}. Note that the approach we took here for estimating this quantity will break down as $T\rightarrow 0$ since the classical assumptions we have considered are not valid in this regime. However, considering the spacing for room temperature should be a conservative bound for this case given the de Broglie wavelength will only increase as the temperature goes down.

\subsection{Electronic grid spacing}
We now estimate the maximum de Broglie wavelength that will appear for electronic wavefunctions. In general we have that the most localized electronic wavefunction will have the associated largest de Broglie wavelength. For a given atom, the maximally localized electronic wavefunction will be its associated $1s$ orbital, which should be occupied for most chemical reactions in non-extreme conditions. Considering an atom with nuclear charge $Z$, its associated $1s$ wavefunction corresponds to
\begin{equation}
	\psi_{1s}(\vec r) = \sqrt{\frac{Z^3}{\pi}} e^{-Zr}
\end{equation}
where here we have considered without loss of generality for the nuclei to be centered around zero. The associated expectation value of the kinetic energy for this wavefunction then corresponds to
\begin{align}
	\langle K \rangle_{1s} &= -\frac{Z^3}{2\pi} \int d^3r\ e^{-Zr} \nabla^2 e^{-Zr} \\
	&= -\frac{Z^3}{2\pi}\cdot 4\pi \int_0^\infty dr e^{-Zr} \frac{\partial}{\partial r}\left(r^2\frac{\partial}{\partial r} e^{-Zr}\right) \\
    &= -2Z^3 \int_0^\infty dr e^{-2Zr} \left(-2Zr + r^2 Z^2\right) \\
    &= -2Z^3 \left(Z^2 \frac{2!}{(2Z)^3} - 2Z \frac{1!}{(2Z)^2} \right) \\
    &= \frac{Z^2}{2},
\end{align}
where we have used the formula $\int_0^\infty dx\ x^n e^{-ax} = n!/a^{n+1}$. From this we can obtain an associated grid spacing based on the de Broglie wavelength as
\begin{align}
    \Delta_r^{(el)} &= \frac{2\pi}{2\sqrt{2 \langle K\rangle_{1s}}} \\
    &= \frac{\pi}{Z}.
\end{align}

\subsection{General grid spacing}
Since our algorithm uses the same grid for both nuclear and electronic wavefunctions, we have that the overall grid spacing $\Delta_{\rm min}$ will be determined by the minimum between the presented nuclear and electronic de Broglie-based criteria. Figure \ref{fig:grid_spacing} shows the associated grid spacing for different atoms. Once a required grid spacing is obtained, and considering some width $L$ for the simulation box, the associated number of qubits for discretizing the real-space grid is obtained from Eq.~\eqref{eq:grid_delta} as
\begin{equation}
    n_g = \ceil{\log_2\left(1+\frac{L}{\Delta_{\rm min}}\right)}, 
\end{equation}
from which an associated effective grid spacing $\Delta$ can be obtained using Eq.~\eqref{eq:grid_delta} again, having that by construction $\Delta \leq \Delta_{\rm min}$.

\begin{figure}
    \centering
    \includegraphics[width=0.75\linewidth]{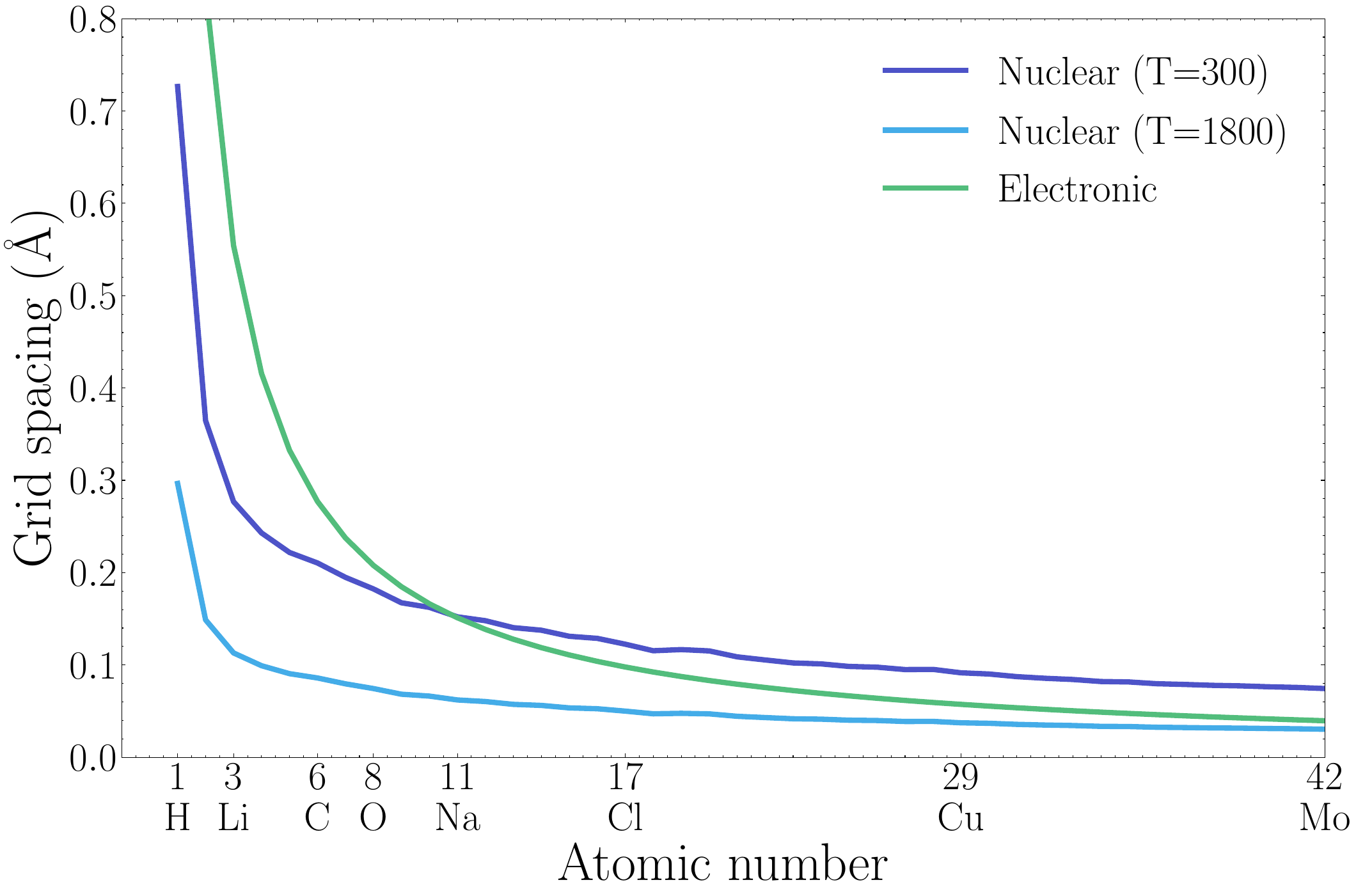}
    \caption{Grid spacings for different nuclei, obtained considering room temperature $T=300\ \rm K$ for \cref{eq:grid_spacing}.}
    \label{fig:grid_spacing}
\end{figure}

\section{Minimum nuclear-nuclear distance} \label{app:Gamma}
For finding a $\Gamma$ value for saturating the nuclear-nuclear interaction, we will consider in a general chemical reaction what is the minimum distance at which we can typically expect to find two nuclei. It is a well known that the molecule with shortest bond length corresponds to $\rm H_2$, with a bond distance of $d_{\rm H_2} \approx 0.74$ \AA. Using a harmonic approximation with the associated vibrational frequency of $\nu_{\rm H_2} \approx 4200\ \rm cm^{-1}$, we now consider the thermal distribution of associated vibrational states at a given temperature $\beta$, which is given by

\begin{equation}
	\rho_{\rm H_2} = \frac{e^{-\beta  H_{\rm H_2}}}{Z},
\end{equation}
where we have defined the vibrational harmonic Hamiltonian
\begin{equation}
	 H_{\rm H_2} := \hbar\omega_{\rm H_2} \left( n + \frac{1}{2}\right).
\end{equation}
Here $\omega_{\rm H_2}=2\pi c\nu_{\rm H_2}$ is the angular frequency associated to the vibrational frequency with $c$ the speed of light, $ n$ is the bosonic number operator, and we used the partition function $Z = \tr{e^{-\beta  H_{\rm H_2}}}$. The associated probability distribution of positions around the equilibrium value $x=0$ is known to be a Gaussian $\propto e^{-x^2/2\sigma_x^2}$ with width 
\begin{equation}
    \sigma_x(\beta) = \frac{\hbar}{2\pi \mu_{\rm H_2}\omega_{\rm H_2}} \coth \frac{\beta\hbar\omega_{\rm H_2}}{2},
\end{equation}
with an associated root-mean-squared displacement $x_{\rm RMS} = \sqrt{\sigma_x}$ \cite{landau2013statistical,pathria2021statistical}. Here we have used the reduced mass for the hydrogen molecule $\mu_{\rm H_2} = m_{\rm H}/2$, with $m_{\rm H}=1.008\rm \ amu$ the mass of a hydrogen atom. An associated minimum nuclear distance can then be obtained as 
\begin{equation}
    \Delta_{nuc}(\beta,N) = d_{\rm H_2} - N\cdot x_{\rm RMS}(\beta),
\end{equation}
where $N$ is the number of standard deviations considered, having that for $N=2$ we have $\sim95.45\%$ of nuclei having a larger separation, while this number becomes $\sim 99.73\%$ and $\sim 99.9937\%$ for $N=3$ and $N=4$ respectively. The associated nuclear distances are shown in \cref{fig:nuc_dist}. The nuclear saturation constant for a given grid spacing $\Delta$ is then obtained as
\begin{equation}
    \Gamma = \frac{\Delta_{nuc}}{\Delta}.
\end{equation}
In practice, we want to use the largest constant $\tilde\Gamma\leq\Gamma$ so that $\tilde\Gamma^2$ is a power of two, making $\tilde\Gamma^2M^2$ have Hamming weight 1 and compatible with the most efficient implementation of the hybrid quantum-classical subtraction in \cref{lemma:hybrid_adder}. This can be done by defining
\begin{align} 
    n_\Gamma &:= \floor{2\log_2\frac{\Delta_{nuc}}{\Delta}} \\
    \tilde\Gamma &= 2^{n_\Gamma/2},
\end{align}
which guarantees that the induced minimal nuclear distance is
\begin{equation} \label{eq:delta_tilde}
    \tilde\Delta_{nuc} = \Delta\tilde\Gamma \leq \Delta_{nuc}.
\end{equation}
Finally, we note that an analogous strategy could be used for obtaining minimal distances of any nuclear pair by using an associated experimental vibrational frequency for that bond. For vibrational frequencies which depend on what functional group the bond is a part of, the lowest frequency yielding an estimate of the smallest $\Delta_{nuc}$ should be used.

\begin{figure}
    \centering
    \includegraphics[width=0.75\linewidth]{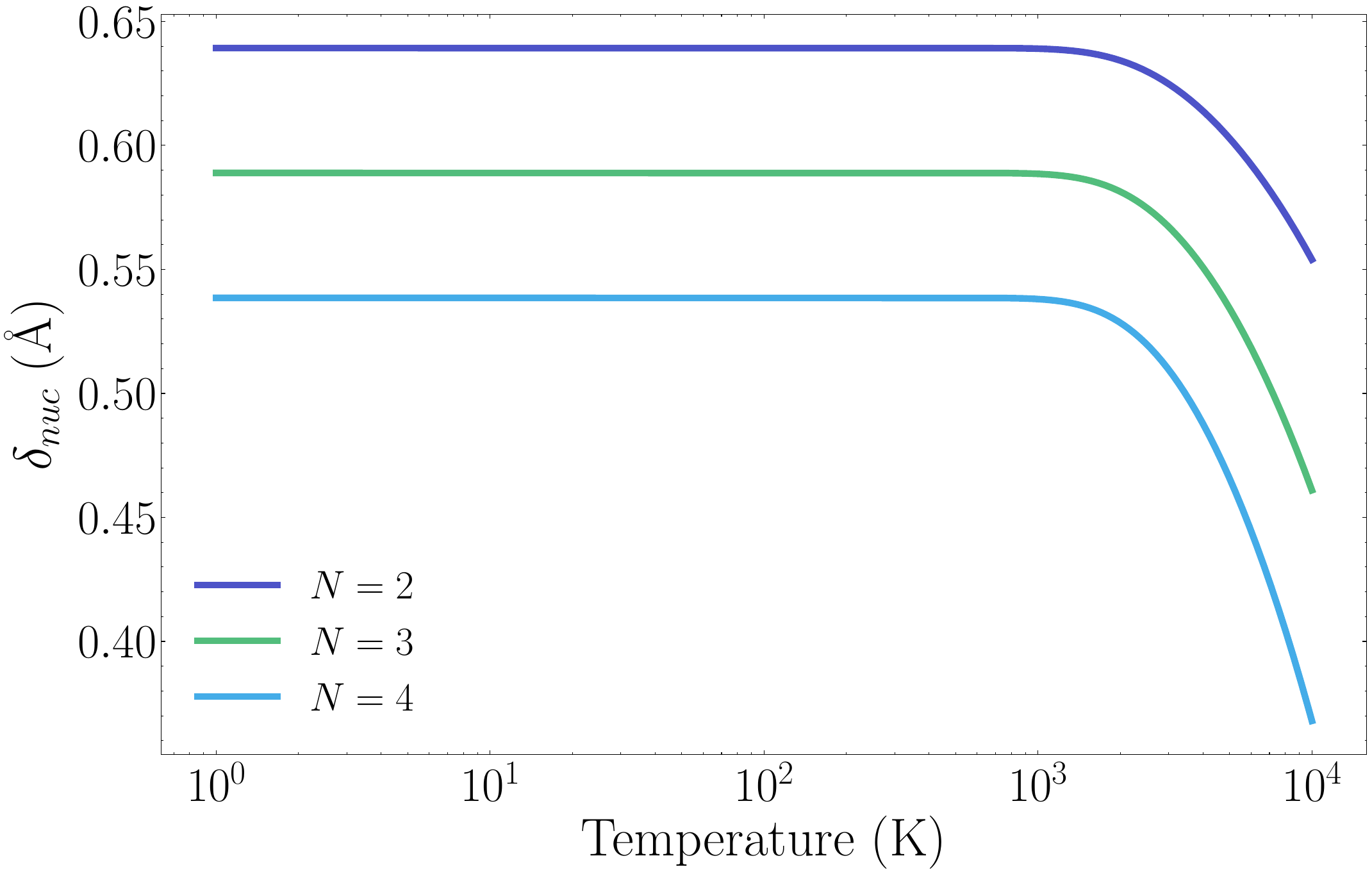}
    \caption{Minimum nuclear distance for $\rm H_2$ bond as a function of temperature. Different curves correspond to associated confidence interval coming from $N$ standard deviations.}
    \label{fig:nuc_dist}
\end{figure}

\section{Rigorous bounds on separations}
The above discussion provides useful heuristic estimates of the separation needed between particles but these estimates may be difficult to apply quantitatively to show that the separation is bounded below by a constant.  Here we provide explicit estimates based on assumptions about the moments of the energy distribution that allow us to bound the separation.

\begin{lemma}\label{lem:Econs}
    For any $p\ge 0$, Hermitian $H\in\mathbb{C}^{2^n\times 2^n}$ and $t\in \mathbb{R}$ the expectation value $\bra{\psi(t)} H^p \ket{\psi(t)} = \bra{\psi(0)} H^p \ket{\psi(0)}$ where $\ket{\psi(t)} = e^{-iHt}\ket{\psi(0)}$.
\end{lemma}
\begin{proof}
    The result immediately follows from the fact that $H$ commutes with $e^{-iHt}$.
\end{proof}

Next if we are promised an initial value of the mean energy and the variance of the energy then we can argue that large fluctuations of the energy will happen with vanishing probability and in turn that configurations that have unreasonably large potential energy must be improbable.  We will use this property to bound, with high probability, the inter electronic and inter nuclear separations.

\begin{theorem}
    Let $\ket{\psi(t)}$ be a bound state for the Coulomb Hamiltonian and assume that it has the property that for all time $\ket{\psi(t)}$ is confined within a box of side length $L$ in 
    position space in three spatial dimensions.  Further let us assume that the mean energy for the system is known to be $\bra{\psi(0)} H \ket{\psi(0)} = \bar{E}$.  Then let $\epsilon$ and $\delta$ be probabilities and let $\gamma$ be an error tolerance.  We then have that with probability at least $1-\epsilon$ for any time $t^*$ sampled uniformly from $[-\infty,\infty]$ the probability of drawing a momentum measurement greater than $\sqrt{\frac{3m_e(\eta_e + \eta_n) \bar{E}}{\delta \epsilon}}$ is at most $\delta$ and the error in approximating the wave function to be constant over a distance of $h$ is at most $\gamma$ if
    $$
    h\le \gamma \left(\sqrt{\frac{6\pi}{L}} \left(\frac{\epsilon \delta}{3m_e(\eta_e +\eta_n)\bar{E}} \right)^{3/4}\right)
    $$
\end{theorem}
\begin{proof}
The Virial theorem states that for the Coulomb potential $V(\vec{r},\vec{R})$
\begin{align}
    \mathbb{E}_t \bra{\psi(t)}  T\ket{\psi(t)} = -2\mathbb{E}_t\bra{\psi(t)} V\ket{\psi(t)}  
\end{align}
for any bound state $\ket{\psi(t)}$ (ignoring for notational simplicity that $\ket{\psi}$ is not necessarily an element of a Hilbert space).  From energy conservation, formalized by Lemma~\ref{lem:Econs}, we then have that
\begin{equation}
    \mathbb{E}_t  \bra{\psi(t)}  H\ket{\psi(t)} = \bar{E} = \mathbb{E}_t(\bra{\psi(t)} T \ket{\psi(t)} + 
    \mathbb{E}_t  \bra{\psi(t)}  V\ket{\psi(t)}= \frac{1}{2}  \mathbb{E}_t(\bra{\psi(t)} T \ket{\psi(t)})
\end{equation}



Thus from the Markov inequality the probability that the kinetic energy is greater than $T_{\max}$ is and any time randomly chosen from $-\infty$ to $\infty$
\begin{equation}
    {\rm Pr}(T \ge T_{\max}) \le \frac{\bar{E}}{2 T_{\max}}
\end{equation}
Since $T =\sum_i p_i^2/2m_i$ in momentum eigenbasis, it follows from the fact that the sum is positive that for any given momenta component $\ell$
\begin{equation}
    {\rm Pr}(p_\ell \ge \sqrt{2 m_eT_{\max}}) \le \frac{\bar{E}}{2 T_{\max}}
\end{equation}
where $m_{e}$ is the mass of an electron, which we assume is less than or equal to the masses of all other charged objects.  Thus from the union bound the probability that none of the momenta are greater than the above threshold is at most
\begin{equation}
    {\rm Pr}(\max_\ell |p_{\ell}| \ge \sqrt{2 m_eT_{\max}}) \le \frac{3(\eta_e + \eta_n)\bar{E}}{2 T_{\max}}
\end{equation}
Assuming that we are considering a simulation in $3$ spatial dimensions.

Thus if we have that $\Pi_{cut}$ is the momentum space cutoff projector then this is equivalent to saying that if $\ket{\psi(t)}$ is the continuum analogue of the state that for a randomly chosen time 
\begin{equation}
    \mathbb{E}_t\left(\bra{\psi(t)} \Pi_{\rm cut} \ket{\psi(t)} \mathbf{1}_{\max_\ell |p_{\ell}| \ge \sqrt{2 m_eT_{\max}}}  \right) \le \frac{3(\eta_e + \eta_n)\bar{E}}{2 T_{\max}}
\end{equation}
Thus we have that if $t^{*}$ is a specific time we then have from the Markov inequality that the probability that the cutoff in momentum of $\sqrt{2m_e T_{\max}}$ for all momenta in the system fails is bounded by $\delta$ if we increase the threshold by a factor of $1/\delta.$  Specfically,
\begin{equation}
    {\rm Pr}\left((\bra{\psi(t^*)} \Pi_{\rm cut} \ket{\psi(t^*)} \ge \epsilon \right) \le \delta
\end{equation}
if
\begin{equation}
    T_{\max} = \frac{3(\eta_e + \eta_n) \bar{E}}{2\delta \epsilon}.
\end{equation}
In turn, this implies that the maximum mentum cutoff for the same assumption is
\begin{equation}
    p_{\max} = \sqrt{2m_e T_{\max}} = \sqrt{\frac{3m_e(\eta_e + \eta_n) \bar{E}}{\delta \epsilon}}
\end{equation}

Now that the dynamics remains restricted to a box of size $L^3$.  We then can use the standard solution for the particle in a box to express the eigenfunctions for the kinetic operator to be $\sin(2\pi k_i x_i/L)$ for position coordinate $i$ of any of the particles where $k_i\in \mathbb{Z}_+$.
Thus our state can be written without loss of generality as
\begin{equation}
    \bra{x}\Pi_{\rm cut}\ket{\psi(t^*)}=\sum_{\vec{k}}  a_{k_i}(t^*)\left(\prod_{i} \sin(2\pi k_i x_i/L) \right)
\end{equation}
We then have that if $t^*$ is a time that obeys the cutoff
\begin{align}
    \left|\frac{\partial}{\partial_{x_q}}\bra{x}\Pi_{\rm cut}\ket{\psi(t^*)}\right| &=\left|\sum_{\vec{k}}  \frac{2\pi k_q a_{k_i}(t^*)}{L}\left(\prod_{{i\ne q}} \sin(2\pi k_i x_i/L) \right) \cos(2\pi k_q x_q/L)\right|\nonumber\\
    &\le \sqrt{\sum_{k_q}\frac{4\pi^2 k_q^2}{L^2}} \sqrt{1-\delta}\nonumber\\
    &\le \sqrt{\sum_{k_q}\frac{4\pi^2 k_q^2}{L^2}}\le \sqrt{\frac{4\pi^2k_{\max}^3}{3L^2}} 
\end{align}
from the Cauchy-Schwarz inequality.  Next using the momenta eigenvalues we have that
\begin{equation}
    k_q \le k_{\max}\le \frac{L \sqrt{2m_e T_{\max}}}{2 \pi} = \frac{L}{2\pi} \sqrt{\frac{3m_e(\eta_e + \eta_n) \bar{E}}{\delta \epsilon}}.
\end{equation}
Thus we have that under the assumptions that $t^*$ is in the set of times that obey the cutoff
\begin{equation}
    \left|\frac{\partial}{\partial_{x_q}}\bra{x}\Pi_{\rm cut}\ket{\psi(t^*)}\right| \le \sqrt{\frac{L}{6\pi}}\left( \frac{3m_e(\eta_e + \eta_n) \bar{E}}{\delta \epsilon}\right)^{3/4}
\end{equation}
The approximation error that arises from discretizing a function over a hypercube of side length $h$ is at most the maximum derivative multiplied by the length $h$.  Thus we have that the discretization error is under these assumptions
\begin{equation}
    h\sqrt{\frac{L}{6\pi}}\left( \frac{3m_e(\eta_e + \eta_n) \bar{E}}{\delta \epsilon}\right)^{3/4}:=\gamma
\end{equation}
Our claim then immediately follows.
\end{proof}
This provides a rigorous criteria about how the spatial discretization should be set given some moderately strong assumptions about the initial state and the dynamics.  As with all discretization bounds, it is often challenging to provide tight estimates that rigorously hold owing to the complicated nature of real space wave functions which will not necessarily be smoothly varying functions of position.  In this case, however, we are able to make stronger statements about the typical behaviour of the system owing to conserved quantities and the quantum Virial theorem.  Unfortunately, direct estimation of the spatial scale of $1/r$ is difficult in this problem because of the fact that there are both attractive and repulsive terms allowing for the possibility of particles being tightly concentrated if the attraction and repulsive effects cancel.  Specifically, if we have two pairs of distant charged particles, wherein one pair is a proton and an electron and the other is a pair of electrons, then we can move the pair of particles together at the same rate and keep energy constant.  This allows in theory an infinitesmal separation to be permitted on purely energetic grounds.  The use of the Virial theorem allows us to go beyond this reasoning of the worst possible configuration, but at the price of restricting our discussion to one of probability measures over randomly chosen configurations.  Nonetheless, while loose and lacking some of the physical insight of the aforementioned estimatates, these results do provide rigorous guarantees on the discretization scale needed for position with high probability.

\section{Hyperparameter selection} \label{app:hyperparams}
A minimum nuclear distance of $\Delta_{nuc}=0.59$ \AA\ was considered for the optimized algorithm, with the derivation from statistical arguments in \cref{app:Gamma} guaranteeing that shorter distances will have no significant contributions for the studied reactions. From the condition of the associated constant $\tilde\Gamma^2$ having Hamming weight $1$, this translated to an effective minimum distance of $\tilde\Delta_{nuc} \approx 0.49$ \AA\ being used for all reactions. All relevant hyperparameters are summarized in \cref{tab:hyperparams}, with the associated temperatures being used for the de Broglie-based grid spacing determination discussed in \cref{app:grid_spacing}. We also note that when performing our resource estimates, the Toffoli-based synthesis of single-qubit rotations (e.g. $R_z(\theta)$) was done by performing an addition of an associated classical constant which depends on $\theta$ on a phase gradient state \cite{sanders2020compilation}, having that achieving an accuracy $\epsilon_{rot}$ requires the resources as stated in \cref{lemma:toffoli_rot}. The total Toffoli counts reported in this work also include the cost of preparing the phase gradient state one time, while the qubit counts also take into account the required qubits for the most precise single qubit rotation required throughout the workflow.

\begin{table}[]
    \centering
    \begin{tabular}{|c|c|c|c|c|c|c|c|c|c|c|c|} \hline
        Reaction number & Box width $L$ (\AA) & Temperature ($^o$C) & $n_g$ & $n_{\Gamma}$& $n_M$ & $f_r^{-1}$ & $f_{\rm exp}^{-1}$ & $f_M^{-1}$ & $f_{\zeta}^{-1}$ & $f_m^{-1}$  \\ \hline 
        1 & 22 & 30 & 7 & 3 & 24 & 51.13 & 25.93 & 1.64 & 3.20 & 56.57  \\  \hline
        2 & 22 & 30 & 7 & 3 & 24 & 158.64 & 1.73 & 3.08 & 12.18 & 130.02  \\  \hline
        3 & 22 & 1500 & 9 & 7 & 23 & 89.89 & 42.15 & 1.19 & 10.31 & 40.70  \\  \hline
        4 & 22 & -90 & 7 & 3 & 23 & 28.39 & 257.08 & 1.09 & 25.87 & 144.65  \\  \hline
        5 & 44 & 30 & 8 & 3 & 30 & 83.45 & 12.95 & 1.38 & 5.92 & 61.28  \\ \hline
    \end{tabular}
    \caption{Summary of hyperparameters used in this work. For all systems considered here, an error of $\epsilon=10^{-2}$ was allowed, with associated subroutine errors being optimized to reduce overall Toffoli cost while complying with the total error [\cref{eq:tot_error}] and associated parameters $f_\cdot$'s. Note that as discussed in the main text, the total cost is extremely similar to what would be obtained with a uniform allocation of errors. To reduce the amount of variables in the optimization problem, we make the choice $\epsilon_W = \frac{\lambda_T-1}{\lambda_T^2}\epsilon_m$, such that $f_m$ determines each. We lift the constraints of the optimization problem by saturating the total error inequality of Eq.~\eqref{eq:tot_error}, setting the total error to $\epsilon$ and then solving for $\epsilon_{\rm exp}$ as a function of the other errors. In the optimization problem, we calculate $f_{\rm exp}(\epsilon, f_m, f_r, f_{\rm{exp}}, f_M, f_{\zeta })$, and optimize the other $f$ parameters, following the strategy of Ref.~\cite{pocrnic2025constant}. We use the gradient boosted trees algorithm from $\mathtt{SciKit}$ \cite{pedregosa2011scikit}, which has been successfully used to optimize the error in prior quantum algorithms \cite{pocrnic2024composite}.}
    \label{tab:hyperparams}
\end{table}

\end{document}